\documentclass{article}
\usepackage[round]{natbib}
\usepackage{amsmath}
\usepackage{amsthm}
\usepackage{amssymb} 
\usepackage{bm}
\usepackage{color}
\usepackage{semantic}
\usepackage{stmaryrd}
\usepackage[hidelinks=true]{hyperref}
\usepackage[T1]{fontenc}
\usepackage[varqu]{zi4}
\usepackage{libertine}

\newcommand{\itm}[1]{\ensuremath{\mathit{#1}}}
\newcommand{\sfm}[1]{\ensuremath{\mathsf{#1}}}
\newcommand\mhyphen{\mathop{\mbox{-}}}

\newcommand{\idealsize}{\sfm{ideal\mhyphen{}size}}

\title{Parameterized Cast Calculi and \\
       Reusable Meta-theory for \\
       Gradually Typed Lambda Calculi}
\author{Jeremy G. Siek\\
        Indiana University, USA\\
        $\mathtt{jsiek@indiana.edu}$
        \and        
        Tianyu Chen\\
        Indiana University, USA\\
        $\mathtt{chen512@indiana.edu}$}
\date{}

\newtheorem{theorem}{Theorem}
\newtheorem{lemma}[theorem]{Lemma}
\newtheorem{corollary}[theorem]{Corollary}
\newtheorem{proposition}[theorem]{Proposition}

\newtheorem{definition}[theorem]{Definition}

\newcommand{\CC}{\itm{CC}}
\newcommand{\SC}{\itm{SC}}
\newcommand{\key}[1]{\ensuremath{\mathtt{#1}}}
\newcommand{\agda}[1]{\ensuremath{\mathsf{#1}}}
\newcommand{\Types}{\mathcal{T}}
\newcommand{\Base}{b}
\newcommand{\Bases}{\mathcal{B}}
\newcommand{\Atomics}{\mathcal{A}}
\newcommand{\Nat}{\key{Nat}}
\newcommand{\Int}{\key{Int}}
\newcommand{\Unit}{\key{Unit}}
\newcommand{\Bool}{\key{Bool}}
\newcommand{\Zero}{\key{Z}}
\newcommand{\Suc}{\key{S}}
\newcommand{\Unk}{\ensuremath{\mathop{?}}}
\newcommand{\dotspace}{.\,}

\newcommand{\dblam}{\lambda \dotspace}
\newcommand{\app}{\;}

\newcommand{\tu}{{\to}}
\newcommand{\compile}[2]{\mathcal{C}\llbracket #2 \rrbracket}
\newcommand{\compilenew}[2]{\mathcal{C}'\llbracket #2 \rrbracket}
\newcommand{\coerce}[3]{\llparenthesis #1 \Rightarrow #2 \rrparenthesis^{#3}}

\newcommand{\sval}[1]{\mathsf{Simple}\,#1}
\newcommand{\val}[1]{\mathsf{Value}\,#1}
\newcommand{\var}[2]{\mathbb{V}\,#1\,#2}
\newcommand{\inert}[1]{\mathsf{Inert}\app #1}
\newcommand{\act}[1]{\mathsf{Active}\app #1}
\newcommand{\cross}[1]{\mathsf{Cross}\app #1}
\newcommand{\cast}[2]{#1 \langle #2 \rangle}
\newcommand{\castInert}[2]{#1 \langle\!\langle #2 \rangle\!\rangle}
\newcommand{\CastOp}[0]{\Rightarrow}
\newcommand{\Cast}[2]{#1 \CastOp #2}
\newcommand{\ICast}[2]{#1 \CastOp_i #2}
\newcommand{\GCast}[2]{#1 \CastOp_g #2}
\newcommand{\Proj}[2]{#1 \CastOp_p #2}
\newcommand{\Middle}[2]{#1 \CastOp_m #2}
\newcommand{\Inj}[2]{#1 \CastOp_i #2}
\newcommand{\Frame}[2]{#1 \rightarrowtail #2}

\newcommand{\blame}[2]{\key{blame}\,#1}
\newcommand{\consis}[2]{#1 \sim #2}
\newcommand{\leprecop}{\ensuremath{\sqsubseteq}}
\newcommand{\leprecgtlc}{\ensuremath{\sqsubseteq^G}}
\newcommand{\lepreccc}{\ensuremath{\sqsubseteq^C}}
\newcommand{\leprecii}[2]{\langle\!\langle #1 \rangle\!\rangle \sqsubseteq \langle\!\langle #2 \rangle\!\rangle}
\newcommand{\leprecit}[2]{\langle\!\langle #1 \rangle\!\rangle \sqsubseteq #2}
\newcommand{\leprecti}[2]{#1 \sqsubseteq \langle\!\langle #2 \rangle\!\rangle}
\newcommand{\leprec}[2]{#1 \leprecop #2}

\newcommand{\matchfun}[3]{#1 \triangleright #2 \tu #3}
\newcommand{\matchpair}[3]{#1 \triangleright #2 \times #3}

\newcommand{\cseq}[2]{#1 \mathop{;} #2}
\newcommand{\cfail}[1]{\bot^{#1}}

\newcommand{\ext}[1]{\mathsf{ext}\,#1}
\newcommand{\exts}[1]{\mathsf{exts}\,#1}
\newcommand{\rename}[2]{\mathsf{rename}\,#1\,#2}
\newcommand{\subst}[2]{\mathsf{subst}\,#1\,#2}
\newcommand{\substz}[1]{\mathsf{substZero}\,#1}
\newcommand{\dom}[1]{\agda{dom}\app #1}
\newcommand{\cod}[1]{\agda{cod}\app #1}
\newcommand{\rep}[1]{\llbracket #1 \rrbracket}
\mathchardef\mhyphen="2D
\newcommand{\sizeok}[3]{#1 \mid #2 \vdash #3 \; \mathrm{ok}}
\newcommand{\safefor}[2]{#1 \text{ safe for } #2}

\begin{document}
\maketitle

\begin{abstract}

  The research on gradual typing has led to many variations on the
  Gradually Typed Lambda Calculus (GTLC) of Siek and Taha (2006) and
  its underlying cast calculus. For example, Wadler and Findler (2009)
  added blame tracking, Siek et al. (2009) investigated alternate cast
  evaluation strategies, and Herman et al. (2010) replaced casts with
  coercions for space-efficiency. The meta-theory for the GTLC has
  also expanded beyond type safety to include blame safety
  (Tobin-Hochstadt and Felleisen 2006), space consumption (Herman et
  al. 2010), and the gradual guarantees (Siek et al. 2015). These
  results have been proven for some variations of the GTLC but not
  others.  Furthermore, researchers continue to develop variations on
  the GTLC but establishing all of the meta-theory for new variations
  is time consuming.
  
  This article identifies abstractions that capture similarities
  between many cast calculi in the form of two parameterized cast
  calculi, one for the purposes of language specification and the
  other to guide space-efficient implementations. The article then
  develops reusable meta-theory for these two calculi, proving type
  safety, blame safety, the gradual guarantees, and space consumption.
  Finally, the article instantiates this meta-theory for eight cast
  calculi including five from the literature and three new calculi.
  All of these definitions and theorems, including the two
  parameterized calculi, the reusable meta-theory, and the eight
  instantiations, are mechanized in Agda making extensive use of
  module parameters and dependent records to define the abstractions.

\end{abstract}

\pagebreak
\tableofcontents
\pagebreak

\section{Introduction}

The theory of gradual typing has grown at a fast pace since the idea
crystallized in the mid
2000's~\citep{Siek:2006bh,Tobin-Hochstadt:2006fk,Matthews:2007zr,Gronski:2006uq}.
Researchers have discovered many choices regarding the design and
formalization of gradually typed languages.
For example, a language designer can choose between runtime casts that
provide lazy, eager, or even partially-eager
semantics~\citep{Siek:2009rt,Garcia-Perez:2014aa}.  Alternatively, the
designer might apply the methodology of Abstracting Gradual Typing
(AGT) to derive the semantics~\citep{Garcia:2016aa}.
When a runtime casts fails, there is the question of who to blame,
using either the D or UD blame-tracking
approaches~\citep{Siek:2009rt}.  Furthermore, with the need to address
the problems of space efficiency \citep{Herman:2010aa}, one might
choose to use threesomes \citep{Siek:2010ya}, supercoercions
\citep{Garcia:2013fk}, or coercions in one of several normal forms
\citep{Siek:2012uq,Siek:2015ab}.

The last decade has also seen tremendous progress in the mechanization
of programming language theory \citep{Aydemir:2005fk}.  It has become
routine for researchers to use proof assistants such as
Coq~\citep{The-Coq-Development-Team:2004kf},
Isabelle~\citep{Nipkow:2002jl}, or Agda~\citep{Bove:2009aa} to verify
the proofs of the meta-theory for a programming language.
From the beginning, researchers in gradual typing used proof
assistants to verify type
safety~\citep{Siek:2006hh,Siek:2006fk,Tobin-Hochstadt:2008lr}. They
continue to mechanize the type soundness of new
designs~\citep{Siek:2013aa,Chung:2018aa} and to mechanize proofs of
new properties such as open world soundness and the gradual
guarantee~\citep{Siek:2015ac,Vitousek:2016aa,Xie:2018aa}.

While machine-checked proofs provide the ultimate degree of trust,
they come at a high development cost. With this in mind it would be
useful to reduce the cost by reusing definitions, theorems, and proofs
about gradually typed languages.
Agda provides particularly nice support for reuse: its combination of
parameterized modules, dependent types, and more specifically,
dependent records, provide a high degree of flexibility at a
relatively low cost in complexity. For this reason we choose to
develop a new mechanization in Agda instead of building on prior
mechanizations of gradual typing in other proof assistants.

\paragraph{The Parameterized Cast Calculus}

The first part of the article designs the Parameterized Cast Calculus.
It is parameterized with respect to the cast representation and
operations on casts. It can therefore model a range of cast calculi,
from ones that represent a cast with a source type, target type, and a
blame label such as the Blame Calculus~\citep{Wadler:2009qv}, to other
calculi that represent casts using the Coercion Calculus of
\citep{Henglein:1994nz}. The article proves the following theorems
about the Parameterized Cast Calculus, with mechanizations in Agda.
\begin{itemize}
\item type safety,
\item blame-subtyping, and 
\item the dynamic gradual guarantee.
\end{itemize}

This article instantiates the Parameterized Cast Calculus to produce
definitions and results for cast calculi in the literature as well as
new cast calculi. We produce the type systems, reduction relations,
and proofs of type safety and blame safety for all of the following
systems.  Our parameterized proof of the dynamic gradual guarantee was
completed only recently; we have instantiated it on two variants of
$\lambda\textup{\textsf{B}}$.
\begin{enumerate}
\item The cast calculus of \citet{Siek:2006bh}
   (Section~\ref{sec:EDA}).
\item A variant that classifies a function cast applied to a value
  as a value (Section~\ref{sec:EDI}).
\item The blame calculus $\lambda\textup{\textsf{B}}$ of \citet{Siek:2015ab}
  (Section~\ref{sec:lambda-B}).
\item A coercion-based version of Siek and Taha's cast calculus
  (Section~\ref{sec:EDC}).
\item A lazy D coercion-based calculus \citep{Siek:2009rt}
  (Section~\ref{sec:LazyCoercions}).
\item The $\lambda\textup{\textsf{C}}$ calculus of
  \citet{Siek:2015ab} (Section~\ref{sec:lambda-C}).
\end{enumerate}

\paragraph{Scope of the Parameterized Cast Calculus}

How many of the designs in the literature on gradually-typed languages
are instances of the Parameterized Cast Calculus? While we do not yet
have a complete answer to this question, we can make some informal
statements about our intent to complement the constructive examples
above. We conjecture that the Parameterized Cast Calculus, with some
minor modifications, could model the intrinsic calculi used in the AGT
methodology~\citep{Toro:2020aa}. The Parameterized Cast Calculus is
designed to model sound gradual languages, and not unsound ones such
as optional type systems \citep{Bracha:2004wa} that compile gradual
programs to dynamically typed programs via type
erasure~\citep{Bracha:1993sn,Chaudhuri:2014aa,Maidl:2014aa,Verlaguet:aa,Bierman:2014aa,Greenman:2018aa}. The
Parameterized Cast Calculus is only applicable to languages that use
casts for runtime enforcement, and not other mechanisms such as the
checks in transient
semantics~\citep{Vitousek:2016ab,Vitousek:2017aa,Greenman:2018aa,Greenman:2018ab,Vitousek:2019aa,Greenman:2020aa}. We
expect that several, different, parameterized cast calculi are needed
to capture the large design space present in the literature. We
encourage other researchers to join us in this project of developing
reusable meta-theory for gradually typed languages.

An important consideration regarding the design of an abstraction such
as the Parameterized Cast Calculus is where to draw the line between
the generic versus the specific. The Parameterized Cast Calculus makes
a natural choice regarding where to draw the line: at the
representation of casts and the operations on them.  However, there
are two competing principles that influence the design:
\begin{enumerate}
\item The more generic the design, the more instances can be
  accommodated, which increases reuse.
\item The more generic the design, the less meta-theory can be carried
  out generically, which reduces reuse.
\end{enumerate}
More experimentation with these design choices is needed to better
understand the tradeoffs. The Parameterized Cast Calculus of this
article is one data point.

\paragraph{Interactions with other language features}

The research on gradual typing has revealed subtle and complex
interactions between gradual typing and other language features such
as
\begin{itemize}
\item mutable state~\citep{Herman:2010aa,Siek:2015aa,Toro:2020aa},
\item subtyping~\citep{Siek:2007qy,Garcia:2016aa,Chung:2018aa,Banados-Schwerter:2021aa},
\item classes and
  objects~\citep{Takikawa:2012ly,Allende:2013ab,Vitek:2016aa,Takikawa:2016ab,Muehlboeck:2017aa,Chung:2018aa},
\item parametric polymorphism~\citep{Ahmed:2011fk,Ahmed:2017aa,Toro:2019aa,New:2019ab},
\item type inference~\citep{Siek:2008sf,Garcia:2015aa,Castagna:2019aa},
\item set-theoretic types~\citep{Castagna:2017aa},
\item dependent types~\citep{Eremondi:2019aa,Lennon-Bertrand:2020to},
\item and many more.
\end{itemize}
The Parameterized Cast Calculus of this article starts at the
beginning, focusing on a language with first-class functions and
constants~\citep{Siek:2006bh} with the addition of pairs and sums
whose presence helps to reveal some patterns that require
abstraction. We look forward to future work that extends the
Parameterized Cast Calculus to more language features.

\paragraph{Space Efficient Gradual Typing}

The second part of the article develops the Space Efficient
Parameterized Cast Calculus, which compresses sequences of casts to
obtain space efficiency. To support this compression, the cast
representation is required to provide a compose operator.  We
instantiate this calculus to produce results for a specific cast
calculi in the literature as well as a new cast calculi.  In
particular, we produce definitions and proofs of space efficiency in
Agda for
\begin{enumerate}
\item the $\lambda\textup{\textsf{S}}$ calculus of \citet{Siek:2015ab}
  (Section~\ref{sec:lambda-S}), and
\item a new cast calculus based on hypercoercions \citep{Lu:2019aa}
  (Section~\ref{sec:hypercoercions}).
\end{enumerate}


\paragraph{Correctness of Space Efficient Calculi}

Concurrently to the development of this article, and with an
overlapping set of authors, \citet{Lu:2019aa} investigate the
semantics of parameterized cast calculi using abstract machines. They
define two parameterized CEK machines~\citep{Felleisen:kx}, one to
define the semantics and the other to implement casts in a
space-efficient manner. The main result is a parameterized proof in
Agda of a weak bisimulation between the two machines, under the
assumptions that both cast representations satisfy certain laws
specific to the lazy D semantics~\citep{Lu:2020aa}. They instantiate
this generic theorem to prove the correctness of hypercoercions.
Generalizing this correctness result to include cast semantics other
than lazy D is future work.

\paragraph{Mechanization in Agda}

The language definitions and proofs are mechanized in Agda version
2.6.1.3 following the style in the textbook Programming Language
Foundations in Agda~\citep{Wadler:2019aa} and is publicly available in
the following github repository.
\begin{center}
  \url{https://github.com/jsiek/gradual-typing-in-agda}
\end{center}

\section{Gradually Typed Lambda Calculus}
\label{sec:gtlc}


In this section we formalize the static semantics of the Gradually
Typed Lambda Calculus (GTLC) of \citet{Siek:2006bh} with the addition
of products and sums. We present the GTLC as an intrinsically-typed
calculus with de Bruijn representation of variables, analogous to the
formalization of the Simply-Typed Lambda Calculus in the
\emph{DeBruijn} Chapter of PLFA.  As such, the terms of the GTLC are
derivations of well-typedness judgments. Readers familiar with the
GTLC may still find it beneficial to skim this section because this
style of formalization uses non-standard notation.

The dynamic semantics of the GTLC is defined by translation to a cast
calculus in Section~\ref{sec:gtlc-to-cc}.

\paragraph{Types}

Figure~\ref{fig:gradual-types} defines the set of types $\Types$ of
the GTLC, which includes function types, product types, sum types, and
atomic types.  The atomic types include the base types (natural
numbers, Booleans, etc.) and the unknown type, written $\Unk$ (aka. the
dynamic type). As usual, the unknown type represents the absence of
static type information. Figure~\ref{fig:gradual-types} defines typing
contexts $\Gamma$, which are sequences of types, that is, they map
variables (represented by de Bruijn indices) to types.

\begin{figure}[tbp]
\[
  \begin{array}{rcll}
    \Base \in \Bases & ::= & \Nat \mid \Bool \mid \Int \mid \Unit \mid \bot
                & \text{base types} \\
    a \in \Atomics & ::= & \Base \mid \Unk & \text{atomic types}\\
    A,B,C,D \in \Types & ::= & a \mid A \to B \mid A \times B
                            \mid A + B
                & \text{types}  \\
    \Gamma,\Delta & ::= & \emptyset \mid \Gamma \cdot A & \text{typing contexts}
  \end{array}
  \]

  \fbox{$\consis{A}{B}$}
  \begin{gather*}
    \mathtt{UnkR}{\sim}[A]:
    \inference{}
              {\consis{A}{\Unk}}
    \quad
    \mathtt{UnkL}{\sim}[B]: \inference{}{\consis{\Unk}{B}}
    \quad
    \mathtt{Base}{\sim}[\Base]: \inference{}{\consis{\Base}{\Base}}
    \\[2ex]
    \mathtt{Fun}{\sim}:
    \inference{\consis{A'}{A} & \consis{B}{B'}}
              {\consis{A \tu B}{A' \tu B'}}
    \\[2ex]
    \mathtt{Pair}{\sim}:
    \inference{\consis{A}{A'} & \consis{B}{B'}}
              {\consis{A \times B}{A' \times B'}}
    \quad
    \mathtt{Sum}{\sim}:
    \inference{\consis{A}{A'} & \consis{B}{B'}}
              {\consis{A + B}{A' + B'}}
  \end{gather*}

  \fbox{$\bigsqcup : \consis{A}{B} \to \Types$}
  \begin{align*}
    \bigsqcup(\mathtt{UnkR}{\sim}[A]) &= A \\
    \bigsqcup(\mathtt{UnkL}{\sim}[B]) &= B \\
    \bigsqcup(\mathtt{Base}{\sim}[\Base]) &= \Base \\
    \bigsqcup(\mathtt{Fun}{\sim}\app d_1\app d_2) &=
        \bigsqcup(d_1)\to\bigsqcup(d_2)\\
    \bigsqcup(\mathtt{Pair}{\sim}\app d_1\app d_2)&=
        \bigsqcup(d_1)\times\bigsqcup(d_2)\\
    \bigsqcup(\mathtt{Sum}{\sim} \app d_1\app d_2)&=
        \bigsqcup(d_1)+\bigsqcup(d_2)
  \end{align*}

  \fbox{$A \triangleright B$}
  \begin{gather*}
    \inference{}{A \to B \;\triangleright\; A \to B} \qquad
    \inference{}{\Unk \;\triangleright\; \Unk \to \Unk} \\[1ex]
    \inference{}{A \times B \;\triangleright\; A \times B} \qquad
    \inference{}{\Unk \;\triangleright\; \Unk \times \Unk} \\[1ex]
    \inference{}{A + B \;\triangleright\; A + B} \qquad
    \inference{}{\Unk \;\triangleright\; \Unk + \Unk}
  \end{gather*}
  
  \caption{Gradual Types, Typing Contexts, Consistency, Join, and Matching}
  \label{fig:gradual-types}
\end{figure}

%
Figure~\ref{fig:gradual-types} defines the consistency relation at the
heart of gradual typing. We say that two types are \emph{consistent},
written $\consis{A}{B}$, if they are equal except in spots where
either type contains the unknown type. For example,
\[
    \Nat \to \Bool \;\;\sim\;\; \Unk \to \Bool
\]
Because $\Unk \sim \Nat$ and $\Bool \sim \Bool$.  The rules for
consistency in Figure~\ref{fig:gradual-types} are labeled with their
proof constructors. For example, the following is a proof of the above
example.
\[
  \mathtt{Fun}{\sim}\app
    (\mathtt{UnkL}{\sim}[\Nat])\app
    (\mathtt{Base}{\sim}[\Bool])
\]
We use a colon for Agda's ``proves'' relation, which can also be read
as a ``has-type'' relation thanks to the Curry-Howard
correspondence~\citep{Howard:1980aa}.  In the rule
$\mathtt{Fun}{\sim}$ the $A$ and $A'$ flip in the premise, which is
unusual but doesn't matter; it just makes some of the Agda proofs
easier.
    
Figure~\ref{fig:gradual-types} also defines a join function
$\bigsqcup$ that computes the least upper bound of two types with
respect to the precision relation $\sqsubseteq$ (with $\Unk$ at the
bottom)~\citep{Siek:2006bh,Siek:2008sf}. This function is typically
defined on a pair of types and is a partial function where it is
defined if and only if the two types are consistent.  Here we instead
make $\bigsqcup$ a total function over proofs of consistency.

Finally, Figure~\ref{fig:gradual-types} defines the matching relation
that is used in the typing rules for eliminators (such as function
application and accessing elements of a pair). The matching relation
handles the situation where the function or pair expression has type
$\Unk$.

\paragraph{Variables}

The function $\mathbb{V}$, defined in Figure~\ref{fig:gtlc-terms},
maps a typing context $\Gamma$ and type $A$ to the set of all
variables that have type $A$ in context $\Gamma$.  As stated above,
variables are de Bruijn indices, that is, natural numbers where the
number $x$ refers to the $x$th enclosing lambda abstraction.  There
are two constructors for variables: \Zero{} (zero) and \Suc{} (plus
one).  The two rules in Figure~\ref{fig:gtlc-terms} correspond to the
signatures of these two constructors, where premises (above the line)
are the parameter types and the conclusion (below the line) is the
result type. The variable \Zero{} refers to the first lexical position
in the enclosing context, so \Zero{} takes no parameters, and its
result type is $\var{\Gamma}{A}$ if $\Gamma$ is a non-empty typing
context where type $A$ is at the front.  A variable of the form $\Suc
\app x$ refers to one position further out than that of $x$.  So the
constructor \Suc{} has one parameter, a variable in $\var{\Gamma}{A}$
for some $\Gamma$ and $A$, and its result type is a variable in
$\var{(\Gamma \cdot B)}{A}$, for any type $B$. An expression formed by
combinations of the constructors \Zero{} and \Suc{} is a proof of a
proposition of the form $\var{\Gamma}{A}$.  For example,
$S\app S \app Z$
is a proof of
$\var{(\emptyset \cdot \Bool \cdot \Nat \cdot \Int)}{\Bool}$
because $\Bool$ is at position $2$ in the typing context $\emptyset
\cdot \Bool \cdot \Nat \cdot \Int$.

\paragraph{Constants}

We represent the constants of GTLC by a particular set of Agda
values. First, we carve out the subset of GTLC types that a constant
is allowed to have, which are the base types and n-ary functions over
base types. We call them \emph{primitive types} and inductively define
the predicate $\mathbb{P}\,A$ in Figure~\ref{fig:gtlc-terms} to
identify them. We then define a mapping $\rep{\cdot}$ from
$\mathbb{P}\,A$ to Agda types (elements of $\agda{Set}$), also in
Figure~\ref{fig:gtlc-terms}.

\paragraph{Terms}

The terms of the GTLC are defined at the bottom of
Figure~\ref{fig:gtlc-terms}. The terms are intrinsically typed and
they are represented by derivations of the typing judgment.  The
judgment has the form $\Gamma \vdash_G A$, which says that type $A$
can be inhabited in context $\Gamma$. So terms, ranged over by $L,M,N$
are proofs of propositions of the form $\Gamma \vdash_G A$. A typing
judgment normally has the form $\Gamma \vdash_G M : A$, but here the
equivalent is written $M : \Gamma \vdash_G A$.

The constant $\$ k$ has type $P$ in context $\Gamma$ provided that the
Agda value $k$ has type $\rep{P}$ and $P$ proves that $A$ is a
primitive type.
A variable $` x$ has type $A$ in context $\Gamma$ if $x$ is a de
Bruijn index in $\var{\Gamma}{A}$ (explained above).

As usual, a lambda abstraction $(\lambda[A] \app M)$ has type $A \to B$
in context $\Gamma$ provided that $M$ has type $B$ in context $\Gamma
\cdot A$.
Lambda abstraction does not include a parameter name because we
represent variables as de Bruijn indices.
An application $(L \app M)_\ell$ has type $B$ if
$L$ has some type $A$ that matches a function type $A_1 \to A_2$,
$M$ has some type $B$, and $B$ is consistent with $A_1$.
The blame label $\ell$ is a unique identifier for this location in the
source code. Each blame label has a complement $\overline{\ell}$ and
the complement operation is involutive, that is,
$\overline{\overline{\ell}} = \ell$.

The term $\key{if}_\ell \app L \app M \app N$ requires that the type
of $L$ is consistent with $\Bool$, $M$ and $N$ have consistent types,
and the type of the $\key{if}$ as a whole is the join of the types of
$M$ and $N$.
The rules for pairs and projection are straightforward.  Regarding
sums, in $\key{case}_\ell \app L \app M \app N$, the type of $L$
matches a sum type $A_1 + A_2$. The terms $M$ and $N$ have function
type and one of them is called depending on whether $L$ evaluates to
$\key{inl}$ or $\key{inr}$. So the inputs $B_1$ and $C_1$ must be
consistent with $A_1$ and $A_2$, respectively, and the outputs $B_2$
and $C_2$ must be consistent. The type of the $\key{case}$ is the join
of $B_2$ and $C_2$.

\begin{figure}[tbp]
  \fbox{$\var{\Gamma}{A}$}\vspace{-10pt}
  \begin{gather*}
    \Zero : \inference{}{ \var{(\Gamma \cdot A)}{A} }
    \qquad
    \Suc : \inference{ \var{\Gamma}{A} }
                     { \var{(\Gamma \cdot B)}{A} }
  \end{gather*}  
  \fbox{$\mathbb{P}\,A$}\vspace{-10pt}
  \begin{gather*}
    \key{PBase}[\Base]:
    \inference{}
              {\mathbb{P}\,\Base}
    \qquad
    \key{PFun}[\Base] :          
    \inference{\mathbb{P}\,A}
              {\mathbb{P}\,(\Base \to A)}
  \end{gather*}
  \fbox{$\rep{-} : \Bases \to \agda{Set}$} \hfill
  \fbox{$\rep{-} : \mathbb{P}\,A \to \agda{Set}$}\\
  \begin{minipage}{0.45\textwidth}
  \begin{align*}
    \rep{\Bool} &= \mathbb{B} \\
    \rep{\Nat} &= \mathbb{N} \\
    \rep{\Int} &= \mathbb{Z} \\
    \rep{\Unit} &= \top \\
    \rep{\bot} &= \bot 
  \end{align*}
  \end{minipage}
  \begin{minipage}{0.45\textwidth}
  \begin{align*}
    \rep{\key{PBase}[\Base]} &= \rep{\Base} \\
    \rep{\key{PFun}[\Base]\app P} &= \rep{\Base} \to \rep{P} \\
  \end{align*}
  \end{minipage}

  \fbox{$\Gamma \vdash_G A$}
  \begin{gather*}
    \$ k : \inference{}
              {\Gamma \vdash_G A} ~ k : \rep{P}, P : \mathbb{P}\,A
    \qquad
    ` x : \inference{}
                    {\Gamma \vdash_G A} ~ x : \var{\Gamma}{A}
    \\[1ex]
    \lambda[A] : \inference{\Gamma \cdot A  \vdash_G B}
                 {\Gamma \vdash_G A \tu B}
    \qquad
    (- \app -)_\ell :
    \inference{\Gamma \vdash_G A & \Gamma \vdash_G B}
              {\Gamma \vdash_G A_2} ~ \matchfun{A}{A_1}{A_2}, A_1 \sim B
    \\[1ex]
    \key{if}_\ell : 
    \inference{\Gamma \vdash_G A & \Gamma \vdash_G B & \Gamma \vdash_G C}
              {\Gamma \vdash_G \bigsqcup(cn)}
              ~ A \sim \Bool,
              cn : B \sim C
    \\[1ex]
    \key{cons} :
    \inference{\Gamma \vdash_G A & \Gamma \vdash_G B}
              {\Gamma \vdash_G A \times B}
    \quad
    \pi_i^\ell :
    \inference{\Gamma \vdash_G A}
              {\Gamma \vdash_G A_i}\matchpair{A}{A_1}{A_2}
    \\[1ex]
    \key{inl}[B] :
    \inference{\Gamma \vdash_G A}
              {\Gamma \vdash_G A + B}
    \quad
    \key{inr}[A] :
    \inference{\Gamma \vdash_G B}
              {\Gamma \vdash_G A + B}
    \\[1ex]
    {\key{case}_\ell[B_1,C_1]}:
    \inference{\Gamma \vdash_G A & \Gamma \cdot B_1 \vdash_G B_2 & \Gamma \cdot C_1 \vdash_G C_2}
              {\Gamma \vdash_G \bigsqcup(bc)}
              \begin{array}{l}
                A_1 \sim (B_1 + C_1), \\
                bc : B_2 \sim C_2
              \end{array}
  \end{gather*}
  \caption{Term constructors of the Gradually Typed Lambda Calculus (GTLC)}
  \label{fig:gtlc-terms}
\end{figure}

\paragraph{Examples}

The following are a few example terms in the GTLC.
\[
\begin{array}{ccc}
  \key{cons}\app \$ 2 \app \$ 3 &:& \emptyset \vdash \Nat \times \Nat\\[1ex]
  ((\lambda[\Unk]\, ` \Zero) \app \$ 4)_{\ell_1} &:& \emptyset \vdash \Nat \\[1ex]
  \key{case}_{\ell_2}\app (\key{inr}[\Bool] \app \$\key{true})
  \app (\lambda[\Bool]\, ` \Zero)
  \app (\lambda[\Unk]\, (\$\neg \app ` \Zero)_{\ell_3}) &:& \emptyset \vdash \Bool
\end{array}
\]

\subsection{Static Gradual Guarantee}

The precision relation $\leprecop^\ast$ between typing contexts is straightforward.
We require that each variable has types that are related by
$\leprecop$ in the respective typing contexts.
We say term $M'$ in GTLC is \textit{more precise} than $M$ if the former
has more type annotations than the latter, written $M \leprecgtlc M'$.

For example,

\begin{equation*}\begin{split}
    \emptyset \cdot \Unk & \leprecop^\ast \emptyset \cdot \Int \\
    ((\lambda [\Unk] \, ` \Zero) \app \$ 42)_{\ell} & \leprecgtlc ((\lambda [\Nat] \, ` \Zero) \app \$ 42)_{\ell'}
\end{split}\end{equation*}

Figure~\ref{fig:gtlc-prec} presents the definitions of the precision
relations on GTLC types, terms, and typing contexts.

\begin{figure}[tp]
  \fbox{$\leprec{A}{A'}$}
  \begin{gather*}
    \inference{}{\leprec{\Unk}{A} }
    \quad
    \inference{}{\leprec{b}{b}}
    \quad
    \inference{\leprec{A}{A'} & \leprec{B}{B'}}{\leprec{A \to B}{A' \to B'}}
    \\[1ex]
    \inference{\leprec{A}{A'} & \leprec{B}{B'}}{\leprec{A \times B}{A' \times B'}}
    \quad
    \inference{\leprec{A}{A'} & \leprec{B}{B'}}{\leprec{A + B}{A' + B'}}    
  \end{gather*}
  
  \fbox{$M \leprecgtlc M'$}
  \begin{gather*}
    \inference{}{\$ k \leprecgtlc \$ k}
    \quad
    \inference{}{` x \leprecgtlc ` x}
    \\[1ex]
    \inference{A \leprecop A' & N \leprecgtlc N'}{\lambda [A] \, N \leprecgtlc \lambda [A'] \, N'}
    \quad
    \inference{L \leprecgtlc L' & M \leprecgtlc M'}{(L \app M)_{\ell} \leprecgtlc (L' \app M')_{\ell'}}
    \\[1ex]
    \inference{L \leprecgtlc L' & M \leprecgtlc M' & N \leprecgtlc N'}
              {\key{if}_{\ell} \, L \, M \, N \leprecgtlc \key{if}_{\ell'} \, L' \, M' \, N'}
    \\[1ex]
    \inference{M \leprecgtlc M' & N \leprecgtlc N'}{\key{cons} \, M \, N \leprecgtlc \key{cons} \, M' \, N'}
    \quad
    \inference{M \leprecgtlc M'}{\pi_1 \, M \leprecgtlc \pi_1 \, M'}
    \quad
    \inference{M \leprecgtlc M'}{\pi_2 \, M \leprecgtlc \pi_2 \, M'}
    \\[1ex]
    \inference{B \leprecop B' & M \leprecgtlc M'}{\key{inl} [B] \, M \leprecgtlc \key{inl} [B'] \, M'}
    \quad
    \inference{A \leprecop A' & M \leprecgtlc M'}{\key{inr} [A] \, M \leprecgtlc \key{inr} [A'] \, M'}
    \\[1ex]
    \inference{B_1 \leprecop B_1' & C_1 \leprecop C_1' & L \leprecgtlc L' & M \leprecgtlc M' & N \leprecgtlc N'}
              {\key{case}_{\ell}[B_1,C_1] \, L \, M \, N \leprecgtlc \key{case}_{\ell'}[B_1',C_1'] \, L' \, M' \, N'}
  \end{gather*}
  \fbox{$\Gamma \leprecop^\ast \Gamma'$}
  \begin{gather*}
    \inference{}{\emptyset \leprecop^\ast \emptyset}
    \qquad
    \inference{A \leprecop A' & \Gamma \leprecop^\ast \Gamma'}{\Gamma \cdot A \leprecop^\ast \Gamma' \cdot A'}
  \end{gather*}
  \caption{Precision on GTLC types, terms and typing contexts.}
  \label{fig:gtlc-prec}
\end{figure}

An important property of a gradual type system is that removing type
annotations from a well-typed term should yield a well-typed term, a
property known as the static gradual guarantee~\citep{Siek:2015ac}.
Here we extend this result for the GTLC to include products and sums.

\begin{lemma}
  \label{lem:ctx-prec-var}
  Suppose $\Gamma \leprecop^\ast \Gamma'$.
  If $x : \var{\Gamma'}{A'}$, there exists type $A$ such that $x : \var{\Gamma}{A}$ and $A \leprecop A'$.
\end{lemma}
\begin{proof}[Proof sketch]
  By induction on $x$.
\end{proof}

\begin{theorem}[Static Gradual Guarantee]\label{thm:static-gradual}
  Suppose $M'$ is well-typed $M' : \Gamma' \vdash_G A'$.
  If $\Gamma \leprecop^\ast \Gamma'$ and $M \leprecgtlc M'$,
  there exists type $A$ such that $M : \Gamma \vdash_G A$ and $A \leprecop A'$.
\end{theorem}
\begin{proof}[Proof sketch]
  By induction on the typing derivation $M' : \Gamma' \vdash_G A'$.
  We briefly describe the main idea of the interesting cases, since the full proof
  is mechanized in Agda.
  \begin{description}
  \item[Case $` x$] By Lemma~\ref{lem:ctx-prec-var} and rule $` x$ of $\Gamma \vdash_G A$.
  \item[Case $\lambda \, N'$] By induction hypothesis about $N'$ and the extended typing context.
  \item[Case $L' \app M'$] The induction hypothesis about $L'$ produces two sub-cases:
    The term $L$ in $L \app M \leprecgtlc L' \app M'$ can be either of $\Unk$ or of function
    type. If $L$ is of $\Unk$, the theorem is proved by type matching and $\Unk$
    being less precise than any type. On the other hand, if $L$ is of function type,
    it is proved by the fact that the codomain types satisfy the precision relation.
  \end{description}
  Other cases follow the same structure as function application:
  by the induction hypotheses about sub-terms and casing on whether
  the left side of a precision relation can be $\Unk$ whenever necessary.
\end{proof}

\section{Parameterized Cast Calculus}
\label{sec:param-cast-calculus}

The term constructors for the Parameterized Cast Calculus
$\CC(\CastOp)$ are defined in Figure~\ref{fig:param-cc-terms}.
Again the terms are intrinsically typed.  Like most cast calculi,
$\CC(\CastOp)$ extends the Simply-Typed Lambda Calculus with the
unknown type \Unk{} and explicit run-time casts.  Unlike other cast
calculi, the $\CC(\CastOp)$ calculus is parameterized over the
representation of casts, that is, the parameter $\CastOp$ is a
function that, given a source and target type, returns the
representation type for casts from $A$ to $B$. So $c : \Cast{A}{B}$
says that $c$ is a cast from $A$ to $B$.

The types and variables of the Parameterized Cast Calculus are the
same as those of the GTLC (Section~\ref{sec:gtlc}).
The intrinsically-typed terms of the Parameterized Cast Calculus are
defined in Figure~\ref{fig:param-cc-terms}.  Cast application is
written $M \langle c \rangle$ where the cast representation $c$ is not
concrete but is instead specified by the parameter $\CastOp$. In some
cast calculi, cast application carries a blame label. In others, such
as coercion-based calculi or threesomes, blame labels instead appear
inside the cast $c$. For purposes of uniformity, we place the
responsibility for blame tracking inside the cast representation, so
they do not appear directly in the cast application form of the
Parameterized Cast Calculus.
As usual there is an uncatchable exception $\blame{\ell}{}$.

\begin{figure}[tb]
  \fbox{$\Gamma \vdash A$}
  \begin{gather*}
    \$ k : \inference{}
              {\Gamma \vdash P} ~k : \rep{P}, P : \mathbb{P}\,A
    \qquad
    ` x : \inference{}
                    {\Gamma \vdash A}~x : \var{\Gamma}{A}
    \\[1ex]
    \lambda : \inference{\Gamma \cdot A  \vdash B}
                 {\Gamma \vdash A \tu B}
    \qquad
    (- \app -) :
    \inference{\Gamma \vdash A_1 \to A_2 & \Gamma \vdash A_1}
              {\Gamma \vdash A_2}
    \\[1ex]
    \key{if} : 
    \inference{\Gamma \vdash \Bool & \Gamma \vdash B & \Gamma \vdash B}
              {\Gamma \vdash B}
    \\[1ex]
    \key{cons} :
    \inference{\Gamma \vdash A & \Gamma \vdash B}
              {\Gamma \vdash A \times B}
    \quad
    \pi_i :
    \inference{\Gamma \vdash A_1 \times A_2}
              {\Gamma \vdash A_i}
    \\[1ex]
    \key{inl}[B] :
    \inference{\Gamma \vdash A}
              {\Gamma \vdash A + B}
    \quad
    \key{inr}[A] :
    \inference{\Gamma \vdash B}
              {\Gamma \vdash A + B}
    \\[1ex]
    \key{case}:
    \inference{\Gamma \vdash A_1 + A_2 & \Gamma \vdash A_1 \to B & \Gamma \vdash A_2 \to B}
              {\Gamma \vdash B}
    \\[1ex]
    \key{\cast{-}{c}}:
    \inference{\Gamma \vdash A}
              {\Gamma \vdash B}~c : \Cast{A}{B}
    \qquad
    \blame{\ell}{A}:
    \inference{}{\Gamma \vdash A}
  \end{gather*}
  \caption{Term constructors for the Parameterized
      Cast Calculus $\CC(\CastOp)$.}
  \label{fig:param-cc-terms}
\end{figure}

\subsection{The \agda{PreCastStruct} Structure}
\label{sec:precaststruct}

We introduce the first of several structures (as in algebraic
structures) that group together parameters needed to define the notion
of values, frames, and the reduction relation for the Parametric Cast
Calculus. The structures are represented in Agda as dependent records.

The \agda{PreCastStruct} structure includes the operations and
predicates that do not depend on the terms of the cast calculus, so
that this structure can be reused for different cast calculi, such as
the space-efficient cast calculus in
Section~\ref{sec:EfficientParamCasts}.  The \agda{CastStruct}
structure extends \agda{PreCastStruct} with the one addition operation
that depends on terms, which is \agda{applyCast}.

One of the main responsibilities of the \agda{PreCastStruct} structure
is categorizing casts as either \emph{active} or \emph{inert}.  An
active cast is one that needs to be reduced by invoking
\agda{applyCast} (see reduction rule \eqref{eq:cast} in
Figure~\ref{fig:param-cast-reduction}). An inert cast is one that does
not need to be reduced, which means that a value with an inert cast
around it forms a larger value (see the definition of \agda{Value} in
Figure~\ref{fig:cc-values-frames}). Different cast calculi make
different choices regarding which casts are active and inert, so the
\agda{PreCastStruct} structure includes the two predicates
\agda{Active} and \agda{Inert} to parameterize these differences. The
proof of Theorem~\ref{thm:cc-progress} (Progress) needs to know that
every cast can be categorized as either active or inert,
\agda{PreCastStruct} also includes the \agda{ActiveOrInert} field
which is a (total) function from casts to their categorization as
active or inert.

The reduction semantics must also identify casts whose source and
target type have the same head type constructor, such as a cast from
$A \times B$ to $C \times D$. We refer to such casts as \emph{cross}
casts and include the \agda{Cross} predicate in the
\agda{PreCastStruct} structure.  When a cross cast is inert, the
reduction semantics must decompose the cast when reducing the
elimination form for that type. For example, when accessing the
\key{fst} element of a pair $V$ that is wrapped in an inert cast $c$,
we decompose the cast using an operator, also called \agda{fst},
provided in the \agda{PreCastStruct} structure.  Here is the relevant
reduction rule, which is \eqref{eq:fst-cast} in
Figure~\ref{fig:param-cast-reduction}.
\[
  \key{fst}\app (\cast{V}{c})
  \longrightarrow
  \cast{(\key{fst}\app V)}{ \agda{fst}\app c \app x}
\]

Finally, the proof of Theorem~\ref{thm:cc-progress} (Progress) also
needs to know that when the target type of an inert cast is a
function, product, or sum type, that the cast is in fact a cross cast.
Thus, the \agda{PreCastStruct} includes the fields
$\agda{InsertCross}{\to}$, $\agda{InsertCross}{\times}$, and
$\agda{InsertCross}{+}$. The target type of an inert cast may not be a
base type (field \agda{baseNotInert}) to ensure that the canonical
forms at base type are just constants.

The fields of the \agda{PreCastStruct} record are:
\begin{description}
\item[$-\CastOp- : \Types \to \Types \to \agda{Set}$] \ \\
  Given the source and target type, this returns the Agda type for casts.

\item[$\agda{Inert} : \forall A B.\, \Cast{A}{B} \to \agda{Set}$]\ \\
  This predicate categorizes the inert casts, that is, casts that
  when combined with a value, form a value without requiring further
  reduction.

\item[$\agda{Active} : \forall A B.\, \Cast{A}{B} \to \agda{Set}$]\ \\
  This predicate categorizes the active casts, that is, casts that
  require a reduction rule to specify what happens when they are
  applied to a value.

\item[$\agda{ActiveOrInert} : \forall A B.\, (c : \Cast{A}{B}) \to \agda{Active}\app c \uplus \agda{Inert} \app c$] \ \\
  All casts must be active or inert, which is used in the proof of
  Progress.

\item[$\agda{Cross} : \forall A B.\, \Cast{A}{B} \to \agda{Set}$]\ \\
  This predicate categorizes the cross casts, that is, casts
  from one type constructor to the same type constructor, such as
  $A \to B \Rightarrow C \to D$. This categorization is needed
  to define other fields below, such as $\agda{dom}$.

\item[$\agda{InertCross}{\to} : \forall A B C.\, (c : \Cast{A}{B \tu C}) \to \agda{Inert}\app c \to \agda{Cross}\app c \times \Sigma A_1 A_2.\, A \equiv A_1 \tu A_2$]
  An inert cast whose target is a function type must be a cross cast.
  This field and the following two fields are used in the proof of
  Progress.

\item[$\agda{InertCross}{\times} : \forall A B C.\, (c : \Cast{A}{ B \times C}) \to \agda{Inert}\app c \to \agda{Cross}\app c \times \Sigma A_1 A_2.\, A \equiv A_1 \times A_2$]
  An inert cast whose target is a pair type must be a cross cast.

\item[$\agda{InertCross}{+} : \forall A B C.\, (c : \Cast{A}{ B + C}) \to \agda{Inert}\app c \to \agda{Cross}\app c \times \Sigma A_1 A_2.\, A \equiv A_1 + A_2$]
  An inert cast whose target is a sum type must be a cross cast.

\item[$\agda{baseNotInert} : \forall A \Base.\, (c : \Cast{A}{\Base}) \to \neg \agda{Inert}\app c$]\ \\
  A cast whose target is a base type must never be inert.
  This field is used in the proof of Progress. 
  
\item[$\agda{dom} : \forall A_1 A_2 B_1 B_2.\, (c : \Cast{(A_1
  \to A_2)}{(B_1 \to B_2)}) \to \agda{Cross}\app c \to \Cast{B_1}{A_1}$]\ \\
  Given a cross cast between function types, $\agda{dom}$ returns
  the part of the cast between their domain types. As usual, domains
  are treated contravariantly, so the result is a cast from $B_1$ to
  $A_1$.

\item[$\agda{cod} : \forall A_1 A_2 B_1 B_2.\,
  (c : \Cast{(A_1 \to A_2)}{(B_1 \to B_2)})
  \to \agda{Cross}\app c \to \Cast{A_2}{B_2}$]\ \\
  Given a cross cast between function types, $\agda{cod}$ returns
  the part of the cast between the codomain types.

\item[$\agda{fst} : \forall A_1 A_2 B_1 B_2.\, (c : \Cast{(A_1\times A_2)}{(B_1 \times B_2)}) \to \agda{Cross}\app c \to \Cast{A_1}{B_1}$]\ \\
  Given a cross cast between pair types, $\agda{fst}$ returns
  the part of the cast between the first components of the pair.

\item[$\agda{snd} : \forall A_1 A_2 B_1 B_2.\, (c : \Cast{(A_1\times A_2)}{(B_1 \times B_2)}) \to \agda{Cross}\app c \to \Cast{A_2}{B_2}$]\ \\
  Given a cross cast between pair types, $\agda{snd}$ returns
  the part of the cast between the second components of the pair.

\item[$\agda{inl} : \forall A_1 A_2 B_1 B_2.\, (c : \Cast{(A_1+ A_2)}{(B_1 + B_2)}) \to \agda{Cross}\app c \to \Cast{A_1}{B_1}$]\ \\
  Given a cross cast between sum types, $\agda{inl}$ returns
  the part of the cast for the first branch.

\item[$\agda{inr} : \forall A_1 A_2 B_1 B_2.\, (c : \Cast{(A_1 + A_2)}{(B_1 + B_2)}) \to \agda{Cross}\app c \to \Cast{A_2}{B_2}$]\ \\
  Given a cross cast between sum types, $\agda{inr}$ returns
  the part of the cast for the second branch.

\end{description}

\subsection{Values and Frames of $\CC(\CastOp)$}
\label{sec:cc-values-frames}

This section is parameterized by a \agda{PreCastStruct}. So, for
example, when we refer to $\CastOp$ and $\agda{Inert}$, we mean
those fields of the \agda{PreCastStruct}.

The values (non-reducible terms) of the Parameterized Cast Calculus
are defined in Figure~\ref{fig:cc-values-frames}. The judgment
$\val{M}$ says that the term $M$ is a value. Let $V$ and $W$ range
over values.  The only rule specific to gradual typing is
($\mathtt{Vcast}$) which states that a cast application $\cast{V}{c}$ is
a value if $c$ is an inert cast.

A value of type $\Unk$ must be a cast application where the cast is
inert and its target type is $\Unk$.

\begin{lemma}[Canonical Form for type $\Unk$]
  \label{lem:canonical-star}
  If $M : \Gamma \vdash \Unk$ and $\val{M}$,
  then $M \equiv \cast{M'}{c}$ where $M' : \Gamma \vdash A$,
  $c : \Cast{A}{\Unk}$, and $\inert{c}$.
\end{lemma}

In the reduction semantics we use \emph{frames} (single-term
evaluation contexts) to collapse the many congruence rules that one
usually finds in a structural operational semantics into a single
congruence rule. Unlike regular evaluation contexts, frames are not
recursive. Instead that recursion is in the congruence reduction rule.

The definition of frames for the Parameterized Cast
Calculus is given in Figure~\ref{fig:cc-values-frames}.  The
definition is typical for a call-by-value calculus. We also define the
$\itm{plug}$ function at the bottom of
Figure~\ref{fig:cc-values-frames}, which replaces the hole in a frame
with a term, producing a term.

\begin{figure}[tbp]
  \fbox{$\val{} : (\Gamma \vdash A) \to \agda{Set}$}
  \begin{gather*}
    \mathtt{V}\lambda : 
    \inference{}
              {\val{ (\lambda M)}}
    \qquad
    \mathtt{Vconst} :
    \inference{}{\val{ (\$k)}}
    \\[1ex]
    \mathtt{Vpair} :
    \inference{\val{M} & \val{N}}
              {\val{(\key{cons}\, M\, N)}}
    \\[1ex]
    \mathtt{Vinl} :
    \inference{\val{M}}
              {\val{(\key{inl}[B]\, M)}}
    \qquad
    \mathtt{Vinr} :
    \inference{\val{M}}
              {\val{(\key{inr}[A]\, M)}}
    \\[1ex]
    \mathtt{Vcast} :
    \inference{\val{M}}
              { \val{(\cast{M}{c})} } ~\agda{Inert}\,c
  \end{gather*}
  \fbox{$\Gamma \vdash \Frame{A}{B}$}
  \begin{gather*}
    (\Box \app -):
    \inference{\Gamma \vdash A}
              {\Gamma \vdash \Frame{(A \to B) }{ B}}
    \qquad
    (- \app \Box):
    \inference{M : \Gamma \vdash (A \to B)}
              {\Gamma \vdash \Frame{A }{ B}} ~\val{M}
    \\[1ex]
    \key{if}\,\Box\,-\,- :
    \inference{\Gamma \vdash A & \Gamma \vdash A}
              {\Gamma \vdash \Frame{\Bool }{ A}}
    \\[1ex]
    \key{cons}\,{-}\,\Box :
    \inference{M : \Gamma \vdash A}
              {\Gamma \vdash \Frame{B }{ A \times B}}~\val{M}
   \qquad
    \key{cons}\,\Box\,- :
    \inference{\Gamma \vdash B}
              {\Gamma \vdash \Frame{A }{ A \times B}}
    \\[1ex]
    \pi_i\,\Box :
    \inference{}
              {\Gamma \vdash \Frame{A_1 \times A_2}{A_i}}
    \\[1ex]
    \key{inl}[B]\,\Box :
    \inference{}
              {\Gamma \vdash \Frame{A}{A \times B}}
    \qquad
    \key{inr}[A]\,\Box :
    \inference{}
              {\Gamma \vdash \Frame{B }{ A \times B}}
    \\[1ex]
    \key{case}\,\Box\,-\,-:
    \inference{\Gamma \vdash A \to C & \Gamma \vdash B \to C}
              {\Gamma \vdash \Frame{A + B}{C}}
    \qquad
    \cast{\Box}{c}:
    \inference{}
              {\Gamma \vdash \Frame{A }{ B}} ~c : \Cast{A}{B}
  \end{gather*}
  \fbox{$\itm{plug} : \forall \Gamma A B.\, (\Gamma \vdash A) \to (\Gamma \vdash \Frame{A }{ B}) \to (\Gamma \vdash B)$}
  \begin{align*}
    \itm{plug}\app L \app (\Box \app M) &= (L \app M) \\
    \itm{plug}\app M \app (L \app \Box) &= (L \app M) \\
    \itm{plug}\app L \app (\key{if}\,\Box\, M\,N) &= \key{if}\,L\,M\,N\\
    \itm{plug}\app N (\key{cons}\,M\,\Box) &= \key{cons}\,M\,N\\
    \itm{plug}\app M (\key{cons}\,\Box\,N) &= \key{cons}\,M\,N\\
    \itm{plug}\app M (\pi_i\,\Box) &= \pi_i\,M \\
    \itm{plug}\app M (\key{inl}[B]\,\Box) &= \key{inl}[B]\,M \\
    \itm{plug}\app M (\key{inr}[A]\,\Box) &= \key{inr}[A]\,M \\
    \itm{plug}\app L (\key{case}\,\Box\,M\,N) &= \key{case}\,L\,M\,N \\
    \itm{plug}\app M (\cast{\Box}{c}) &= \cast{M}{c} 
  \end{align*}
\caption{Values and frames of $\CC(\CastOp)$.}
  \label{fig:cc-values-frames}
\end{figure}

The $\itm{plug}$ function is type preserving.  This is proved in
Agda by embedding the statement of this lemma into the type of
$\itm{plug}$ (see Figure~\ref{fig:cc-values-frames}) and then
relying on Agda to check that the definition of $\itm{plug}$
satisfies its declared type.

\begin{lemma}[Frame Filling]
  If $\Gamma \vdash M : A$ and $\vdash F : \Frame{A }{ B}$,
  then $\Gamma \vdash \itm{plug}\app M\app F : B$.
\end{lemma}

\subsection{The Eta Cast Reduction Rules}
\label{sec:eta-cast-reduction}

This section is parameterized by a \agda{PreCastStruct}.

Some cast calculi include reduction rules that resemble
$\eta$-reduction~\citep{Flanagan:2006mn,Siek:2006bh}.  For example,
the following rule reduces a cast between two function types, applied
to a value $V$, by $\eta$-expanding $V$ and inserting the appropriate
casts.
\[
\cast{V}{A \tu B \Rightarrow C \tu D}
\longrightarrow
\lambda \cast{(V \app (\cast{`\Zero}{C \Rightarrow A}))}{B \Rightarrow D}
\]
Here we define three auxiliary functions that apply casts between two
function types, two pair types, and two sum types, respectively. Each
of these functions requires the cast $c$ to be a cross cast.  These
auxiliary functions are used by cast calculi that choose to
categorize these cross casts as active casts.

\[
\begin{array}{lcl}
  \agda{eta}{\to}&:& \forall \Gamma A B C D.\;
  (M : \Gamma \vdash A \tu B) \to (c : \Cast{(A \tu B)}{(C \tu D)}) \\
  &&\quad 
  \to \agda{Cross}\app c \to (\Gamma \vdash C \to D) \\
  \agda{eta}{\to}\,M\,c\,x&=&
  \lambda \app \cast{((\agda{rename}\,\Suc\,M) \app
        (\cast{`\Zero}{\agda{dom}\,c\,x}))}{\agda{cod}\,c\,x}
\\[2ex]
    \agda{eta}{\times}&:& \forall \Gamma A B C D.\;
  (M : \Gamma \vdash A \times B) \to (c : \Cast{(A \times B)}{(C \times D)}) \\
  &&\quad  \to \agda{Cross}\app c \to (\Gamma \vdash C \times D) \\
  \agda{eta}{\times}\,M\,c\,x&=&
    \key{cons}\app \cast{(\pi_1 M)}{\key{fst}\,c\,x}
              \app \cast{(\pi_2 M)}{\key{snd}\,c\,x}
\\[2ex]
  \agda{eta}{+}&:& \forall \Gamma A B C D.\;
  (M : \Gamma \vdash A + B) \to (c : \Cast{(A + B)}{(C + D)}) \\
  &&\quad \to \agda{Cross}\app c \to (\Gamma \vdash C + D) \\
  \agda{eta}{+}\,M\,c\,x&= & \key{case}\,M \,
      (\lambda \key{inl}[D] (\cast{\Zero}{\key{inl}\,c\,x})) \,
      (\lambda \key{inr}[C] (\cast{\Zero}{\key{inr}\,c\,x})) 
\end{array}
\]

\subsection{The \agda{CastStruct} Structure}
\label{sec:caststruct}

The \agda{CastStruct} record type extends \agda{PreCastStruct} with one
more field, for applying an active cast to a value. Thus, this
structure depends on terms of the cast calculus.

\begin{description}
\item[$\agda{applyCast} : \forall \Gamma A B.\, (M : \Gamma \vdash A) \to \agda{Value}\app M\to (c : \Cast{A}{B}) \to \agda{Active}\app c \to \Gamma \vdash B$]
\end{description}

\subsection{Substitution in $\CC(\CastOp)$}

We define substitution functions (Figure~\ref{fig:substitution}) for
$\CC(\CastOp)$ in the style of PLFA~\citep{Wadler:2019aa}, due to
Conor McBride.
A \emph{renaming} is a map $\rho$ from variables (natural numbers) to
variables.
A \emph{substitution} is a map $\sigma$ from variables to terms.
The notation $M[N]$ substitutes term $N$ for all occurrences of
variable $\Zero$ inside $M$.  It's definition relies on several
auxiliary functions.
Renaming extension, $\ext{\rho}$, transports $\rho$ under one lambda
abstraction. The result maps $\Zero$ to itself, because $\Zero$ is
bound by the lambda abstraction.  For any other variable $\ext{\rho}$
transports the variable above the lambda by subtracting one, looking it
up in $\rho$, and then transports it back under the lambda by adding
one.
Simultaneous renaming, $\rename{\rho}{M}$, applies $\rho$ to all the
free variables in $M$.
Substitution extension, $\exts{\sigma}$, transports $\sigma$ under one
lambda abstraction. The result maps $\Zero$ to itself. For any other
variable $\exts{\sigma}$ transports the variable above the lambda by
subtracting one, looking it up in $\sigma$, and then transporting the
resulting term under the lambda by incrementing every free variable,
using simultaneous renaming.
Simultaneous substitution, $\subst{\sigma}{M}$, applies $\sigma$ to
the free variables in $M$.
The notation $M[N]$ is meant to be used for $\beta$ reduction, where
$M$ is the body of the lambda abstraction and $N$ is the argument.
What $M[N]$ does is substitute $\Zero$ for $N$ in $M$ and also
transports $M$ above the lambda by incrementing the other free
variables. All this is accomplished by building a substitution
$\substz{N}$ (also defined in Figure~\ref{fig:substitution}) and then
applying it to $M$.

\begin{figure}\small
  \fbox{$\rename{\rho}{M}$}
  \begin{align*}
    \rename{\rho}{k} &= k \\
    \rename{\rho}{x} &= \rho(x) \\
    \rename{\rho}{(\lambda\app  M)} &= \lambda \app (\rename{(\ext{\rho})}{N})\\
    \rename{\rho}{(M\app N)} &= (\rename{\rho}{M}) \app (\rename{\rho}{N}) \\
    \rename{\rho}{(\key{if}\app L \app M \app N)} &=
      \key{if}\app(\rename{\rho}{L})\app(\rename{\rho}{M})\app(\rename{\rho}{N})\\
    \rename{\rho}{(\key{cons}\app M \app N)} &=
      \key{cons}\app (\rename{\rho}{M}) \app (\rename{\rho}{N})\\
    \rename{\rho}{(\pi_i \app M)} &= \pi_i \app (\rename{\rho}{M}) \\
    \rename{\rho}{(\key{inl}[B]\app M)} &= \key{inl}[B]\app(\rename{\rho}{M}) \\
    \rename{\rho}{(\key{inr}[A]\app M)} &= \key{inr}[A]\app(\rename{\rho}{M}) \\
    \rename{\rho}{(\key{case}\app L \app M \app N)} &=
       \key{case}\app(\rename{\rho}{L})\app(\rename{\rho}{M})\app(\rename{\rho}{N})\\
    \rename{\rho}{(\cast{M}{c})} &= \cast{(\rename{\rho}{M})}{c}\\
    \rename{\rho}{(\blame{\ell}{})} &= \blame{\ell}{}
  \end{align*}

  \begin{minipage}{0.45\textwidth}      
  \fbox{$\ext{\rho}$}
  \begin{align*}
    \ext{\rho}\,\Zero &= \Zero \\
    \ext{\rho}\,(\Suc \, x) &= \Suc \, \rho(x)
  \end{align*}
  \end{minipage}
  \begin{minipage}{0.45\textwidth}      
  \fbox{$\exts{\sigma}$}
  \begin{align*}
  \exts{\sigma}\, \Zero &= \Zero\\
  \exts{\sigma}\, (\Suc \, x) &= \rename{\Suc}{(\sigma x)}
  \end{align*}
  \end{minipage} \\[2ex] 
  
  \fbox{$\subst{\sigma}{M}$}
  \begin{align*}
    \subst{\sigma}{k} &= k \\
    \subst{\sigma}{x} &= \sigma(x)\\
    \subst{\sigma}{(\dblam M)} &= \dblam \subst{(\exts{\sigma})}{M}\\
    \subst{\sigma}{(M \app N)} &= \subst{\sigma}{M} \app \subst{\sigma}{N}\\
    \subst{\rho}{(\key{if}\app L \app M \app N)} &=
      \key{if}\app(\subst{\rho}{L})\app(\subst{\rho}{M})\app(\subst{\rho}{N})\\
    \subst{\rho}{(\key{cons}\app M \app N)} &=
      \key{cons}\app (\subst{\rho}{M}) \app (\subst{\rho}{N})\\
    \subst{\rho}{(\pi_i \app M)} &= \pi_i \app (\subst{\rho}{M}) \\
    \subst{\rho}{(\key{inl}[B]\app M)} &= \key{inl}[B]\app(\subst{\rho}{M}) \\
    \subst{\rho}{(\key{inr}[A]\app M)} &= \key{inr}[A]\app(\subst{\rho}{M}) \\
    \subst{\rho}{(\key{case}\app L \app M \app N)} &=
       \key{case}\app(\subst{\rho}{L})\app(\subst{\rho}{M})\app(\subst{\rho}{N})\\
    \subst{\sigma}{(\cast{M}{c})} &= \cast{\subst{\sigma}{M}}{c}\\
    \subst{\sigma}{(\blame{\ell}{})} &= \blame{\ell}{}
  \end{align*}

  \begin{minipage}{0.45\textwidth}    
  \fbox{$\substz{N}$}
  \begin{align*}
  \substz{N} \, \Zero &= N \\
  \substz{N} \, (\Suc \, x) &= x
  \end{align*}
  \end{minipage}
  \begin{minipage}{0.45\textwidth}      
  \fbox{$M[N]$}
  \begin{align*}
    M[N] &= \subst{(\substz{N})}{M}
  \end{align*}
  \end{minipage}
  \caption{Substitution and its auxiliary functions.}
  \label{fig:substitution}
\end{figure}

Substitution is type preserving, which is established by the following
sequences of lemmas. As usual, we prove one theorem per function. In
Agda, these theorems are proved by embedding their statements into the
types of the four functions.  Given a sequence $S$, we write $S!i$ to
access its $i$th element.  Recall that $\Gamma$ and $\Delta$ range
over typing contexts (which are sequences of types).

\begin{lemma}[Renaming Extension]\label{lem:ext}
  Suppose that for position $x$ in $\Gamma$, $\Gamma!x = \Delta!\rho(x)$.\\
  For any $y$ and $B$, $(\Gamma \cdot B)!y = (\Delta \cdot B)!(\ext{\rho})(y)$.
\end{lemma}

\begin{lemma}[Renaming Variables]\label{lem:rename}
  Suppose that for any $x$, $\Gamma!x = \Delta!\rho(x)$.\\
  If $M : \Gamma \vdash A$, then $\rename{\rho}{M} : \Delta \vdash A$.
\end{lemma}

\begin{lemma}[Substitution Extension]\label{lem:exts}
  Suppose that for any $x$, $\sigma(x) : \Delta \vdash \Gamma!x$.\\
  For any $y$ and $B$, $\sigma(y) : (\Delta \cdot B) \vdash (\Gamma \cdot B)!y$.
\end{lemma}

\begin{proposition}[Simultaneous Substitution]\label{lem:subst}
  Suppose that for any $x$, $\sigma(x) : \Delta \vdash \Gamma!x$.\\
  If $M : \Gamma \vdash A$, then $\subst{\sigma}{M} : \Delta \vdash A$.
\end{proposition}

\begin{corollary}[Substitution]\label{lem:substitution}
  If $M : \Gamma \cdot B \vdash A$ and $\Gamma \vdash N : B$,
  then $M[N] : \Gamma \vdash A$.
\end{corollary}

\subsection{Reduction Semantics of $\CC(\CastOp)$}
\label{sec:dynamic-semantics-CC}

This section is parameterized by \agda{CastStruct}.

Figure~\ref{fig:param-cast-reduction} defines the reduction relation
for $\CC(\CastOp)$. The last eight rules are typical of the Simply
Typed Lambda Calculus, including rules for function application,
conditional branching, projecting the first or second element of a
pair, and case analysis on a sum. The congruence rule \eqref{eq:xi}
says that reduction can happen underneath a single frame.  The rule
\eqref{eq:xi-blame} propagates an exception up one frame.
Perhaps the most important rule is \eqref{eq:cast}, for applying an
active cast to a value. This reduction rule simply delegates to the
\agda{applyCast} field of the \agda{CastStruct}.
The next four rules (\ref{eq:fun-cast}, \ref{eq:fst-cast},
\ref{eq:snd-cast}, and \ref{eq:case-cast}) handle the possibility that
the \agda{CastStruct} categorizes casts between functions, pairs, or
sums as inert casts. In such situations, we need reduction rules for
when cast-wrapped values flow into an elimination form.  First, recall
that the \agda{PreCastStruct} record includes a proof that every inert
cast between two function types is a cross cast. Also recall that the
\agda{PreCastStruct} record includes fields for decomposing a cross
cast between function types into a cast on the domain and codomain.
Putting these pieces together, the reduction rule \eqref{eq:fun-cast}
says that applying the cast-wrapped function $\cast{V}{c}$ to argument
$W$ reduces to an application of $V$ to $\cast{W}{\dom{c \app x}}$
followed by the cast $\cod{c \app x}$, where $x$ is the proof that $c$
is a cross cast. The story is similar for for pairs and sums.

\begin{figure}[tbp]
  \fbox{$M \longrightarrow N$}
  \begin{gather}
    \inference{M \longrightarrow M'}
              {\itm{plug}\app M \app F
                \longrightarrow \itm{plug}\app M' \app F}
      \tag{$\xi$}\label{eq:xi} \\[2ex]
    \inference{}
              {\itm{plug}\app (\blame{\ell}{A}) \app F
                \longrightarrow \blame{\ell}{B}}
      \tag{$\xi\mhyphen\key{blame}$} \label{eq:xi-blame}\\[2ex]
    \inference{}
      {\cast{V}{c} \longrightarrow \mathsf{applyCast}\app V \app c \app a}
      ~ a : \act{c} \tag{$\key{cast}$} \label{eq:cast} \\[2ex]
    \inference{}
      {\cast{V}{c} \app W \longrightarrow
      \cast{(V \app \cast{W}{\dom{c \app x}})}{\cod{c \app x}}}
      ~ x : \agda{Cross}\app c, \inert{c}
      \tag{$\key{fun \mhyphen cast}$} \label{eq:fun-cast} \\[2ex]
    \inference{}
              {\key{fst}\app (\cast{V}{c}) \longrightarrow
                \cast{(\key{fst}\app V)}{ \agda{fst}\app c \app x}}
              ~ x : \agda{Cross}\app c, \inert{c}
              \tag{$\key{fst \mhyphen cast}$}\label{eq:fst-cast}\\[2ex]
    \inference{}
              {\key{snd}\app (\cast{V}{c}) \longrightarrow
                \cast{(\key{snd}\app V)}{ \agda{snd}\app c \app x}}
              ~ x : \agda{Cross}\app c, \inert{c}
              \tag{$\key{snd \mhyphen cast}$}\label{eq:snd-cast}\\[2ex]
   \inference{}
             {\key{case}\app (\cast{V}{c}) \app W_1 \app W_2 \longrightarrow
               \key{case}\app V \app W'_1 \app W'_2}
             \tag{$\key{case \mhyphen cast}$}\label{eq:case-cast}\\
             \text{where }
             \begin{array}{l}
               x : \agda{Cross}\app c, \inert{c} \\
               W'_1 = \lambda (\agda{rename}\app \key{S} \app W_1) \app (\cast{\key{Z}}{ \agda{inl}\app c\app x})\\
               W'_2 = \lambda (\agda{rename}\app \key{S} \app W_2) \app (\cast{\key{Z}}{ \agda{inr}\app c\app x})
             \end{array} \notag \\[2ex]
    \inference{}
              {(\lambda M) \app V \longrightarrow M[V]}
              \tag{$\beta$}\label{eq:beta}  \\[2ex]
    \inference{}
              {\key{if}\app \$\key{true} \app M \app N \longrightarrow M}
              \tag{$\beta\mhyphen\key{true}$}\label{eq:beta-true}  \\[2ex]
    \inference{}
              {\key{if}\app \$\key{false} \app M \app N \longrightarrow N}
              \tag{$\beta\mhyphen\key{false}$}\label{eq:beta-false}  \\[2ex]
    \inference{}
              {\key{fst}\app (\key{cons}\app V \app W) \longrightarrow V}
              \tag{$\beta\mhyphen\key{fst}$}\label{eq:beta-fst}  \\[2ex]
    \inference{}
              {\key{snd}\app (\key{cons}\app V \app W) \longrightarrow W}
              \tag{$\beta\mhyphen\key{snd}$}\label{eq:beta-snd}  \\[2ex]
    \inference{}
              {\key{case}\app (\key{inl}\app V)\app L \app M \longrightarrow
                L \app V}
              \tag{$\beta\mhyphen\key{caseL}$}\label{eq:beta-caseL}\\[2ex]
    \inference{}
              {\key{case}\app (\key{inr}\app V)\app L \app M \longrightarrow
                M \app V} 
              \tag{$\beta\mhyphen\key{caseR}$}\label{eq:beta-caseR}\\[2ex]
    \inference{}
              {k \app k' \longrightarrow
              \llbracket k \rrbracket \app \llbracket k' \rrbracket}
              \tag{$\delta$}\label{eq:delta}
  \end{gather}
  
  \caption{Reduction for the Parameterized Cast Calculus $\CC(\CastOp)$,
    parameterized by the \agda{CastStruct} structure.}
\label{fig:param-cast-reduction}
\end{figure}

\subsection{Type Safety of $\CC(\CastOp)$}
\label{sec:type-safety-CC}

The Preservation theorem is a direct consequence of the type that we
give to the reduction relation and that it was checked by Agda.

\begin{theorem}[Preservation]
  If $\Gamma \vdash M : A$  and $M \longrightarrow M'$,
  then $\Gamma \vdash M' : A$.
\end{theorem}

We prove the Progress theorem by defining an Agda function named
\agda{progress} that takes a closed, well-typed term $M$ and either 1)
returns a redex inside $M$, 2) identifies $M$ as a value, or 3)
identifies $M$ as an exception.

\begin{theorem}[Progress]\label{thm:cc-progress}
  If $\emptyset \vdash M : A$, then 
  \begin{enumerate}
  \item $M \longrightarrow M'$ for some $M'$,
  \item $\val{M}$, or
  \item $M \equiv \blame{\ell}{}$.
  \end{enumerate}
\end{theorem}
\begin{proof}[Proof sketch] To convey the flavor of the proof,
  we detail the cases for function application and cast application.
  The reader may read the proofs of the other cases in the Agda
  development.

  \begin{description}
  \item[Case $M_1 \app M_2$]
    The induction hypothesis for $M_1$ yields the following sub cases.

    \begin{description}
    \item[Subcase $M_1 \longrightarrow M'_1$.]
      By rule \eqref{eq:xi}, we conclude that
      \[
      M_1 \app M_2 \longrightarrow M'_1 \app M_2
      \]
    \item[Subcase $M_1 \equiv \blame{\ell}{}$.]
      By rule \eqref{eq:xi-blame}, we conclude that
      \[
      (\blame{\ell}{}) \app M_2 \longrightarrow \blame{\ell}{}
      \]
    \item[Subcase $\val{M_1}$.]
      The induction hypothesis for $M_2$ yields three sub cases.
      \begin{description}
      \item[Subcase $M_2 \longrightarrow M'_2$.]  By rule
        \eqref{eq:xi}, using $\val{M_1}$, we conclude that
        \[
        M_1 \app M_2 \longrightarrow M_1 \app M'_2
        \]
      \item[Subcase $M_2 \equiv \blame{\ell}{}$.]  By rule
        \eqref{eq:xi-blame}, using $\val{M_1}$, we conclude
        that
        \[
        M_1 \app (\blame{\ell}{}) \longrightarrow \blame{\ell}{}
        \]
      \item[Subcase $\val{M_2}$.]
        We proceed by cases on $\val{M_1}$, noting it is of function type.
        \begin{description}
        \item[Subcase $M_1 \equiv \lambda M_{11}$.]
          By rule \eqref{eq:beta}, using $\val{M_2}$, we conclude that
          \[
          (\lambda M_{11}) \app M_2 \longrightarrow M_{11}[M_2]
          \]
        \item[Subcase $M_1 \equiv \cast{V}{c}$ and $i : \inert{c}$.]
          The field $\agda{InertCross}{\to}$ of record
          \agda{PreCastStruct} gives us a proof $x$ that $c$ is
          a cross cast. So by rule \eqref{eq:fun-cast} we have
          \[
          \cast{V}{c} \app M_2 \longrightarrow
              \cast{(V \app \cast{M_2}{\dom{c \app x}})}{\cod{c \app x}}
          \]
        \item[Subcase $M_1 \equiv k_1$.]  We proceed by cases on
          $\val{M_2}$, most of which lead to contradictions and are
          therefore vacuously true.
          Suppose $M_2 \equiv k_2$. By rule \eqref{eq:delta} we conclude
          that
          \[
          k_1 \app k_2 \longrightarrow \llbracket k_1 \rrbracket \app \llbracket k_2 \rrbracket
          \]
          Suppose $M_2 \equiv \cast{M_{21}}{c}$ and $c$ is inert.
          Then $c$ is a cast on base types, which contradicts
          that $c$ is inert thanks to the
          \agda{baseNotInert} field of
          \agda{PreCastStruct}.
        \end{description}
      \end{description}
    \end{description}

  \item[Case $\cast{M}{c}$]
    The induction hypothesis for $M$ yields three sub cases.
    \begin{description}
    \item[Subcase $M \longrightarrow M'$.]
      By rule \eqref{eq:xi} we conclude that
      \[
      \cast{M}{c} \longrightarrow \cast{M'}{c}
      \]
    \item[Subcase $M \equiv \blame{\ell}{}$.]
      By rule \eqref{eq:xi-blame} we conclude that
      \[
      \cast{(\blame{\ell}{})}{c} \longrightarrow \blame{\ell}{}
      \]
    \item[Subcase $\val{M}$.]
      Here we use the \agda{ActiveOrInert} field of the
      \agda{PreCastStruct} on the cast $c$.
      Suppose $c$ is active, so we have $a : \act{c}$.
      By rule \eqref{eq:cast}, using $\val{M}$, we conclude that
      \[
      \cast{M}{c} \longrightarrow \agda{applyCast} \app M \app c \app a
      \]
      Suppose $c$ is inert. Then we conclude that
      $\val{\cast{M}{c}}$.
    \end{description}
  \end{description}
\end{proof}

\subsection{The $\CC'(\CastOp)$ Variant}

\citet{Siek:2015ac}, in their Isabelle mechanization of the GTLC and
the Dynamic Gradual Guarantee, make a syntactic distinction between
cast applications as unevaluated expressions versus cast applications
that are part of values. (The paper glosses over this detail.)  In
particular, they introduce a term constructor named \agda{Wrap} for
casts between function types that are part of a value and another term
constructor named \agda{Inj} for casts into $\Unk$ that are part of a
value. The reason they make this distinction is to enable a smaller
simulation relation in their proof of the Dynamic Gradual
Guarantee.\footnote{\citet{Wadler:2009qv} make similar distinctions
for different reasons in the Blame Calculus.}
To aid the proof the Dynamic Gradual Guarantee in
Section~\ref{sec:dynamic-gradual-guarantee}, we introduce a variant of
$\CC(\CastOp)$, named $\CC'(\CastOp)$ (Figure~\ref{fig:cc-alt}), that
makes a syntactic distinction between arbitrary casts and inert
casts. We write inert cast application as $\castInert{M}{c}$. Instead
of making the further distinction as in \agda{Wrap} and \agda{Inj}, we
rely on the source and target types of the casts to distinguish those
cases.
With the addition of inert casts, we change the definition of
$\val{M}$, removing the ($\mathtt{Vcast}$) rule and replacing it with
($\mathtt{Vwrap}$), which states that $\castInert{V}{c}$ is a value if
$c$ is inert. We also change the reduction rules \eqref{eq:fun-cast},
\eqref{eq:fst-cast}, \eqref{eq:snd-cast}, and \eqref{eq:case-cast} to
eliminate values of the form $\castInert{V}{c}$ instead of
$\cast{V}{c}$.
This change regarding inert casts does not affect the observable
behavior of programs with respect to $\CC(\CastOp)$.

We also change the reduction rule \eqref{eq:case-cast} to the rule
\eqref{eq:case-cast-alt} which uses substitution instead of
introducing $\lambda$-abstractions.  This change allows us to avoid
relating terms to their $\eta$-expansions. We conjecture that this
second change is not necessary, but would require a more complex
definition of the simulation relation to take $\eta$-expansion into
account. This change to \eqref{eq:case-cast} is an observable
difference. Consider an expression of the form:
\[
\key{case}\app \castInert{(\agda{inl}\app V)}{c} \app M \app N
\]
Suppose variable $\Zero$ does not occur in $M$ and the cast
$\agda{inl}\app c$ always fails.  With \eqref{eq:case-cast} there is
an error but with \eqref{eq:case-cast-alt} there is not.

\begin{figure}[tp]
\fbox{$M : \Gamma \vdash' A$}
\begin{gather*}
\key{case}:
    \inference{\Gamma \vdash' A_1 + A_2 & \Gamma \cdot A_1 \vdash' B & \Gamma \cdot A_2 \vdash' B}
              {\Gamma \vdash' B}
\\[2ex]  
\castInert{-}{c}:
    \inference{\Gamma \vdash' A}{\Gamma \vdash' B}~ c : \Cast{A}{B}, \inert{c}
\end{gather*}

  \fbox{$\val{} : (\Gamma \vdash' A) \to \agda{Set}$}
  \begin{gather*}
    \mathtt{Vwrap} :
    \inference{\val{M}}
              { \val{(\castInert{M}{c})} } ~\agda{Inert}\,c 
  \end{gather*}

  \fbox{$\rename{\rho}{M}$}
  \begin{align*}
    \rename{\rho}{(\key{case}\app L \app M \app N)} &=
    \key{case}\app(\rename{\rho}{L})\app(\rename{(\ext{\rho})}{M})\app(\rename{(\ext{\rho})}{N}) \\
    \rename{\rho}{(\castInert{M}{c})} &= \castInert{(\rename{\rho}{M})}{c}
  \end{align*}
    \fbox{$\subst{\sigma}{M}$}
  \begin{align*}
    \subst{\rho}{(\key{case}\app L \app M \app N)} &=
    \key{case}\app(\subst{\rho}{L})\app(\subst{(\exts{\rho})}{M})\app(\subst{(\exts{\rho})}{N})\\
    \subst{\rho}{(\castInert{M}{c})} &= \castInert{(\subst{\rho}{M})}{c}
  \end{align*}
  
\fbox{$M \longrightarrow N$}
\begin{gather}
  \inference{}
            {\cast{V}{c} \longrightarrow \castInert{V}{c}}~\inert{c}
            \tag{\key{wrap}}\label{eq:wrap}
            \\[2ex]
    \inference{}
      {\castInert{V}{c} \app W \longrightarrow
      \cast{(V \app \cast{W}{\dom{c \app x}})}{\cod{c \app x}}}
      ~ x : \agda{Cross}\app c, \inert{c}
      \tag{$\key{fun \mhyphen cast \mhyphen alt}$} \label{eq:fun-cast-alt} \\[2ex]
    \inference{}
              {\key{fst}\app (\castInert{V}{c}) \longrightarrow
                \cast{(\key{fst}\app V)}{ \agda{fst}\app c \app x}}
              ~ x : \agda{Cross}\app c, \inert{c}
              \tag{$\key{fst \mhyphen cast \mhyphen alt}$}\label{eq:fst-cast-alt}\\[2ex]
    \inference{}
              {\key{snd}\app (\castInert{V}{c}) \longrightarrow
                \cast{(\key{snd}\app V)}{ \agda{snd}\app c \app x}}
              ~ x : \agda{Cross}\app c, \inert{c}
              \tag{$\key{snd \mhyphen cast \mhyphen alt}$}\label{eq:snd-cast-alt}\\[2ex]
   \inference{}
             {\key{case}\app (\castInert{V}{c}) \app M \app N \longrightarrow
               \key{case}\app V \app M' \app N'}
             \tag{$\key{case \mhyphen cast \mhyphen alt}$}\label{eq:case-cast-alt}\\
             \text{where }
             \begin{array}{l}
               x : \agda{Cross}\app c, \inert{c} \\
               M' = (\rename{(\ext{\Suc})}{M}) [\cast{`\Zero}{\agda{inl}\app c \app x}]\\
               N' = (\rename{(\ext{\Suc})}{N}) [\cast{`\Zero}{\agda{inr}\app c \app x}]
             \end{array} \notag
\end{gather}
\caption{The $\CC'(\CastOp)$ Variant.}
\label{fig:cc-alt}
\end{figure}

\subsection{Blame-Subtyping Theorem for $\CC'(\CastOp)$}
\label{sec:blame-CC}

Recall that the Blame-Subtyping Theorem~\citep{Siek:2015ac} states
that a cast from $A$ to $B$ cannot be blamed (cause a runtime error)
if the source type is a subtype of the target type, i.e., $A <: B$.
We identify a cast with a blame label $\ell$. Thanks to the
Blame-Subtyping Theorem, programmer or static analyzer can inspect the
source type, the target type, and the blame label of a cast and decide
whether the cast is entirely safe or whether it might cause a runtime
error.  During execution, casts with different labels can be composed
(see Section~\ref{sec:EfficientParamCasts}), producing casts that
contain multiple blame labels. So more generally, we say that a cast
is blame-safe w.r.t a label $\ell$ if the cast will not give rise to a
runtime error labeled $\ell$.

In this section we develop a parameterized proof of the
blame-subtyping theorem for the Parameterized Cast Calculus
$\CC'(\CastOp)$ following the approach of \citet{Wadler:2009qv}.
However, the appropriate subtyping relation and definition of
blame-safe for casts depends on the representation and semantics of
casts~\citep{Siek:2009rt}, which is a parameter of
$\CC'(\CastOp)$. Thus, the definition of blame-safe is a parameter of
the proof: the \agda{CastBlameSafe} predicate. The proof is not
parameterized on the subtyping relation, clients of the proof will
typically use subtyping to define whether a cast is blame-safe.

We require that applying a blame-safe cast to a value does not raise a
run-time error, that is, if $c$ is blame-safe for $\ell$, then
$\agda{applyCast} \app V \app c \app a \neq \blame{\ell}{}$.  However,
to handle the application of casts to higher-order values, we must
generalize this requirement. The output of \agda{applyCast} can be any
term, so in addition to ruling out $\blame{\ell}{}$ we also need to
rule out terms that contain casts that are not blame-safe for
$\ell$. Thus, we define a predicate, written
$\safefor{M}{\ell}$~\citep{Wadler:2009qv}, that holds if all the casts
in $M$ are blame-safe for $\ell$, defined in
Figure~\ref{fig:castsallsafe}.  Most importantly, $\blame{\ell'}{}$ is
safe for $\ell$ only when $\ell \neq \ell'$. The other important case
is the one for casts: $\cast{M}{c}$ is safe for $\ell$ if
$\agda{CastBlameSafe}\app c \app \ell$ and $M$ is recursively safe for
$\ell$.

\begin{figure}[tp]
  \fbox{$\safefor{M}{\ell}$}
  \begin{gather*}
    \inference[(\textsc{Cast})]
              {\agda{CastBlameSafe} \app c \app \ell & \safefor{M}{\ell}}
              {\safefor{(\cast{M}{c})}{\ell}}
              \quad
    \inference[(\textsc{Var})]
              {}
              {\safefor{(` x)}{\ell}}
              \\[1em]
    \inference[(\textsc{Lam})]
              {\safefor{N}{\ell}}
              {\safefor{(\lambda N)}{\ell}}
              \quad
    \inference[(\textsc{App})]
              {\safefor{L}{\ell} & \safefor{M}{\ell}}
              {\safefor{(L \app M)}{\ell}}
              \\[1em]
    \inference[(\textsc{Lit})]
              {}
              {\safefor{(\$ k)}{\ell}}
              \quad
    \inference[(\textsc{Cons})]
              {\safefor{M}{\ell} & \safefor{N}{\ell}}
              {\safefor{(\key{cons} \, M \, N)}{\ell}}
              \\[1em]
    \inference[(\textsc{If})]
              {\safefor{L}{\ell} & \safefor{M}{\ell} & \safefor{N}{\ell}}
              {\safefor{(\key{if} \, L \, M \, N)}{\ell}}
              \\[1em]
    \inference[(\textsc{Fst})]
              {\safefor{M}{\ell}}
              {\safefor{(\key{fst} \, M)}{\ell}}
              \quad
    \inference[(\textsc{Snd})]
              {\safefor{M}{\ell}}
              {\safefor{(\key{snd} \, M)}{\ell}}
              \\[1em]
    \inference[(\textsc{InL})]
              {\safefor{M}{\ell}}
              {\safefor{(\key{inl} [B] \, M)}{\ell}}
              \quad
    \inference[(\textsc{InR})]
              {\safefor{M}{\ell}}
              {\safefor{(\key{inr} [A] \, M)}{\ell}}
              \\[1em]
    \inference[(\textsc{Case})]
              {\safefor{L}{\ell} & \safefor{M}{\ell} & \safefor{N}{\ell}}
              {\safefor{(\key{case} \, L \, M \, N)}{\ell}}
              \\[1em]
    \inference[(\textsc{Blame})]
              {\ell \neq \ell'}
              {\safefor{(\key{blame} \, \ell')}{\ell}}
  \end{gather*}
  \caption{The predicate that a term contains only blame-safe casts.}
  \label{fig:castsallsafe}
\end{figure}

We also require that \agda{CastBlameSafe} be preserved under the
operations on casts: \agda{dom}, \agda{cod}, \agda{fst}, \agda{snd},
\agda{inl}, and \agda{inr}.

We capture these requirements in two structures:
\begin{enumerate}
\item \agda{PreBlameSafe} and
\item \agda{BlameSafe}.
\end{enumerate}
The first structure includes the \agda{CastBlameSafe} predicate and
the fields that preserve blame-safety under the casts operations
(Figure~\ref{fig:PreBlameSafe}).  The second structure
includes the requirement that \agda{applyCast} preserves blame safety.

\begin{figure}[tp]
\begin{description}
\item[$\agda{CastBlameSafe} : \forall A B . \, \Cast{A}{B} \to \agda{Label} \to \agda{Set}$]\ \\
  Predicate for blame-safe casts, i.e. casts that never cause runtime errors.
  
\item[$\agda{domBlameSafe} :
    \begin{array}{l}
      \forall A_1 A_2 B_1 B_2.\,  (c : \Cast{(A_1 \to A_2)}{(B_1 \to B_2)}) \to \agda{CastBlameSafe} \app c \app \ell \\
         \to (x : \agda{Cross} \app c) \to \agda{CastBlameSafe} \app (\agda{dom} \app c \app x) \app \ell
    \end{array}$] 
  Given a cross cast between function types that is blame-safe,
  $\agda{domBlameSafe}$ returns a proof that the cast between
  the domain types is blame-safe.

\item[$\agda{codBlameSafe} :
    \begin{array}{l}
      \forall A_1 A_2 B_1 B_2.\,  (c : \Cast{(A_1 \to A_2)}{(B_1 \to B_2)}) \to \agda{CastBlameSafe} \app c \app \ell \\
       \to (x : \agda{Cross} \app c) \to \agda{CastBlameSafe} \app (\agda{cod} \app c \app x) \app \ell
    \end{array}$]
  Similar to the above but for the codomain.
  
\item[$\agda{fstBlameSafe} :
    \begin{array}{l}
      \forall A_1 A_2 B_1 B_2.\, (c : \Cast{(A_1 \times A_2)}{(B_1 \times B_2)}) \to \agda{CastBlameSafe} \app c \app \ell \\
       \to (x : \agda{Cross} \app c) \to \agda{CastBlameSafe} \app (\agda{fst} \app c \app x) \app \ell
    \end{array}$]
  Given a cross cast between product types that is blame-safe,
  $\agda{fstBlameSafe}$ returns a proof that the cast between
  the first components of the pair is blame-safe.

\item[$\agda{sndBlameSafe} :
    \begin{array}{l}
      \forall A_1 A_2 B_1 B_2.\, (c : \Cast{(A_1 \times A_2)}{(B_1 \times B_2)}) \to \agda{CastBlameSafe} \app c \app \ell \\
      \to (x : \agda{Cross} \app c) \to \agda{CastBlameSafe} \app (\agda{snd} \app c \app x) \app \ell
    \end{array}$]
  Similar to the above but for the second component of the pair.
  
\item[$\agda{inlBlameSafe} :
    \begin{array}{l}
      \forall A_1 A_2 B_1 B_2.\, (c : \Cast{(A_1 + A_2)}{(B_1 + B_2)}) \to \agda{CastBlameSafe} \app c \app \ell \\
       \to (x : \agda{Cross} \app c) \to \agda{CastBlameSafe} \app (\agda{inl} \app c \app x) \app \ell
    \end{array}$]
  Given a cross cast between sum types that is blame-safe,
  $\agda{inlBlameSafe}$ returns a proof that the cast for the first
  branch is blame-safe.

\item[$\agda{inrBlameSafe} :
    \begin{array}{l}
      \forall A_1 A_2 B_1 B_2.\,  (c : \Cast{(A_1 + A_2)}{(B_1 + B_2)}) \to \agda{CastBlameSafe} \app c \app \ell \\
       \to (x : \agda{Cross} \app c) \to \agda{CastBlameSafe} \app (\agda{inr} \app c \app x) \app \ell
    \end{array}$]
  Similar to the above but for the second component of the sum.
\end{description}
\caption{\agda{PreBlameSafe}
   extends \agda{PreCastStruct} (Section~\ref{sec:precaststruct}).}
\label{fig:PreBlameSafe}
\end{figure}

The structure \agda{BlameSafe} extends
\agda{PreBlameSafe} and \agda{CastStruct}.  It also
includes the following field which, roughly speaking, requires that
applying a blame-safe cast to a blame-safe value $V$ results in a
blame-safe term.

\[
\begin{split}
\agda{applyCast\mhyphen{}pres\mhyphen{}allsafe} :
\forall \Gamma A B \ell . \, & (V : \Gamma \vdash A) \to \agda{Value} \app V \to (c : \Cast{A}{B}) \\
& \to (a : \agda{Active} \app c) \to \agda{CastBlameSafe} \app c \app \ell\\
& \to \safefor{V}{\ell} \\
& \to \safefor{(\agda{applyCast} \app V \app c \app a)}{\ell}
\end{split}
\]

We turn to the proof of blame safety. We first prove that ``safe for''
is preserved during reduction in the following
Lemma~\ref{lem:preserve-cas}.  The proof depends on a number of
technical lemmas about renaming and substitution.  They are omitted
here but can be found in the Agda formalization.

\begin{lemma}[Preservation of Blame Safety]
  \label{lem:preserve-cas}
  If $M : \Gamma \vdash A$, $M' : \Gamma \vdash A$,
  $\safefor{M}{\ell}$, and $M \longrightarrow M'$,
  then $\safefor{M'}{\ell}$.
\end{lemma}

Using the preservation Lemma~\ref{lem:preserve-cas}, we prove the
following Blame-Subtyping Theorem~\ref{thm:blame-subtyping}. The
theorem says that if every cast with label $\ell$ is blame-safe in
$M$, then these casts never fail and $\ell$ is guaranteed not to be
blamed.

\begin{theorem}[Blame-Subtyping Theorem]
  \label{thm:blame-subtyping}
  If $M : \Gamma \vdash A$ and $\safefor{M}{\ell}$,
  then $\neg (M \longrightarrow^{*} \key{blame} \, \ell)$.
\end{theorem}
\begin{proof}[Proof sketch]
  By induction on the reduction sequence and by inversion on the
  ``safe for'' predicate.
\end{proof}

\subsection{Dynamic Gradual Guarantee}
\label{sec:dynamic-gradual-guarantee}

Recall that the dynamic gradual guarantee~\citep{Siek:2015ac} states
that changing type annotations in a program to be more precise should
either result in the same observable behavior or, in the event of a
conflict between the new type annotation and some other part of the
program, the program could produce a runtime cast error. On the other
hand, changing the program to be less precise should always result in
the same observable behavior.

Our proof of the dynamic gradual guarantee follows the same structure
as that of \citet{Siek:2015ac}.  Their proof involves two main lemmas:
(1) compilation from GTLC to its cast calculus preserves precision and
(2) a more-precise program in the cast calculus simulates any
less-precise version of itself.  In this article, the proof that
compilation preserves precision is in
Section~\ref{sec:gtlc-to-cc}. Our proof of the simulation is in this
section and is parameterized with respect to the cast representation,
using two new structures, so that it can be applied to different cast
calculi.

\subsubsection{Precision Structures and Precision for $\CC'$}

The definition of term precision relies on a notion of precision for
the cast representation, so precision for casts must be a parameter of
the proof, that is, fields of a required structure.  In particular, we
define \agda{PreCastStructWithPrecision} as an extension of
\agda{PreCastStruct} in
Figure~\ref{fig:PreCastStructWithPrecision}. This structure includes
three fields that define the precision relation between two casts and
precision between a cast and a type (on the left or right). The record
also includes four requirements on those relations that correspond to
the precision rules in Figure 9 of \citet{Siek:2015ac}:
\agda{\mbox{$\leprecop\to$}lpit} corresponds to the forward direction
of the bottom right rule while \agda{lpii\mbox{$\to\leprecop$}},
\agda{lpit\mbox{$\to\leprecop$}}, and \agda{lpti\mbox{$\to\leprecop$}}
are inversion properties about the top left, bottom right, and bottom
left rules respectively.

\begin{figure}[tp]
\begin{description}
\item[$\agda{\leprecii{-}{-}} : \forall c \, c' . \, \agda{Inert} \, c \to \agda{Inert} \, c' \to \agda{Set}$]\  \\
  The precision relation between two inert casts.
\item[$\agda{\leprecit{-}{-}} : \forall c . \, \agda{Inert} \, c \to \Types \to \agda{Set}$]\  \\
  The precision relation between an inert cast and a type. 
\item[$\agda{\leprecti{-}{-}} : \forall c' . \, \Types \to \agda{Inert} \, c' \to \agda{Set}$]\  \\
  The precision relation between a type and an inert cast. 
\item[$\agda{\mbox{$\leprecop\to$}lpit} :$]\ \\
  \begin{equation*} \begin{split}
      \forall A B A' . \, (c : \Cast{A}{B}) \to (i : \agda{Inert} \app c) \to A \leprecop A' \to B \leprecop A' \to \leprecit{i}{A'}
  \end{split} \end{equation*}
  If the source and target of an inert cast are less precise than
  another type, then so is the inert cast.
\item[$\agda{lpii\mbox{$\to\leprecop$}} :$]\ \\
  \begin{equation*} \begin{split}
      \forall A B A' B' . \, & (c : \Cast{A}{B}) \to (c' : \Cast{A'}{B'}) \to (i : \agda{Inert} \app c) \to (i' : \agda{Inert} \app c') \\
      & \to \leprecii{i}{i'} \to (A \leprecop A') \times (B \leprecop B')
  \end{split} \end{equation*}
  If an inert cast is less precise than another, then its source and
  target are too.
\item[$\agda{lpit\mbox{$\to\leprecop$}} :$]\ \\
  \begin{equation*} \begin{split}
      \forall A A' B . \, & (c : \Cast{A}{B}) \to (i : \agda{Inert} \app c)
      \to \leprecit{i}{A'} \to (A \leprecop A') \times (B \leprecop A')
  \end{split} \end{equation*}
  If an inert cast is less precise than a type, then so are its source
  and target.
\item[$\agda{lpti\mbox{$\to\leprecop$}} :$]\ \\
  \begin{equation*} \begin{split}
      \forall A A' B' . \, & (c' : \Cast{A'}{B'}) \to (i' : \agda{Inert} \app c')
      \to \leprecti{A}{i'} \to (A \leprecop A') \times (A \leprecop B')
  \end{split} \end{equation*}
  If a type is less precise than an inert cast, then it is less
  precise than the source and target.
\end{description}
\caption{\agda{PreCastStructWithPrecision} extends \agda{PreCastStruct}.}
\label{fig:PreCastStructWithPrecision}
\end{figure}

We define the precision relation between $\CC'(\CastOp)$ terms in
Figure~\ref{fig:cc'-prec}.  The rules (\textsc{Cast}),
(\textsc{CastL}), and (\textsc{CastR}) correspond to the precision
rules for casts in Figure 9 of \citet{Siek:2015ac}.  The rule
(\textsc{Wrap}) corresponds to the rules \agda{p\_wrap\_wrap} and
\agda{p\_inj} in the Isabelle mechanization of \citet{Siek:2015ac}.
The side condition that $B = \Unk$ implies $B' = \Unk$ ensures that
our (\textsc{Wrap}) rule handles the same cases as \agda{p\_inj}
(where both $B = \Unk$ and $B' = \Unk$) and \agda{p\_wrap\_wrap}
(where neither $B = \Unk$ nor $B' = \Unk$).
The (\textsc{WrapL}) rule corresponds to the rules \agda{p\_wrap\_r}
and \agda{p\_injr}.\footnote{The precision ordering in
\citet{Siek:2015ac} is flipped with respect to our which explains the
switching of left versus right.}  The (\textsc{WrapR}) rule
corresponds to the rule \agda{p\_wrap\_l}.  The side condition $A \neq
\Unk$ prevents overlap with the (\textsc{Wrap}) rule which we explain
as follows. If $A = \Unk$, then by Lemma~\ref{lem:canonical-star} the
term $M$ would be an injection and therefore the (\textsc{Wrap}) rule
would apply. (For the same reason, there is no \agda{p\_injl} rule in
the mechanization of \citet{Siek:2015ac}.)
The term $\blame{\ell}{}$ is treated as more precise than any term,
just as in \citet{Siek:2015ac}.
The rest of the rules are equivalent to
the precision rules for GTLC (Figure~\ref{fig:gtlc-prec}).

\begin{figure}[tp]
  \fbox{$\sigma \leprecop^s \sigma'$}
  \begin{gather*}
    \inference[($\sigma_0$)]{M \lepreccc M'}{\agda{substZero} \app M \leprecop^s \agda{substZero} \app M'}
    \qquad
    \inference[(\textsc{Exts})]{\sigma \leprecop^s \sigma'}{\agda{exts} \app \sigma \leprecop^s \agda{exts} \app \sigma'}
  \end{gather*}

  \fbox{$M \lepreccc M'$}
  \begin{gather*}
    \inference[(\textsc{Lit})]{}{\$ k \lepreccc \$ k}
    \quad
    \inference[(\textsc{Var})]{}{` x \lepreccc ` x}
    \\[1ex]
    \inference[(\textsc{Lam})]{A \leprecop A' & N \lepreccc N'}{\lambda \, N \lepreccc \lambda \, N'}
    N : \Gamma \cdot A \vdash B, N' : \Gamma' \cdot A' \vdash B'
    \\[1ex]
    \inference[(\textsc{App})]{L \lepreccc L' & M \lepreccc M'}{L \app M \lepreccc L' \app M'}
    \quad
    \inference[(\textsc{If})]{L \lepreccc L' & M \lepreccc M' & N \lepreccc N'}
              {\key{if} \, L \, M \, N \lepreccc \key{if} \, L' \, M' \, N'}
    \\[1ex]
    \inference[(\textsc{Cons})]{M \lepreccc M' & N \lepreccc N'}{\key{cons} \, M \, N \lepreccc \key{cons} \, M' \, N'}
    \quad
    \inference[(\textsc{Fst})]{M \lepreccc M'}{\pi_1 \, M \lepreccc \pi_1 \, M'}
    \quad
    \inference[(\textsc{Snd})]{M \lepreccc M'}{\pi_2 \, M \lepreccc \pi_2 \, M'}
    \\[1ex]
    \inference[(\textsc{InL})]{B \leprecop B' & M \lepreccc M'}{\key{inl} [B] \, M \lepreccc \key{inl} [B'] \, M'}
    \quad
    \inference[(\textsc{InR})]{A \leprecop A' & M \lepreccc M'}{\key{inr} [A] \, M \lepreccc \key{inr} [A'] \, M'}
    \\[1ex]
    \inference[(\textsc{Case})]{A_1 \leprecop A_1' & A_2 \leprecop A_2' & L \lepreccc L' & M \lepreccc M' & N \lepreccc N'}
              {\key{case} \, L \, M \, N \lepreccc \key{case} \, L' \, M' \, N'} \\
              L : \Gamma \vdash A_1 + A_2 , L' : \Gamma' \vdash A_1' + A_2'
    \\[1ex]
    \inference[(\textsc{Cast})]
              {A \leprecop A' & B \leprecop B' & M \lepreccc M'}
              {\cast{M}{c} \lepreccc \cast{M'}{c'}}
              c : \Cast{A}{B}, c' : \Cast{A'}{B'}
    \\[1ex]
    \inference[(\textsc{CastL})]
              {A \leprecop A' & B \leprecop A' & M \lepreccc M'}
              {\cast{M}{c} \lepreccc M'}
              c : \Cast{A}{B}, M' : \Gamma' \vdash A'
    \\[1ex]
    \inference[(\textsc{CastR})]
              {A \leprecop A' & A \leprecop B' & M \lepreccc M'}
              {M \lepreccc \cast{M'}{c'}}
              M : \Gamma \vdash A, c' : \Cast{A'}{B'}
    \\[1ex]
    \inference[(\textsc{Wrap})]
              {\leprecii{i}{i'} & M \lepreccc M'}
              {\castInert{M}{i} \lepreccc \castInert{M'}{i'}}
              i : \Cast{A}{B}, i' : \Cast{A'}{B'}, (B = \Unk) \to (B' = \Unk)
    \\[1ex]
    \inference[(\textsc{WrapL})]
              {\leprecit{i}{A'} & M \lepreccc M'}
              {\castInert{M}{i} \lepreccc M'}
              M' : \Gamma' \vdash A'
    \\[1ex]
    \inference[(\textsc{WrapR})]
              {\leprecti{A}{i'} & M \lepreccc M'}
              {M \lepreccc \castInert{M'}{i'}}
              M : \Gamma \vdash A, A \neq \Unk
    \\[1ex]
    \inference[(\textsc{Blame})]
              {A \leprecop A'}
              {M \lepreccc \key{blame} \, \ell}
              M : \Gamma \vdash A, \key{blame} \, \ell : \Gamma' \vdash A'
  \end{gather*}
  \caption{Precision between $\CC'(\CastOp)$ terms and between substitutions.}
  \label{fig:cc'-prec}
\end{figure}

There are five technical lemmas of \citet{Siek:2015ac} used in the
simulation proof that directly involve the cast representation, so
those too become structure fields. We extend \agda{CastSruct} to
create the structure \agda{CastStructWithPrecision}
in Figure~\ref{fig:CastStructWithPrecision}

\begin{figure}[tp]
\begin{description}
\item[$\agda{applyCast\mbox{-}catchup} :$]\ \\
  \begin{equation*} \begin{split}
      \forall \Gamma \Gamma' A A' B . \, & (V : \Gamma \vdash A) \to (V' : \Gamma' \vdash A') \to (c : \Cast{A}{B}) \\
      & \to (a : \agda{Active} \app c) \to \agda{Value} \app V \to \agda{Value} \app V' \\
      & \to A \leprecop A' \to B \leprecop A' \to V \lepreccc V' \\
      & \to \exists W . (\agda{Value} \app W) \times (\agda{applyCast} \app V \app c \app a \longrightarrow^{*} W) \times (W \lepreccc V')
  \end{split} \end{equation*}
\item[$\agda{sim\mbox{-}cast} :$]\ \\
  \begin{equation*} \begin{split}
      \forall A A' B B' . \, & (V : \emptyset \vdash A) \to (V' : \emptyset \vdash A') \to (c : \Cast{A}{B}) \to (c' : \Cast{A'}{B'}) \\
      & \to \agda{Value} \app V \to \agda{Value} \app V' \to (a' : \agda{Active} \app c') \\
      & \to A \leprecop A' \to B \leprecop B' \to V \lepreccc V' \\
      & \to \exists N . (\cast{V}{c} \longrightarrow^{*} N) \times (N \lepreccc \agda{applyCast} \app V' \app c' \app a')
  \end{split} \end{equation*}
\item[$\agda{sim\mbox{-}wrap} :$]\ \\
  \begin{equation*} \begin{split}
      \forall A A' B B' . \, & (V : \emptyset \vdash A) \to (V' : \emptyset \vdash A') \to (c : \Cast{A}{B}) \to (c' : \Cast{A'}{B'}) \\
      & \to \agda{Value} \app V \to \agda{Value} \app V' \to (i' : \agda{Inert} \app c') \\
      & \to A \leprecop A' \to B \leprecop B' \to V \lepreccc V' \\
      & \to \exists N . (\cast{V}{c} \longrightarrow^{*} N) \times (N \lepreccc \castInert{V'}{i'})
  \end{split} \end{equation*}
\item[$\agda{castr\mbox{-}cast} :$]\ \\
  \begin{equation*} \begin{split}
      \forall A A' B . \, & (V : \emptyset \vdash A) \to (V' : \emptyset \vdash A') \to (c' : \Cast{A'}{B'}) \\
      & \to \agda{Value} \app V \to \agda{Value} \app V' \to (a' : \agda{Active} \app c') \\
      & \to A \leprecop A' \to A \leprecop B' \to V \lepreccc V' \\
      & \to V \lepreccc \agda{applyCast} \app V' \app c' \app a'
  \end{split} \end{equation*}
\item[$\agda{castr\mbox{-}wrap} :$]\ \\
  \begin{equation*} \begin{split}
      \forall A A' B . \, & (V : \emptyset \vdash A) \to (V' : \emptyset \vdash A') \to (c' : \Cast{A'}{B'}) \\
      & \to \agda{Value} \app V \to \agda{Value} \app V' \to (i' : \agda{Inert} \app c') \\
      & \to A \leprecop A' \to A \leprecop B' \to V \lepreccc V' \\
      & \to V \lepreccc \castInert{V'}{i'}
  \end{split} \end{equation*}
\end{description}
\caption{The \agda{CastStructWithPrecision} structure extends \agda{CastStruct}.}
\label{fig:CastStructWithPrecision}
\end{figure}





\subsubsection{Proof of the Simulation}

This section is parameterized by \agda{CastStructWithPrecision} and
contains two lemmas that lead up to the proof of the simulation.  The
first lemma states that a term that is less precise than a value will
catch up by reducing to a value.

\begin{lemma}[Catchup to Value]
  \label{prop:catchup}
  Suppose $M : \Gamma \vdash A$ and $V' : \Gamma' \vdash A'$.  If $M
  \lepreccc V'$, then $M \longrightarrow^{*} V$ and $V \lepreccc V'$
  for some $V$.
\end{lemma}

\begin{proof}[Proof sketch]
  The proof is by induction on the precision relation $M \lepreccc
  V'$.  All other cases are trivial using induction hypotheses about
  subterms except the one for $\textsc{CastL}$.  Suppose $M \equiv
  \cast{N}{c}$.  By the induction hypothesis for $N$, we have $N
  \longrightarrow^{*} V$ for some $V$ and $V \lepreccc V'$.  For
  $\cast{V}{c}$, we cast on whether cast $c$ is inert or active: if
  inert, then $\cast{V}{c} \longrightarrow \castInert{V}{i}$ where $i
  : \agda{Inert} \app c$ by rule \key{wrap}; if active, then it is
  proved by the field \agda{applyCast\mbox{-}catchup} of
  \agda{CastStructWithPrecision}
  (Figure~\ref{fig:CastStructWithPrecision}).
\end{proof}

\noindent Next we prove that substitution preserves precision.

\begin{lemma}[Substitution Preserves $\lepreccc$]
  \label{lem:subst-pres-prec}
  If $\sigma \leprecop^s \sigma'$ and $N \lepreccc N'$, \\
  then $\agda{subst} \app \sigma \app N \lepreccc \agda{subst} \app
  \sigma' \app N'$.
\end{lemma}
\begin{proof}[Proof sketch]
  The proof is by induction on the precision relation $N \lepreccc
  N'$.
\end{proof}

\noindent We come to the proof of the main lemma, that a program
simulates any more-precise version of itself.

\begin{lemma}[Simulation of More Precise Programs]
  \label{lem:simulation-more-precise}
  Suppose $M_1 : \emptyset \vdash A$ and $M_1' : \emptyset \vdash A'$
  and $M_2' : \emptyset \vdash A'$. If $M_1 \lepreccc M_1'$ and $M_1'
  \longrightarrow M_2'$, then $M_1 \longrightarrow^{*} M_2$ and $M_2
  \lepreccc M_2'$ for some $M_2$.
\end{lemma}
\begin{proof}[Proof sketch]
  The proof is by induction on the precision relation $M_1 \lepreccc
  M_1'$. The full proof is in Agda; we briefly describe a few cases
  here.
  \begin{description}
  \item[Application ($M_1 \equiv L \app M$ and $M_1' \equiv L' \app M'$).]
    We case on the reduction $L' \app M' \longrightarrow M_2'$, which yields 7 sub cases:
    \begin{description}
    \item[Subcase $\xi$, where frame $F = (\Box \app \_)$.]
      By the induction hypothesis about the precision between $L$ and $L'$,
      there exists $L_2$ such that $L \longrightarrow^{*} L_2$ and $L_2$ satisfies the precision relation.
      Since multi-step reduction is a congruence, there exists $M_2 \equiv L_2 \app M$
      that satisfies $M_1 \longrightarrow^{*} M_2$ and $M_2 \lepreccc M_2'$.
    \item[Subcase $\xi$, where frame $F = (\_ \app \Box)$.]
      By Proposition~\ref{prop:catchup}, there exists $L_2$ such that:
      1) $L_2$ is a value;
      2) $L \longrightarrow^{*} L_2$;
      3) $L_2 \lepreccc L'$.
      By the induction hypothesis about $M$ and $M'$, there exists $N_2$ such that
      $M \longrightarrow^{*} N_2$ and $N_2$ satisfies the precision relation. Since multi-step reduction
      is a congruence and is transitive, there exists $M_2 \equiv L_2 \app N_2$ that satisfies
      $M_1 \longrightarrow^{*} M_2$ and $M_2 \lepreccc M_2'$.
    \item[Subcase $\xi\mhyphen\key{blame}$, where frame $F = (\Box \app \_)$.]
      $L \app M$ itself satisfies the precision relation since blame is more precise
      than any term, so there exists $M_2 \equiv L \app M$.
    \item[Subcase $\xi\mhyphen\key{blame}$, where frame $F = (\_ \app \Box)$.]
      Same as the previous case.
    \item[Subcase $\beta$.] Using Proposition~\ref{prop:catchup} twice,
      $L \app M \longrightarrow^{*} V \app W$, where $V$, $W$ are both
      values and $V \lepreccc L'$, $W \lepreccc M'$. We know that
      $L' \equiv \lambda \, N'$ for some $N'$ due to $\beta$.
      By induction on the precision relation $V \lepreccc \lambda \, N'$ and
      inversion on $\agda{Value} \app V$, it further generates two sub sub-cases:
      \begin{description}
      \item[Sub-subcase: $V$ is $\mathtt{V}\lambda$]
        Suppose $V \equiv \lambda \, N$, then there exists $M_2 \equiv N [ W ]$
        that satisfies both the reduction (from $V \app W$ by $\beta$) and the precision.
      \item[Sub-subcase: $V$ is $\mathtt{Vwrap}$]
        Suppose $V \equiv \castInert{V_1}{i}$ for some value $V_1$ and $i : \agda{Inert} \app c$:
        \[
        \castInert{V_1}{i} \app W \longrightarrow \cast{(V_1 \app \cast{W}{\agda{dom} \app c})}{\agda{cod} \app c}
        \text{, by } \key{fun\mbox{-}cast}.
        \]
        By using Proposition~\ref{prop:catchup} and congruence of reduction:
        \[
        \cast{(V_1 \app \cast{W}{\agda{dom} \app c})}{\agda{cod} \app c} \longrightarrow^{*} \cast{V_1 \app W_1}{\agda{cod} \app c}
        \]
        for some value $W_1$. Note that $V_1$, $W_1$ are now both values,
        about which by induction hypothesis there exists $N$ such that $V_1 \app W_1 \longrightarrow^{*} N$
        and $N$ satisfies the precision relation.
      \end{description}
    \item[Subcase $\delta$.] Similar to the previous case,
      by repeatedly using Proposition~\ref{prop:catchup} (``catch-up'') and the induction hypothesis.
    \item[Subcase $\key{fun{-}cast}$.] Similar to the previous case,
      by repeatedly using Proposition~\ref{prop:catchup} (``catch-up'') and the induction hypothesis.
    \end{description}
  \item[Cast ($M_1 \equiv \cast{M}{c}$ and $M_1' \equiv \cast{M'}{c'}$).]
    We case on the reduction $\cast{M'}{c'} \longrightarrow M_2'$, which yields 4 sub cases:
    \begin{description}
    \item[Subcase $\xi$.] By the induction hypothesis about subterms $M$ and $M'$ and that
      multi-step reduction is a congruence.
    \item[Subcase $\xi\mhyphen\key{blame}$.] $\cast{M}{c}$ itself satisfies the precision
      relation since blame is more precise than any term.
    \item[Subcase $\key{cast}$.] By Proposition~\ref{prop:catchup} and that $M'$ is value, there exists
      value $V$ such that $M \longrightarrow^{*} V$ and $V \lepreccc M'$. Then this case
      is directly proved by the field \agda{sim\mbox{-}cast} of \agda{CastStructWithPrecision}
      (Figure~\ref{fig:CastStructWithPrecision}).
    \item[Subcase $\key{wrap}$.] Similar to the previous case, first use Proposition~\ref{prop:catchup}
      and then use field \agda{sim\mbox{-}wrap} of \agda{CastStructWithPrecision}
      (Figure~\ref{fig:CastStructWithPrecision}).
    \end{description}
    Similarly, fields \agda{castr\mbox{-}cast} and \agda{castr\mbox{-}wrap} of \agda{CastStructWithPrecision}
    are used in the case for \key{castr}, which is omitted since the idea is the same as the
    \key{cast} case described above.
  \item[Wrap ($M_1 \equiv \castInert{M}{i}$ and $M_1' \equiv \cast{M'}{i'}$).]
    We case on the reduction $\castInert{M'}{i'} \longrightarrow M_2'$, which yields 2 sub cases:
    \begin{description}
    \item[Subcase $\xi$.] By the induction hypothesis about subterms $M$ and $M'$ and that
      multi-step reduction is a congruence.
    \item[Subcase $\xi\mhyphen\key{blame}$.] $\castInert{M}{i}$ itself satisfies the precision
      relation since blame is more precise than any term.
    \end{description}
  \end{description}
\end{proof}

\section{Compilation of GTLC to $\CC(\CastOp)$.}
\label{sec:gtlc-to-cc}

This section is parameterized by the cast representation $\CastOp$ and
by a cast constructor, written $\coerce{A}{B}{\ell}$, that builds a
cast from a source type $A$, target type $B$, blame label $\ell$, and
(implicitly) a proof that $A$ and $B$ are consistent.

The compilation functions $\compile{}{-}$ and $\compilenew{}{-}$ are
defined in Figure~\ref{fig:compile-gtlc-cc} and map a well-typed term
of the GTLC to a well-typed term of the Parameterized Cast Calculus
$\CC(\CastOp)$ or $\CC'(\CastOp)$, respectively. The two functions are
the same except in the equation for \key{case}, so we only show that
equation for $\compilenew{}{-}$. The Agda type signatures of the
compilation functions ensure that they are type preserving.

\begin{figure}[tbp]
  \fbox{$\compile{-}{-} : \forall \Gamma A. \, (\Gamma \vdash_G A) \to (\Gamma \vdash A)$}\\
  \fbox{$\compilenew{-}{-} : \forall \Gamma A. \, (\Gamma \vdash_G A) \to (\Gamma \vdash' A)$}
\begin{align*}
\compile{\Gamma}{\$ k} &= k \\
\compile{\Gamma}{x} &= x \\
\compile{\Gamma}{\lambda[A]\, M} &= \lambda \, \compile{}{M}\\
\compile{\Gamma}{(L \app M)_\ell} &=
    \cast{\compile{\Gamma}{L}}{c} \app \cast{\compile{\Gamma}{M}}{d}\\
  & \text{where } L : \Gamma \vdash_G A, M : \Gamma \vdash_G B,
   \matchfun{A}{A_1}{A_2}, \\
  & c = \coerce{A}{(A_1 \tu A_2)}{\ell},  \text{and } d = \coerce{B}{A_1}{\ell}\\
\compile{\Gamma}{\key{if}_\ell \app L \app M \app N }] &=
        \key{if}\app (\cast{\compile{}{L}}{c})
         \app (\cast{\compile{}{M}}{d_1})
         \app (\cast{\compile{}{N}}{d_2}) \\
      & \text{where } L : \Gamma \vdash_G A, 
          M : \Gamma \vdash_G B_1, N : \Gamma \vdash_G B_2, \\
          & c = \coerce{A}{\Bool}{\ell},\\
      &  d_1 = \coerce{B_1}{B_1 \sqcup B_2}{\ell}, 
          d_2 = \coerce{B_2}{B_1 \sqcup B_2}{\ell} \\
\compile{\Gamma}{\key{cons}\app M \app N}] &= \key{cons}\app \compile{}{M} \app \compile{}{N} \\
  \compile{\Gamma}{\pi_i^\ell \app M}] &= \pi_i \app (\cast{\compile{}{M}}{c})\\
    & \text{where } M : \Gamma \vdash_G A, \matchpair{A}{A_1}{A_2},
       c = \coerce{A}{A_1 \times A_2}{\ell} \\
\compile{\Gamma}{\key{inl}[B]\app M}] &= \key{inl} \app \compile{}{M} \\
\compile{\Gamma}{\key{inr}[A]\app M}] &= \key{inr} \app \compile{}{M} \\
\compile{\Gamma}{\key{case}_\ell[B_1,C_1] \app L \app M \app N}] &=
  \key{case}\app \cast{\compile{}{L}}{c} \app (\lambda \cast{\compile{}{M}}{d_1})
  \app (\lambda \cast{\compile{}{N}}{d_2}) \quad \text{(for $\CC(\CastOp)$)}\\
  & \text{where } L : \Gamma \vdash_G A, M : \Gamma \cdot B_1 \vdash_G B_2, N : \Gamma \cdot C_1\vdash_G C_2\\
  & c = \coerce{A}{B_1 + C_1}{\ell} \\
  & d_1 = \coerce{B_2}{B_2 \sqcup C_2}{\ell}, d_2 = \coerce{C_2}{B_2 \sqcup C_2}{\ell}\\
\compilenew{\Gamma}{\key{case}_\ell[B_1,C_1] \app L \app M \app N}] &=
  \key{case}\app \cast{\compile{}{L}}{c} \app \cast{\compile{}{M}}{d_1}
  \app \cast{\compile{}{N}}{d_2} \quad \text{(for $\CC'(\CastOp)$)}
\end{align*}
\caption{Compilation from GTLC to $\CC(\CastOp)$ and $\CC'(\CastOp)$.}
\label{fig:compile-gtlc-cc}
\end{figure}

\begin{lemma}[Compilation Preserves Types]\ 
  \begin{enumerate}
  \item If $M : \Gamma \vdash_G A$, then $\compile{}{M} : \Gamma \vdash A$.
  \item If $M : \Gamma \vdash_G A$, then $\compilenew{}{M} : \Gamma \vdash' A$.
  \end{enumerate}
\end{lemma}

The compilation function also preserves precision, which is an
important lemma in the proof of the dynamic gradual guarantee.  This
lemma is parameterized over the \agda{PreCastStructWithPrecision}
structure (Figure~\ref{fig:PreCastStructWithPrecision}).

\begin{lemma}[Compilation Preserves Precision]
  \label{lem:compile-pres-prec}
  Suppose $M : \Gamma \vdash_G A$ and $M' : \Gamma' \vdash_G A'$.
  If $\leprec{\Gamma}{\Gamma'}$ and $M \leprecgtlc M'$,
  then $\compilenew{}{M} \lepreccc \compilenew{}{M'}$ and $\leprec{A}{A'}$.
\end{lemma}

\section{A Half-Dozen Cast Calculi}
\label{sec:CC-instances}

We begin to reap the benefits of creating the Parameterized Cast
Calculus by instantiating it with six different implementations of
\agda{CastStruct} to produce six cast calculi.

\subsection{Partially-Eager ``D'' Casts with Active Cross Casts (EDA)}
\label{sec:EDA}

The cast calculus defined in this section corresponds to the original
one of \citet{Siek:2006bh}, although their presentation used a
big-step semantics instead of a reduction semantics and did not
include blame tracking. In the nomenclature of \citet{Siek:2009rt},
this calculus is partially-eager and uses the ``D'' blame tracking
strategy. As we shall see shortly, we categorize cross casts as
active, so we refer to this cast calculus as the partially-eager ``D''
cast calculus with active cross casts (EDA).

We define the EDA calculus and prove that it is type safe and blame
safe by instantiating the meta-theory of the Parameterized Cast
Calculus. This requires that we define instances of the structures
\agda{PreCastStruct}, \agda{CastStruct},
\agda{PreBlameSafe}, and
\agda{BlameSafe}.

We begin by defining the cast representation type $\CastOp$ as a data
type with a single constructor also named $\CastOp$ that takes two
types, a blame label, and a proof of consistency:
\\[1ex]
\fbox{$c :\Cast{A}{B}$}
\[
\inference{}
          {A \Rightarrow^\ell B  : \Cast{A}{B}} ~ A \sim B
\]

\noindent The cast constructor is defined as follows.
\[
  \coerce{A}{B}{\ell} = (A \Rightarrow^\ell B)
\]

\subsubsection{Reduction Semantics and Type Safety}
          
We categorize just the casts into $\Unk$, the injections, as inert
casts.\\[1ex]
\fbox{$\inert{c}$}
\[
\inference{}
          {\inert{(A \Rightarrow^\ell \Unk)}}~ A \neq \Unk
\]
We categorize casts between function, pair, and sum types
as cross casts.\\[1ex]
\fbox{$\cross{c}$}
\begin{gather*}
  \inference
    {}
    {\mathtt{Cross}\otimes : \cross{(A \otimes B \Rightarrow^\ell C \otimes D)}} ~ \otimes \in \{ \to , \times, + \} 
\end{gather*}
We categorize the identity, projection, and cross casts as active.\\[1ex]
\fbox{$\act{c}$}
\begin{gather*}
  \inference
      {}
      {\mathtt{ActId} : \act{(a \Rightarrow^\ell a)}}
  \quad
  \inference
    {}
    {\mathtt{ActProj} : \act{(\Unk \Rightarrow^\ell B)}}~B \neq \Unk
  \\[1ex]
  \inference
    {\cross{c}}
    {\mathtt{ActCross} : \act{c}}
\end{gather*}

\begin{lemma}
  \label{lem:simple-active-or-inert}
  For any types $A$ and $B$, $c : \Cast{A}{B}$ is either an active or
  inert cast.
\end{lemma}
\begin{proof}
  The cast $c$ must be of the form $A \Rightarrow^\ell B$ where $A \sim B$.
  We proceed by cases on $A \sim B$.
  \begin{description}
  \item[Case $\consis{A}{\Unk}$] If $A \equiv \Unk$, then $c$ is active
    by $\mathtt{ActId}$. Otherwise $c$ is inert.
  \item[Case $\consis{\Unk}{B}$] If $B \equiv \Unk$, then $c$ is active
    by $\mathtt{ActId}$. Otherwise $c$ is active by $\mathtt{ActProj}$.
  \item[Case $\consis{\Base}{\Base}$]
    $c$ is active by $\mathtt{ActId}$.
  \item[Case $\consis{A \otimes B}{A' \otimes B'}$]
    where $\otimes \in \{ \to, \times, + \}$.
    $c$ is active by $\mathtt{ActCross}$ and $\mathtt{Cross}\otimes$.
  \end{description}
\end{proof}

Next we show that inert casts into non-atomic types are cross casts.
\begin{lemma}\label{lem:simple-inert-cross}
  If $c : \Cast{A}{(B \otimes C)}$ and $\inert{c}$,
  then $\cross{c}$ and $A \equiv D \otimes E$ for some $D$ and $E$.
\end{lemma}
\begin{proof}
  There are no inert casts whose target type is
  $B \otimes C$ (it must be $\Unk$), so this lemma is
  vacuously true.
\end{proof}

Continuing on the topic of cross casts, we define $\agda{dom}$,
$\agda{cod}$, etc. as follows. The second parameter $x$ is the
evidence that the first parameter is a cross cast.
\begin{align*}
  \dom{(A \to B \Rightarrow^\ell C \to D) \app x} &=
      C \Rightarrow^\ell A \\
  \cod{(A \to B \Rightarrow^\ell C \to D) \app x} &=
      B \Rightarrow^\ell D \\
  \agda{fst}\app(A \times B \Rightarrow^\ell C \times D) \app x &=
      A \Rightarrow^\ell C \\
  \agda{snd}\app(A \times B \Rightarrow^\ell C \times D) \app x &=
      B \Rightarrow^\ell D \\
  \agda{inl} \app (A + B \Rightarrow^\ell C + D) \app x &=
      A \Rightarrow^\ell C \\
  \agda{inr} \app (A + B \Rightarrow^\ell C + D) \app x &=
      B \Rightarrow^\ell D 
\end{align*}

We check that a cast to a base type is not inert.
\begin{lemma}\label{lem:simple-base-not-inert}
   A cast $c : \Cast{A}{b}$ is not inert.
\end{lemma}
\begin{proof}
  This is easy to verify because $b \neq \Unk$.
\end{proof}

\begin{proposition}\label{prop:simple-active-pre-cast-struct}
The EDA calculus is an instance of the $\agda{PreCastStruct}$
structure.
\end{proposition}

We import and instantiate the definitions and lemmas from
Section~\ref{sec:cc-values-frames} (values and frames) and
Section~\ref{sec:eta-cast-reduction} (eta-like reduction rules).

Next we define the \agda{applyCast} function by cases on the proof
that the cast is active. The case below for $\mathtt{ActProj}$ relies
on Lemma~\ref{lem:canonical-star} to know that the term is of the form
$\cast{M}{A \Rightarrow^{\ell_1} \Unk}$.
\[
  \agda{applyCast} : \forall \Gamma A B.\, (M : \Gamma \vdash A) \to \agda{Value}\app M\to (c : \Cast{A}{B}) \to \agda{Active}\app c \to \Gamma \vdash B
\]
\[
\begin{array}{lllllll}
  \agda{applyCast} & M & v & A \Rightarrow^\ell A & \mathtt{ActId} &=&
     M \\
  \agda{applyCast} & \cast{M}{A \Rightarrow^{\ell_1} \Unk} & v & \Unk \Rightarrow^{\ell_2} B & \mathtt{ActProj} &=& 
  \begin{cases}
    \cast{M}{A \Rightarrow^{\ell_2} B} & \text{if } A \sim B\\
    \blame{\ell_2}{} & \text{otherwise}
  \end{cases}\\
  \agda{applyCast} &M& v & c & \mathtt{Act}{\otimes} &=&
    \agda{eta}{\otimes}\app M\app c \app \mathtt{Cross}\otimes
\end{array}
\]

\begin{proposition}
  \label{prop:peac-cast-struct}
The EDA calculus is an instance of the $\agda{CastStruct}$
structure.
\end{proposition}

We instantiate from $\CC(\CastOp)$ the reduction semantics
(Section~\ref{sec:dynamic-semantics-CC}) and proof of type safety
(Section~\ref{sec:type-safety-CC}). 

\begin{definition}[Reduction for EDA]
  The reduction relation $M \longrightarrow N$ for the EDA calculus
  is the reduction relation of $\CC(\CastOp)$ instantiated with
  EDA's instance of \agda{CastStruct}.
\end{definition}

\begin{corollary}[Preservation for EDA]
  \label{thm:EDA-preservation}
  If $\Gamma \vdash M : A$  and $M \longrightarrow M'$,
  then $\Gamma \vdash M' : A$.
\end{corollary}

\begin{corollary}[Progress for EDA]\label{thm:EDA-progress}
  If $\emptyset \vdash M : A$, then 
  \begin{enumerate}
  \item $M \longrightarrow M'$ for some $M'$,
  \item $\val{M}$, or
  \item $M \equiv \blame{\ell}{}$.
  \end{enumerate}
\end{corollary}

Let EDA$'$ be the variant of EDA obtained by instantiating $\CC'$
instead of $\CC$.

\subsubsection{Blame-Subtyping}
\label{sec:EDA-blame-subtyping}

Because the EDA$'$ calculus uses the ``D'' blame-tracking strategy,
the subtyping relation that corresponds to safe casts is the one in
Figure~\ref{fig:subtype-D} where the unknown type $\Unk$ is the top of
the subtyping order.

\begin{figure}[tp]
  \begin{gather*}
    \inference{}{A <: \Unk}
    \quad
    \inference{}{b <: b}
    \quad
    \inference{C <: A & B <: D}{A \to B <: C \to D}
    \\[1ex]
    \inference{A <: C & B <: D}{A \times B <: C \times D}
    \quad
    \inference{A <: C & B <: D}{A + B <: C + D} 
  \end{gather*}
  \caption{Subtyping for ``D'' blame-tracking}
  \label{fig:subtype-D}
\end{figure}

We define the \agda{CastBlameSafe} predicate as follows
\[
\inference{ A <: B}
          {\agda{CastBlameSafe}\app (A \Rightarrow^\ell B) \app \ell}
\quad
\inference{\ell \neq \ell'}
          {\agda{CastBlameSafe}\app (A \Rightarrow^{\ell'} B) \app \ell}
\]

\begin{lemma}[Blame Safety of Cast Operators]
  \label{lem:blame-safety-cast-operators-EDA}
  Suppose $\agda{CastBlameSafe}\app c \app \ell$ and $c$ is a cross
  cast, that is, $x : \cross{c}$.
  \begin{itemize}
  \item If $x=\mathtt{Cross}\to$,
    then $\agda{CastBlameSafe}\app (\agda{dom}\app c \app x)\app \ell$
    and $\agda{CastBlameSafe}\app (\agda{cod}\app c \app x)\app \ell$.
  \item If $x=\mathtt{Cross}\times$,
    then $\agda{CastBlameSafe}\app (\agda{fst}\app c \app x)\app \ell$
    and $\agda{CastBlameSafe}\app (\agda{snd}\app c \app x)\app \ell$.
  \item If $x=\mathtt{Cross}+$,
    then $\agda{CastBlameSafe}\app (\agda{inl}\app c \app x)\app \ell$
    and $\agda{CastBlameSafe}\app (\agda{inr}\app c \app x)\app \ell$.
  \end{itemize}
\end{lemma}

\begin{proposition}
  \label{prop:EDA-preblamesafe}
  EDA$'$ is an instance of the \agda{PreBlameSafe} structure.
\end{proposition}

\begin{lemma}[\agda{applyCast} preserves blame safety in EDA$'$]
  If $\agda{CastBlameSafe}\app c \app \ell$ and $a : \act{c}$
  and $\agda{CastsAllSafe}\app V \app \ell$
  and $v : \val{V}$, then
  $\agda{CastsAllSafe}\app (\agda{applyCast}\app V \app v \app c \app a) \app \ell$.
\end{lemma}
\begin{proof}
  The proof is by cases on $a$ (the cast being active and if a cross
  cast, cases on the three kinds) and then by cases on
  $\agda{CastBlameSafe}\app c \app \ell$, except that if the cast is
  an identity cast, then there is no need for casing on
  $\agda{CastBlameSafe}\app c \app \ell$. So there are 9 cases to
  check, but they are all straightforward.
\end{proof}

\begin{proposition}
  EDA$'$ is an instance of the \agda{BlameSafe} structure.
\end{proposition}

We instantiate Theorem~\ref{thm:blame-subtyping} with this
\agda{BlameSafe} instance to obtain the Blame-Subtyping Theorem for
EDA$'$.

\begin{corollary}[Blame-Subtyping Theorem for EDA$'$]
  \label{thm:blame-subtyping-EDA}
  If $M : \Gamma \vdash A$ and $\agda{CastsAllSafe} \app M \app \ell$,
  then $\neg (M \longrightarrow^{*} \key{blame} \, \ell)$.
\end{corollary}


\subsection{Partially-Eager ``D'' Casts with Inert Cross Casts (EDI)}
\label{sec:EDI}

Many cast
calculi~\citep{Wadler:2007lr,Wadler:2009qv,Siek:2009rt,Siek:2010ya,Siek:2015ab}
categorize terms of the form
\[
  \cast{V}{A \to B \Rightarrow C \to D}
\]
as values and then define the reduction rule
\[
\cast{V}{A \to B \Rightarrow C \to D} \app W
\longrightarrow
\cast{(V \app (\cast{W}{C \Rightarrow A}))}{B \Rightarrow D}
\]
That is, we categorize cross casts as inert instead of active.  In
this section we take this approach for functions, pairs, and sums, but
keep everything else the same as in the previous section.  We
conjecture that this cast calculus is equivalent in observable
behavior to EDA.  We refer to this calculus as EDI, where the I stands
for ``inert''.

So the cast representation type is again made of a source type, target
type, blame label, and consistency proof.\\[1ex]
\fbox{$c : \Cast{A}{B}$}
\[
\inference{}
          {A \Rightarrow^\ell B  : \Cast{A}{B}} A \sim B
\]
\noindent The cast constructor is defined as follows.
\[
  \coerce{A}{B}{\ell} = (A \Rightarrow^\ell B)
\]

\subsubsection{Reduction Semantics and Type Safety}
          
Again, we categorize casts between function, pair, and sum types as
cross casts.  \\[1ex]
\fbox{$\cross{c}$}
\begin{gather*}
  \inference
    {\otimes \in \{ \to , \times, + \} }
    {\mathtt{Cross}\otimes : \cross{(A \otimes B \Rightarrow^\ell C \otimes D)}}
\end{gather*}
But for the inert casts, we include the cross casts this time.\\[1ex]
\fbox{$\inert{c}$}
\[
\inference{A \neq \Unk}
          {\mathtt{InInj} : \inert{(A \Rightarrow^\ell \Unk)}}
\quad
\inference{\cross{c}}
          {\mathtt{InCross} : \inert{c}}
\]
The active casts include just the identity casts and projections.\\[1ex]
\fbox{$\act{c}$}
\[
  \inference
      {}
      {\mathtt{ActId} : \act{(a \Rightarrow^\ell a)}}
  \quad
  \inference
    {B \neq \Unk}
    {\mathtt{ActProj} : \act{(\Unk \Rightarrow^\ell B)}}
\]

\begin{lemma}
  \label{lem:pedi-active-or-inert}
  For any types $A$ and $B$, $c : \Cast{A}{B}$ is either an active or
  inert cast.
\end{lemma}
\begin{proof}
  The proof is similar to that of Lemma~\ref{lem:simple-active-or-inert},
  except the cross casts are categorized as inert instead of active.
\end{proof}

\begin{lemma}\label{lem:pedi-inert-cross}
  If $c : \Cast{A}{(B \otimes C)}$ and $\inert{c}$,
  then $\cross{c}$ and $A \equiv D \otimes E$ for some $D$ and $E$.
\end{lemma}
\begin{proof}
  We proceed by cases on $\inert{c}$.
  \begin{description}
  \item[Case $\mathtt{InInj}$:] The target type is $\Unk$, not $B \otimes C$,
    so we have a contradiction.
  \item[Case $\mathtt{InCross}$:] We have $\cross{c}$.
    By cases on $\cross{c}$, we have
    $A \equiv D \otimes E$ for some $D$ and $E$.
  \end{description}
\end{proof}

The definitions of $\agda{dom}$, $\agda{cod}$, etc. are exactly the
same as in Section~\ref{sec:EDA}.

\begin{lemma}\label{lem:pedi-base-not-inert}
   A cast $c : \Cast{A}{b}$ is not inert.
\end{lemma}
\begin{proof}
  The inert casts in this section have a target type of either $\Unk$
  or a non-atomic type $A_1 \otimes A_2$, but not a base type.
\end{proof}

\begin{proposition}\label{prop:pedi-inert-pre-cast-struct}
The EDI calculus is an instance of the $\agda{PreCastStruct}$
structure.
\end{proposition}

Again we define the \agda{applyCast} function by cases on the proof
that the cast is active, but this time there is one less case to
consider (the cross casts). The cases for $\mathtt{ActId}$ and
$\mathtt{ActProj}$ are the same as in Section~\ref{sec:EDA}.
\[
  \agda{applyCast} : \forall \Gamma A B.\, (M : \Gamma \vdash A) \to \agda{Value}\app M\to (c : \Cast{A}{B}) \to \agda{Active}\app c \to \Gamma \vdash B
\]
\[
\begin{array}{lllllll}
  \agda{applyCast} & M & v & a \Rightarrow^\ell a & \mathtt{ActId} &=&
     M \\
  \agda{applyCast} & \cast{M}{A \Rightarrow^{\ell_1} \Unk} & v & \Unk \Rightarrow^{\ell_2} B & \mathtt{ActProj} &=& 
  \begin{cases}
    \cast{M}{A \Rightarrow^{\ell_2} B} & \text{if } A \sim B\\
    \blame{\ell_2}{} & \text{otherwise}
  \end{cases}
\end{array}
\]

\begin{proposition}
The EDI calculus is an instance of the $\agda{CastStruct}$ structure.
\end{proposition}

We instantiate from $\CC(\CastOp)$ the reduction semantics
(Section~\ref{sec:dynamic-semantics-CC}) and proof of type safety
(Section~\ref{sec:type-safety-CC}) for EDI.

\begin{definition}[Reduction for EDI]
  The reduction relation $M \longrightarrow N$ for the EDI calculus
  is the reduction relation of $\CC(\CastOp)$ instantiated with EDI's
  instance of \agda{CastStruct}.
\end{definition}

\begin{corollary}[Preservation for EDI]
  \label{thm:pedi-cast-preservation}
  If $\Gamma \vdash M : A$  and $M \longrightarrow M'$,
  then $\Gamma \vdash M' : A$.
\end{corollary}

\begin{corollary}[Progress for EDI]\label{thm:pedi-cast-progress}
  If $\emptyset \vdash M : A$, then 
  \begin{enumerate}
  \item $M \longrightarrow M'$ for some $M'$,
  \item $\val{M}$, or
  \item $M \equiv \blame{\ell}{}$.
  \end{enumerate}
\end{corollary}

Let EDI$'$ be the variant of EDI obtained by instantiating $\CC'$
instead of $\CC$.

\subsubsection{Blame-Subtyping}

We define the \agda{CastBlameSafe} predicate for EDI$'$ exactly the
same way as for EDA$'$, using the subtyping relation of
Figure~\ref{fig:subtype-D}.

Because the cast operators for EDI are defined exactly the same as
for EDA, Lemma~\ref{lem:blame-safety-cast-operators-EDA} also
applies to EDI, that is, the cast operators such as \agda{dom}
preserve blame safety.

\begin{proposition}
  \label{prop:EDA-preblamesafe}
  EDI$'$ is an instance of the \agda{PreBlameSafe} structure.
\end{proposition}

\begin{lemma}[\agda{applyCast} preserves blame safety in EDI$'$]
  \label{lem:applyCast-safe-pedi}
  If $\agda{CastBlameSafe}\app c \app \ell$ and $a : \act{c}$
  and $\agda{CastsAllSafe}\app V \app \ell$
  and $v : \val{V}$, then
  $\agda{CastsAllSafe}\app (\agda{applyCast}\app V \app v \app c \app a) \app \ell$.
\end{lemma}
\begin{proof}
  The proof of this lemma is much shorter than the corresponding one
  for EDA$'$ because there are fewer casts categorized as active.  So
  there are just 3 straightforward cases to check.
\end{proof}

\begin{proposition}
  EDI$'$ is an instance of the \agda{BlameSafe} structure.
\end{proposition}

We instantiate Theorem~\ref{thm:blame-subtyping} with this
\agda{BlameSafe} instance to obtain the Blame-Subtyping Theorem for
EDI$'$.

\begin{corollary}[Blame-Subtyping Theorem for EDI$'$]
  \label{thm:blame-subtyping-EDA}
  If $M : \Gamma \vdash A$ and $\agda{CastsAllSafe} \app M \app \ell$,
  then $\neg (M \longrightarrow^{*} \key{blame} \, \ell)$.
\end{corollary}

\subsection{The $\lambda\textup{\textsf{B}}$ Blame Calculus}
\label{sec:lambda-B}

This section generates the $\lambda\textup{\textsf{B}}$
variant~\citep{Siek:2015ab} of the Blame
Calculus~\citep{Wadler:2009qv} as an instance of the Parameterized
Cast Calculus.  It also extends $\lambda\textup{\textsf{B}}$ in
\citet{Siek:2015ab} with product types and sum types. We explore
treating cross casts on products and sums both as active and inert
casts in Section~\ref{sec:type-safety-lambda-B}.  Compared to the
previous sections, the main difference in $\lambda\textup{\textsf{B}}$
is that all injections and projections factor through \emph{ground
types}, defined as follows:
\[
  \text{Ground Types} \quad G, H ::= b \mid \Unk \to \Unk \mid \Unk \times \Unk \mid \Unk + \Unk
\]

The cast representation type consists of a source type, target type,
blame label, and consistency proof.\\[1ex]
\fbox{$c : \Cast{A}{B}$}
\[
\inference{}
          {A \Rightarrow^\ell B  : \Cast{A}{B} } A \sim B
\]
\noindent The cast constructor is defined as follows.
\[
  \coerce{A}{B}{\ell} = (A \Rightarrow^\ell B)
\]

\subsubsection{Reduction Semantics and Type Safety}
\label{sec:type-safety-lambda-B}

We categorize casts between function, pair, and sum types as cross
casts.  \\[1ex]
\fbox{$\cross{c}$}
\begin{gather*}
  \inference
    {}
    {\mathtt{Cross}\otimes : \cross{(A \otimes B \Rightarrow^\ell C \otimes D)}} ~ \otimes \in \{ \to , \times, + \} 
\end{gather*}

Regarding inert casts in $\lambda\textup{\textsf{B}}$, an injection
from ground type is inert and a cast between function types is inert.
As for casts between product types and sum types, there are two
choices, either to make them inert or to make them active.
This yields two variants of $\lambda\textup{\textsf{B}}$:
Variant 1 where all cross casts are inert and Variant 2
where only function casts are inert. \\[1ex]

\fbox{$\inert{c}$. Variant 1}
\[
\inference{}
          {\mathtt{InInj} : \inert{(G \Rightarrow^\ell \Unk)}}
\quad
\inference{\cross{c}}
          {\mathtt{InCross} : \inert{c}}
\]

\fbox{$\inert{c}$. Variant 2}
\[
\inference{}
          {\mathtt{InInj} : \inert{(G \Rightarrow^\ell \Unk)}}
\quad
\inference{}
          {\mathtt{InFun} : \inert{(A \to B \Rightarrow^\ell C \to D)}}
\]

The active casts in both $\lambda\textup{\textsf{B}}$ variations
include injections from non-ground type, projections, and identity
casts.  In Variant 2, we categorize only function casts as active, so
compared to Variant 1 we have an additional rule
$\mathtt{Act}{\otimes}$ for casts between product types and sum types
being active. \\[1ex]

\fbox{$\act{c}$. Variant 1}
\begin{gather*}
  \inference
      {}
      {\mathtt{ActId} : \act{(a \Rightarrow^\ell a)}}
  \;\;
  \inference
  {}
  {\mathtt{ActInj} : \act{(A \Rightarrow^\ell \Unk)}}~
  A \not\equiv \Unk, \not\exists G.\, A \equiv G
  \\[1ex]
  \inference
  {}
  {\mathtt{ActProj} : \act{(\Unk \Rightarrow^\ell B)}}~B \neq \Unk
\end{gather*}

\fbox{$\act{c}$. Variant 2}
\begin{gather*}
  \inference
      {}
      {\mathtt{ActId} : \act{(a \Rightarrow^\ell a)}}
  \;\;
  \inference
      {}
      {\mathtt{ActInj} : \act{(A \Rightarrow^\ell \Unk)}}
      ~ A \not\equiv \Unk, \not\exists G.\, A \equiv G
  \\[1ex]
  \inference
    {}
    {\mathtt{ActProj} : \act{(\Unk \Rightarrow^\ell B)}}~B \neq \Unk
  \\[1ex]
  \inference
    {}
    {\mathtt{Act}{\otimes} :\act{(A \otimes B \Rightarrow^\ell C \otimes D)}}
    \otimes \in \{ \times, + \}
\end{gather*}

\begin{lemma}
  \label{lem:lambda-B-active-or-inert}
  For any types $A$ and $B$, $c : \Cast{A}{B}$ is either an active or
  inert cast.
\end{lemma}

\begin{lemma}\label{lem:lambda-B-inert-cross}
  If $c : \Cast{A}{(B \otimes C)}$ and $\inert{c}$,
  then $\cross{c}$ and $A \equiv D \otimes E$ for some $D$ and $E$.
\end{lemma}

\begin{lemma}\label{lem:lambda-B-base-not-inert}
   A cast $c : \Cast{A}{b}$ is not inert.
\end{lemma}

\begin{proposition}\label{prop:lambda-B-pre-cast-struct}
  Both variants of $\lambda\textup{\textsf{B}}$ are instances of
  the $\agda{PreCastStruct}$ structure.
\end{proposition}

We define the following partial function named $\agda{gnd}$ (short for
``ground'').  It is defined on all types except for $\Unk$.

\fbox{$\agda{gnd}\app A$}
\begin{align*}
  \agda{gnd} \app b &= b \\
  \agda{gnd} \app A \otimes B &= \Unk \otimes \Unk
    & \text{for } \otimes \in \{ \to, \times, + \} 
\end{align*}

Also, we use the following shorthand for a sequence of two casts:
\[
\cast{M}{A \Rightarrow^{\ell_1} B \Rightarrow^{\ell_2} C}
=
\cast{\cast{M}{A \Rightarrow^{\ell_1} B}}{B \Rightarrow^{\ell_2} C}
\]

The following is the definition of \agda{applyCast} for
$\lambda\textup{\textsf{B}}$ Variant 1:

\[
  \agda{applyCast} : \forall \Gamma A B.\, (M : \Gamma \vdash A) \to \agda{Value}\app M\to (c : \Cast{A}{B}) \to \agda{Active}\app c \to \Gamma \vdash B
\]
\[
\begin{array}{lllllll}
  \agda{applyCast} & M & v & a \Rightarrow^\ell a & \mathtt{ActId} &=&
     M \\
  \agda{applyCast} & M & v & A \Rightarrow^\ell \Unk & \mathtt{ActInj} &=&
  \cast{M}{A \Rightarrow^\ell \agda{gnd}\app A \Rightarrow^\ell \Unk} \\
  \agda{applyCast} & \cast{M}{G \Rightarrow^{\ell_1} \Unk} & v & \Unk \Rightarrow^{\ell_2} B & \mathtt{ActProj} &=& M' 
\end{array}
\]
where
\[
  M' =
  \begin{cases}
    M  & \text{if } B = \agda{gnd} \app B = G \\
    \blame{\ell_2}{}   & \text{if } B = \agda{gnd} \app B \neq G\\
    \cast{M}{G \Rightarrow^{\ell_1} \Unk \Rightarrow^{\ell_2} \agda{gnd} \app B \Rightarrow^{\ell_2} B}   & \text{otherwise}
  \end{cases}
\]

The definition of \agda{applyCast} for $\lambda\textup{\textsf{B}}$
Variant 2 has two additional cases,
for $\mathtt{Act}{\times}$ and $\mathtt{Act}{+}$ respectively:

\[
  \agda{applyCast} : \forall \Gamma A B.\, (M : \Gamma \vdash A) \to \agda{Value}\app M\to (c : \Cast{A}{B}) \to \agda{Active}\app c \to \Gamma \vdash B
\]
\[
\begin{array}{lllllll}
  \agda{applyCast} & M & v & a \Rightarrow^\ell a & \mathtt{ActId} &=&
     M \\
  \agda{applyCast} & M & v & A \Rightarrow^\ell \Unk & \mathtt{ActInj} &=&
  \cast{M}{A \Rightarrow^\ell \agda{gnd}\app A \Rightarrow^\ell \Unk} \\
  \agda{applyCast} & M & v & c & \mathtt{Act}{\times} &=&
    \agda{eta}{\times}\app M\app c \app \mathtt{Cross}\times\\
  \agda{applyCast} & M & v & c & \mathtt{Act}{+} &=& 
    \agda{eta}{+}\app M\app c \app \mathtt{Cross}+\\
  \agda{applyCast} & \cast{M}{G \Rightarrow^{\ell_1} \Unk} & v & \Unk \Rightarrow^{\ell_2} B & \mathtt{ActProj} &=& M' 
\end{array}
\]
where
\[
  M' =
  \begin{cases}
    M  & \text{if } B = \agda{gnd} \app B = G \\
    \blame{\ell_2}{}   & \text{if } B = \agda{gnd} \app B \neq G\\
    \cast{M}{G \Rightarrow^{\ell_1} \Unk \Rightarrow^{\ell_2} \agda{gnd} \app B \Rightarrow^{\ell_2} B}   & \text{otherwise}
  \end{cases}
\]

\begin{proposition}
  Both variants of $\lambda\textup{\textsf{B}}$ are instances
  of the $\agda{CastStruct}$ structure.
\end{proposition}

We import and instantiate the reduction semantics and proof of type
safety from Sections~\ref{sec:dynamic-semantics-CC} and
\ref{sec:type-safety-CC} to obtain the following definition and
results.

\begin{definition}[Reduction for $\lambda\textup{\textsf{B}}$]
  The reduction relation $M \longrightarrow N$ for
  $\lambda\textup{\textsf{B}}$ is the reduction relation of
  $\CC(\CastOp)$ instantiated with $\lambda\textup{\textsf{B}}$'s
  instance of \agda{CastStruct}.
\end{definition}

\begin{corollary}[Preservation for $\lambda\textup{\textsf{B}}$]
  \label{thm:lambda-B-cast-preservation}
  If $\Gamma \vdash M : A$  and $M \longrightarrow M'$,
  then $\Gamma \vdash M' : A$.
\end{corollary}

\begin{corollary}[Progress for $\lambda\textup{\textsf{B}}$]
  \label{thm:lambda-B-cast-progress}
  If $\emptyset \vdash M : A$, then 
  \begin{enumerate}
  \item $M \longrightarrow M'$ for some $M'$,
  \item $\val{M}$, or
  \item $M \equiv \blame{\ell}{}$.
  \end{enumerate}
\end{corollary}

Let $\lambda\textup{\textsf{B}}'$ be the variant of
$\lambda\textup{\textsf{B}}$ obtained by instantiating $\CC'$ instead
of $\CC$.

\subsubsection{Blame-Subtyping}

$\lambda\textup{\textsf{B}}'$ uses the ``UD'' blame-tracking
strategy~\citep{Wadler:2009qv,Siek:2009rt}, so the subtyping relation
that corresponds to safe casts is the one in
Figure~\ref{fig:subtype-UD}. In this subtyping relation, a type $A$ is
a subtype of $\Unk$ only if it is a subtype of a ground type.  For
example, $\Int \to \Int \not<: \Unk$ because $\Int \to \Int \not<:
\Unk \to \Unk$ because $\Unk \not<: \Int$.

The \agda{CastBlameSafe} predicate is then defined as it was for EDA
in Section~\ref{sec:EDA-blame-subtyping}, except using the subtyping
relation of Figure~\ref{fig:subtype-UD}.

\begin{figure}[t]
  \begin{gather*}
    \inference{}{\Unk <: \Unk}
    \quad
    \inference{}{b <: b}
    \quad
    \inference{A <: G}{A <: \Unk}
    \quad
    \inference{C <: A & B <: D}{A \to B <: C \to D}
    \\[1ex]
    \inference{A <: C & B <: D}{A \times B <: C \times D}
    \quad
    \inference{A <: C & B <: D}{A + B <: C + D} 
  \end{gather*}
  \caption{Subtyping for ``UD'' blame-tracking}
  \label{fig:subtype-UD}
\end{figure}

The cast operators for $\lambda\textup{\textsf{B}}'$ are defined
exactly the same as for EDA, so
Lemma~\ref{lem:blame-safety-cast-operators-EDA} also applies to
$\lambda\textup{\textsf{B}}'$, that is, the cast operators such as
\agda{dom} preserve blame safety.

\begin{proposition}
  \label{prop:lambda-B-preblamesafe}
  $\lambda\textup{\textsf{B}}'$ is an instance of the
  \agda{PreBlameSafe} structure.
\end{proposition}

\begin{lemma}[\agda{applyCast} preserves blame safety in both variants of $\lambda\textup{\textsf{B}}'$]\ \\
  \label{lem:applyCast-safe-lambda-B}
  If $\agda{CastBlameSafe}\app c \app \ell$ and $a : \act{c}$
  and $\agda{CastsAllSafe}\app V \app \ell$
  and $v : \val{V}$, then
  $\agda{CastsAllSafe}\app (\agda{applyCast}\app V \app v \app c \app a) \app \ell$.
\end{lemma}

\begin{proposition}
  Both variants of $\lambda\textup{\textsf{B}}'$ are an instance of
  the \agda{BlameSafe} structure.
\end{proposition}

We instantiate Theorem~\ref{thm:blame-subtyping} with this
\agda{BlameSafe} instance to obtain the Blame-Subtyping Theorem for
both variants of $\lambda\textup{\textsf{B}}'$.

\begin{corollary}[Blame-Subtyping Theorem for both variants of $\lambda\textup{\textsf{B}}'$]\ \\
  \label{thm:blame-subtyping-EDA}
  If $M : \Gamma \vdash A$ and $\agda{CastsAllSafe} \app M \app \ell$,
  then $\neg (M \longrightarrow^{*} \key{blame} \, \ell)$.
\end{corollary}

\subsubsection{Dynamic Gradual Guarantee}
\label{sec:gradual-guarantee-lambda-B}

We define the precision relations between two casts and
between a cast and a type,
corresponding to the first three fields in \agda{PreCastStructWithPrecision}
(Figure~\ref{fig:PreCastStructWithPrecision}).
\\[1ex]

\fbox{$\leprecii{i}{i'}$. Variant 1}
\begin{gather*}
\inference{i : \agda{Inert} \app G \Rightarrow^\ell \Unk & i' : \agda{Inert} \app G \Rightarrow^{\ell'} \Unk}
          {\leprecii{i}{i'}}
\;
\inference{A \to B \leprecop A' \to B' & C \to D \leprecop C' \to D' \\
           i : \agda{Inert} \app A \to B \Rightarrow^\ell C \to D \\
           i' : \agda{Inert} \app A' \to B' \Rightarrow^{\ell'} C' \to D'}
          {\leprecii{i}{i'}}
\\[1ex]
\inference{A \times B \leprecop A' \times B' & C \times D \leprecop C' \times D' \\
           i : \agda{Inert} \app A \times B \Rightarrow^\ell C \times D \\
           i' : \agda{Inert} \app A' \times B' \Rightarrow^{\ell'} C' \times D'}
          {\leprecii{i}{i'}}
\;
\inference{A + B \leprecop A' + B' & C + D \leprecop C' + D' \\
           i : \agda{Inert} \app A + B \Rightarrow^\ell C + D \\
           i' : \agda{Inert} \app A' + B' \Rightarrow^{\ell'} C' + D'}
{\leprecii{i}{i'}}
\end{gather*}

\fbox{$\leprecit{i}{A'}$. Variant 1}
\begin{gather*}
\inference{i : \agda{Inert} \app G \Rightarrow^\ell \Unk & G \leprecop A'}
          {\leprecit{i}{A'}}
\;
\inference{A \to B \leprecop A' \to B' & C \to D \leprecop A' \to B' \\
           i : \agda{Inert} \app A \to B \Rightarrow^\ell C \to D }
          {\leprecit{i}{A' \to B'}}
\\[1ex]
\inference{A \times B \leprecop A' \times B' & C \times D \leprecop A' \times B' \\
           i : \agda{Inert} \app A \times B \Rightarrow^\ell C \times D }
          {\leprecit{i}{A' \times B'}}
\;
\inference{A + B \leprecop A' + B' & C + D \leprecop A' + B' \\
           i : \agda{Inert} \app A + B \Rightarrow^\ell C + D }
          {\leprecit{i}{A' + B'}}
\end{gather*}

\fbox{$\leprecti{A}{i'}$. Variant 1}
\begin{gather*}
\inference{A \to B \leprecop A' \to B' & A \to B \leprecop C' \to D' \\
           i' : \agda{Inert} \app A' \to B' \Rightarrow^\ell C' \to D' }
          {\leprecti{A \to B}{i'}}
\;
\inference{A \times B \leprecop A' \times B' & A \times B \leprecop C' \times D' \\
           i' : \agda{Inert} \app A' \times B' \Rightarrow^\ell C' \times D' }
          {\leprecit{A \times B}{i'}}
\\[1ex]
\inference{A + B \leprecop A' + B' & A + B \leprecop C' + D' \\
           i' : \agda{Inert} \app A' + B' \Rightarrow^\ell C' + D' }
          {\leprecit{A + B}{i'}}
\end{gather*}

The precision relations for Variant 2 have fewer cases compared to
Variant 1, since the casts between product types and sum types are active:
\\[1ex]

\fbox{$\leprecii{i}{i'}$. Variant 2}
\begin{gather*}
\inference{i : \agda{Inert} \app G \Rightarrow^\ell \Unk & i' : \agda{Inert} \app G \Rightarrow^{\ell'} \Unk}
          {\leprecii{i}{i'}}
\;
\inference{A \to B \leprecop A' \to B' & C \to D \leprecop C' \to D' \\
           i : \agda{Inert} \app A \to B \Rightarrow^\ell C \to D \\
           i' : \agda{Inert} \app A' \to B' \Rightarrow^{\ell'} C' \to D'}
          {\leprecii{i}{i'}}
\end{gather*}

\fbox{$\leprecit{i}{A'}$. Variant 2}
\begin{gather*}
\inference{i : \agda{Inert} \app G \Rightarrow^\ell \Unk & G \leprecop A'}
          {\leprecit{i}{A'}}
\quad
\inference{A \to B \leprecop A' \to B' & C \to D \leprecop A' \to B' \\
           i : \agda{Inert} \app A \to B \Rightarrow^\ell C \to D }
          {\leprecit{i}{A' \to B'}}
\end{gather*}

\fbox{$\leprecti{A}{i'}$. Variant 2}
\begin{gather*}
\inference{A \to B \leprecop A' \to B' & A \to B \leprecop C' \to D' \\
           i' : \agda{Inert} \app A' \to B' \Rightarrow^\ell C' \to D' }
          {\leprecti{A \to B}{i'}}
\end{gather*}

Then we instantiate and prove the four lemmas that correspond
to the last four fields in \agda{PreCastStructWithPrecision}
(Figure~\ref{fig:PreCastStructWithPrecision}).
They are forward direction and inversion lemmas about
the precision relations defined above between casts and
between a cast and a type.

\begin{lemma}[Type Precision Implies Cast-Type Precision]
  \label{lem:lp-lpit-lambda-B}
  Suppose $c : \Cast{A}{B}$ and $i : \agda{Inert} \app c$.
  If $A \leprecop A'$, $B \leprecop A'$, then $\leprecit{i}{A'}$.
\end{lemma}

\begin{lemma}[Cast-Cast Precision Implies Type Precision]
  \label{lem:lpii-lp-lambda-B}
  Suppose $c : \Cast{A}{B}$, $c' : \Cast{A'}{B'}$,
  $i : \agda{Inert} \app c$, $i' : \agda{Inert} \app c'$.
  If $\leprecii{i}{i'}$, then $A \leprecop A'$ and $B \leprecop B'$.
\end{lemma}

\begin{lemma}[Cast-Type Precision Implies Type Precision]
  \label{lem:lpit-lp-lambda-B}
  Suppose $c : \Cast{A}{B}$ and $i : \agda{Inert} \app c$.
  If $\leprecit{i}{A'}$, then $A \leprecop A'$ and $B \leprecop A'$.
\end{lemma}

\begin{lemma}[Type-Cast Precision Implies Type Precision]
  \label{lem:lpti-lp-lambda-B}
  Suppose $c' : \Cast{A'}{B'}$ and $i' : \agda{Inert} \app c'$.
  If $\leprecti{A}{i'}$, then $A \leprecop A'$ and $A \leprecop B'$.
\end{lemma}

We instantiate \agda{PreCastStructWithPrecision} using
the definitions of precision and their related lemmas
(Lemma~\ref{lem:lp-lpit-lambda-B}, Lemma~\ref{lem:lpii-lp-lambda-B},
Lemma~\ref{lem:lpit-lp-lambda-B}, and Lemma~\ref{lem:lpti-lp-lambda-B}).

\begin{proposition}
  \label{prop:lambda-B-preprec}
  $\lambda\textup{\textsf{B}}'$ is an instance of the
  \agda{PreCastStructWithPrecision} structure.
\end{proposition}

We instantiate Lemma~\ref{lem:compile-pres-prec} to prove that
compilation of the GTLC into $\lambda\textup{\textsf{B}}'$ preserves
the precision relation.

\begin{corollary}[Compilation into $\lambda\textup{\textsf{B}}'$ Preserves Precision]
  \label{lem:compile-pres-prec-lambda-B}
  Suppose $M : \Gamma \vdash_G A$ and $M' : \Gamma' \vdash_G A'$.
  If $\leprec{\Gamma}{\Gamma'}$ and $M \leprecgtlc M'$,
  then $\compilenew{}{M} \lepreccc \compilenew{}{M'}$ and $\leprec{A}{A'}$.
\end{corollary}

\begin{lemma}[\agda{applyCast} Catches Up to the Right]
  \label{lem:applyCast-catchup-lambda-B}
  Suppose $V' : \Gamma' \vdash A'$, $c : \Cast{A}{B}$,
  and $a : \agda{Active} \app c$.
  If $A \leprecop A'$, $B \leprecop A'$, and $V \lepreccc V'$,
  then $\agda{applyCast} \app V \app c \app a \longrightarrow^{*} W$
  and $W \lepreccc V'$ for some value $W$.
\end{lemma}
\begin{proof}[Proof sketch]
  We briefly describe the proof for Variant 1, since the proof
  methodology is the same for both variants.

  By induction on the premise $\agda{Active} \app c$, it generates
  three cases:
  \begin{description}
  \item[$\mathtt{ActId}$.]
    Since $c$ is identity cast, $W \equiv V$ satisfies
    the reduction and $W \lepreccc V'$.
  \item[$\mathtt{ActInj}$.]
    We follow the branch structure of \agda{applyCast} and proceed.
  \item[$\mathtt{ActProj}$.]
    We follow the branch structure of \agda{applyCast}
    and case on whether the target type $B$ of the projection
    is ground. If $B$ is ground, then the proof is straightforward
    by inversion on the canonical form of the projected value.
    Otherwise, if the $B$ is not ground, \agda{applyCast}
    expands the cast by routing through a ground type,
    from where we proceed using the induction hypothesis.
  \end{description}
\end{proof}

\begin{lemma}[Simulation Between Cast and \agda{applyCast}]
  \label{lem:sim-cast-lambda-B}
  Suppose $c : \Cast{A}{B}$, $c' : \Cast{A'}{B'}$,
  and $a' : \agda{Active} \app c'$.
  If $A \leprecop A'$, $B \leprecop B'$, and $V \lepreccc V'$,
  then $\cast{V}{c} \longrightarrow^{*} N$ and
  $N \lepreccc \agda{applyCast} \app V' \app c' \app a'$ for some $N$.
\end{lemma}
\begin{proof}[Proof sketch]
  We case on $\agda{Active} \app c'$ and generate three cases
  that are all straightforward.
\end{proof}

\begin{lemma}[Simulation Between Cast and \agda{Wrap}]
  \label{lem:sim-wrap-lambda-B}
  Suppose $c : \Cast{A}{B}$, $c' : \Cast{A'}{B'}$,
  and $i' : \agda{Inert} \app c'$.
  If $A \leprecop A'$, $B \leprecop B'$, and $V \lepreccc V'$,
  then $\cast{V}{c} \longrightarrow^{*} N$ and
  $N \lepreccc \castInert{V'}{i'}$ for some $N$.
\end{lemma}
\begin{proof}[Proof sketch]
  We case on the premise $\agda{Inert} \app c'$:
\item[$\mathtt{InInj}$.]
  In this case, by inversion on $A \leprecop A'$ and $B \leprecop B'$,
  $A$ can be either $\Unk$ or ground. If $A$ is $\Unk$, an identity cast
  $\Cast{\Unk}{\Unk}$ is active so we use rule $\mathtt{cast}$ and
  proceed; otherwise, the cast is inert so we use rule $\mathtt{wrap}$
  and proceed.
\item[$\mathtt{InCross}$.]
  Similar to the $\mathtt{InInj}$ case, by inversion
  on the two type precision relations.
\end{proof}

\begin{lemma}[Simulation Between Value and \agda{applyCast}]
  \label{lem:castr-cast-lambda-B}
  Suppose $V : \emptyset \vdash A$, $c' : \Cast{A'}{B'}$,
  and $a' : \agda{Active} \app c'$.
  If $A \leprecop A'$, $A \leprecop B'$, and $V \lepreccc V'$,
  then $V \lepreccc \agda{applyCast} \app V' \app c' \app a'$.
\end{lemma}
\begin{proof}[Proof sketch]
  Straightforward. By case analysis on $a' : \agda{Active} \app c'$.
\end{proof}

\begin{lemma}[Simulation Between Value and \agda{Wrap}]
  \label{lem:castr-wrap-lambda-B}
  Suppose $V : \emptyset \vdash A$, $c' : \Cast{A'}{B'}$,
  and $i' : \agda{Inert} \app c'$.
  If $A \leprecop A'$, $A \leprecop B'$, and $V \lepreccc V'$,
  then $V \lepreccc \castInert{V'}{i'}$.
\end{lemma}
\begin{proof}[Proof sketch]
  Straightforward. By case analysis on $i' : \agda{Inert} \app c'$
  and by inversion on the type precision relations.
\end{proof}

We prove that both variants of $\lambda\textup{\textsf{B}}'$
are instances of \agda{CastStructWithPrecision} using
Lemma~\ref{lem:applyCast-catchup-lambda-B}, Lemma~\ref{lem:sim-cast-lambda-B},
Lemma~\ref{lem:sim-wrap-lambda-B}, Lemma~\ref{lem:castr-cast-lambda-B},
and Lemma~\ref{lem:castr-wrap-lambda-B}.

\begin{proposition}
  \label{prop:lambda-B-csprec}
  Both variants of $\lambda\textup{\textsf{B}}'$ are an instance of
  the \agda{CastStructWithPrecision} structure.
\end{proposition}

We instantiate Lemma~\ref{lem:simulation-more-precise} (Simulation of
More Precise Programs) with the $\lambda\textup{\textsf{B}}'$
instances of \agda{CastStructWithPrecision}
(Proposition~\ref{prop:lambda-B-csprec}) to obtain the main lemma of
the dynamic gradual guarantee for both variants of
$\lambda\textup{\textsf{B}}'$.

\begin{corollary}[Simulation of More Precise Programs for $\lambda\textup{\textsf{B}}'$]
  \label{lem:simulation-more-precise-lambda-B}
  Suppose $M_1 : \emptyset \vdash A$ and $M_1' : \emptyset \vdash A'$
  and $M_2' : \emptyset \vdash A'$. If $M_1 \lepreccc M_1'$ and $M_1'
  \longrightarrow M_2'$, then $M_1 \longrightarrow^{*} M_2$ and $M_2
  \lepreccc M_2'$ for some $M_2$.
\end{corollary}

We prove the dynamic gradual guarantee for
$\lambda\textup{\textsf{B}}'$ following the reasoning of
\citet{Siek:2015ac}. We give the full proof here.

\begin{theorem}[Dynamic Gradual Guarantee for both variants of $\lambda\textup{\textsf{B}}'$]\ \\
  Suppose $M \leprecgtlc N$ and $M : \Gamma \vdash_G A$
  and $N : \Gamma \vdash_G B$.
  \begin{enumerate}
  \item If $\compile{}{N} \longrightarrow^{*} W$, then
    $\compile{}{M} \longrightarrow^{*} V$ and $V \lepreccc W$.

  \item If $\compile{}{N}$ diverges, so does $\compile{}{M}$.
    
  \item If $\compile{}{M} \longrightarrow^{*} V$, then either
    $\compile{}{N} \longrightarrow^{*} W$ and $V \lepreccc W$, or
    $\compile{}{N} \longrightarrow^{*} \blame{\ell}{}$ for some $\ell$.

  \item If $\compile{}{M}$ diverges then either $\compile{}{N}$
    diverges or $\compile{}{N} \longrightarrow^{*} \blame{\ell}{}$ for
    some $\ell$.
  \end{enumerate}
\end{theorem}
\begin{proof}\ 
  \begin{enumerate}
  \item By Corollary~\ref{lem:compile-pres-prec-lambda-B} we have
    $\compile{}{M} \lepreccc \compile{}{N}$. Then by induction on
    $\compile{}{N} \longrightarrow^{*} W$ and
    Corollary~\ref{lem:simulation-more-precise-lambda-B}, we have
    $\compile{}{M} \longrightarrow^{*} M'$ and $M' \lepreccc W$ for
    some $M'$.  Finally, by Lemma~\ref{prop:catchup} (instantiated for
    $\lambda\textup{\textsf{B}}'$), we have $M' \longrightarrow^{*} V$
    and $V \lepreccc W$.
  \item By Corollary~\ref{lem:compile-pres-prec-lambda-B} we have
    $\compile{}{M} \lepreccc \compile{}{N}$. Then by 
    Corollary~\ref{lem:simulation-more-precise-lambda-B},
    $\compile{}{M}$ also diverges.
  \item Because $\lambda\textup{\textsf{B}}'$ is type safe
    (Corollaries~\ref{thm:lambda-B-cast-preservation} and
    \ref{thm:lambda-B-cast-progress}), we have the following
    cases.
    \begin{itemize}
    \item Case $\compile{}{N} \longrightarrow^{*} W$ for some $W$.
      Then by Part 1 of this theorem and because reduction is
      deterministic, $V \lepreccc W$.
    \item Case $\compile{}{N} \longrightarrow^{*} \blame{\ell}{}$ for some $\ell$.
      We immediately conclude.
    \item Case $\compile{}{N}$ diverges. Then by Part 1, $\compile{}{M}$
      also diverges, but that contradicts the assumption that
      $\compile{}{M} \longrightarrow^{*} V$.
    \end{itemize}
    
  \item Again because $\lambda\textup{\textsf{B}}'$ is type safe
    (Corollaries~\ref{thm:lambda-B-cast-preservation} and
    \ref{thm:lambda-B-cast-progress}), we have the following
    cases.
    \begin{itemize}
    \item Case $\compile{}{N} \longrightarrow^{*} W$ for some $W$.
      Then by Part 1, $\compile{}{M} \longrightarrow^{*} V$ for some
      $V$ but that contradicts the assumption that $\compile{}{M}$
      diverges.
    \item Case $\compile{}{N} \longrightarrow^{*} \blame{\ell}{}$ for some $\ell$.
      We immediately conclude.
    \item Case $\compile{}{N}$ diverges. We immediately conclude.
    \end{itemize}
\end{enumerate}
\end{proof}


\subsection{Partially-Eager ``D'' Coercions (EDC)}
\label{sec:EDC}

The next three cast calculi use cast representation types based on the
Coercion Calculus of \citet{Henglein:1994nz}. We start with one that
provides the same behavior as the cast calculus of
\citet{Siek:2006bh}, that is, partially-eager casts with active cross
casts (Section~\ref{sec:EDA}). We use the abbreviation EDC this cast
calculus.

We define coercions as follows, omitting sequence coercions because
they are not necessary in this calculus. \\[1ex]
\fbox{$c,d : \Cast{A}{B}$}
\begin{gather*}
\inference
    {}
    {\agda{id} : \Cast{a}{a}}
\quad
\inference
    {A \neq \Unk}
    {A! : \Cast{A}{\Unk}} 
\quad
\inference
    {B \neq \Unk}
    {B?^\ell : \Cast{\Unk}{B}} 
\\[2ex]
\inference
    {c : \Cast{C}{A} & d : \Cast{B}{D}}
    {c \to d : \Cast{(A \to B)}{(C \to D)}}
\quad
\inference
    {c : \Cast{A}{C} & d : \Cast{B}{D}}
    {c \otimes d : \Cast{(A \otimes B)}{(C \otimes D)}} \otimes \in \{ \times, + \}
\end{gather*}

The cast constructor is defined by applying the \agda{coerce} function
(defined later in this section) to the implicit proof of consistency
between $A$ and $B$.
\[
  \coerce{A}{B}{\ell} \app \{ p : A \sim B\} = \agda{coerce} \app p \app \ell
\]

\subsubsection{Reduction Semantics and Type Safety}

Injections are categorized as inert casts.\\[1ex]
\fbox{$\inert{c}$}
\[
\inference{}
          {\inert{A!}}
\]
The coercions between function, pair, and sum types are categorized as
cross casts.\\[1ex]
\fbox{$\cross{c}$}
\begin{gather*}
  \inference{}
            {\agda{Cross}{\otimes} : \cross{(c \otimes d)}} \otimes \in \{ \to, \times, + \}
\end{gather*}

We categorize the identity, projection, and cross casts as active.\\[1ex]
\fbox{$\act{c}$}
\begin{gather*}
  \inference
      {}
      {\mathtt{ActId} :\act{\agda{id}}}
  \quad
  \inference
    {}
    {\mathtt{ActProj} : \act{A?^\ell}}
  \quad
  \inference
    {\cross{c}}
    {\mathtt{ActCross} :\act{c}}
\end{gather*}

\begin{lemma}
  \label{lem:EDC-active-or-inert}
  For any types $A$ and $B$, $c : \Cast{A}{B}$ is either an active or
  inert cast.
\end{lemma}

\begin{lemma}\label{lem:EDC-inert-cross}
  If $c : \Cast{A}{(B \otimes C)}$ and $\inert{c}$,
  then $\cross{c}$ and $A \equiv D \otimes E$ for some $D$ and $E$.
\end{lemma}

The cast operators \agda{dom}, \agda{cod}, etc. are defined as
follows.
\begin{align*}
  \dom{(c \to d) \app x} &= c \\
  \cod{(c \to d) \app x} &= d \\
  \agda{fst}\app (c \times d) \app x &= c \\
  \agda{snd} \app (c \times d) \app x &= d \\
  \agda{inl} \app (c + d) \app x &= c \\
  \agda{inr} \app (c + d) \app x &= d
\end{align*}

\begin{lemma}\label{lem:EDC-base-not-inert}
   A cast $c : \Cast{A}{b}$ is not inert.
\end{lemma}

\begin{proposition}\label{prop:simple-active-pre-cast-struct}
The EDC calculus is instance of the $\agda{PreCastStruct}$ structure.
\end{proposition}

To help define the \agda{applyCast} function, we define an auxiliary
function named \agda{coerce} for converting two consistent types and a
blame label into a coercion.  (The \agda{coerce} function is also
necessary for compiling from the GTLC to this calculus.)

\[
  \agda{coerce} : \forall A\,B.\, A \sim B \to \agda{Label} \to \Cast{A}{B}
\]
\begin{align*}
  \agda{coerce} \app \agda{UnkL}{\sim}[B] \app \ell &=
    \begin{cases}
      \agda{id} & \text{if } B \equiv \Unk \\
      B?^\ell & B \not\equiv \Unk
    \end{cases}\\
  \agda{coerce} \app \agda{UnkR}{\sim}[A] \app \ell &=
    \begin{cases}
      \agda{id} & \text{if } A \equiv \Unk \\
      A! & A \not\equiv \Unk
    \end{cases} \\
  \agda{coerce} \app \mathtt{Base}{\sim}[\Base] \app \ell &= \agda{id}\\
  \agda{coerce} \app (\mathtt{Fun}{\sim}\app d_1\app d_2) \app \ell &=
        (\agda{coerce} \app d_1 \app \overline{\ell})
    \to (\agda{coerce} \app d_2 \app \ell)\\
  \agda{coerce} \app (\mathtt{Pair}{\sim}\app d_1\app d_2) \app \ell &=
        (\agda{coerce} \app d_1 \app \ell)
    \times (\agda{coerce} \app d_2 \app \ell)\\
  \agda{coerce} \app (\mathtt{Sum}{\sim}\app d_1\app d_2) \app \ell &=
        (\agda{coerce} \app d_1 \app \ell)
    + (\agda{coerce} \app d_2 \app \ell)
\end{align*}

The structure of coercions is quite similar to that of the active
casts but a bit more convenient to work with, so we define the
\agda{applyCast} function by cases on the coercion. We omit the case
for injection because that coercion is inert.
\[
  \agda{applyCast} : \forall \Gamma A B.\, (M : \Gamma \vdash A) \to \agda{Value}\app M\to (c : \Cast{A}{B}) \to \agda{Active}\app c \to \Gamma \vdash B
\]
\[
\begin{array}{lllllll}
  \agda{applyCast} & M & v & \agda{id} & a &=&
     M \\
  \agda{applyCast} & \cast{M}{A!} & v & B?^{\ell} & a &=& 
  \begin{cases}
    \cast{M}{\agda{coerce} \app ab \app \ell} & \text{if } ab : A \sim B\\
    \blame{\ell}{} & \text{otherwise}
  \end{cases}\\
  \agda{applyCast} &M& v & c \otimes d & a &=&
    \agda{eta}{\otimes}\app M\app c \app \mathtt{Cross}\otimes
\end{array}
\]

\begin{proposition}
The EDC calculus is an instance of the $\agda{CastStruct}$ structure.
\end{proposition}

We import and instantiate the reduction semantics and proof of type
safety from Sections~\ref{sec:dynamic-semantics-CC} and
\ref{sec:type-safety-CC}.

\begin{definition}[Reduction for EDC]
  The reduction relation $M \longrightarrow N$ for the EDC calculus
  is the reduction relation of $\CC(\CastOp)$ instantiated with EDC's
  instance of \agda{CastStruct}.
\end{definition}

\begin{corollary}[Preservation for EDC]
  \label{thm:pedi-cast-preservation}
  If $\Gamma \vdash M : A$  and $M \longrightarrow M'$,
  then $\Gamma \vdash M' : A$.
\end{corollary}

\begin{corollary}[Progress for EDC]\label{thm:pedi-cast-progress}
  If $\emptyset \vdash M : A$, then 
  \begin{enumerate}
  \item $M \longrightarrow M'$ for some $M'$,
  \item $\val{M}$, or
  \item $M \equiv \blame{\ell}{}$.
  \end{enumerate}
\end{corollary}

Let EDC$'$ be the variant of EDC obtained by instantiating $\CC'$
instead of $\CC$.

\subsubsection{Blame-Subtyping}

The \agda{CastBlameSafe} predicate for EDC$'$ is defined in
Figure~\ref{fig:CastBlameSafe-EDC}.

\begin{figure}[tb]
  \begin{gather*}
    \inference{}{\agda{CastBlameSafe}\app \agda{id}\app \ell}
    \inference{}{\agda{CastBlameSafe}\app A!\app \ell}
    \inference{\ell\neq \ell'}{\agda{CastBlameSafe}\app B?^{\ell'}\app \ell}
    \\[2ex]
    \inference{\agda{CastBlameSafe}\app c \app \ell &
               \agda{CastBlameSafe}\app d \app \ell}
              {\agda{CastBlameSafe}\app (c \otimes d) \app \ell}
    \otimes \in \{ \to , \times, + \}
\end{gather*}
  \caption{Definition of the \agda{CastBlameSafe} predicate for EDC$'$.}
  \label{fig:CastBlameSafe-EDC}
\end{figure}

\begin{lemma}[Blame Safety of Cast Operators]
  \label{lem:blame-safety-cast-operators-EDC}
  Suppose $\agda{CastBlameSafe}\app c \app \ell$ and $c$ is a cross
  cast, that is, $x : \cross{c}$.
  \begin{itemize}
  \item If $x=\mathtt{Cross}\to$,
    then $\agda{CastBlameSafe}\app (\agda{dom}\app c \app x)\app \ell$
    and $\agda{CastBlameSafe}\app (\agda{cod}\app c \app x)\app \ell$.
  \item If $x=\mathtt{Cross}\times$,
    then $\agda{CastBlameSafe}\app (\agda{fst}\app c \app x)\app \ell$
    and $\agda{CastBlameSafe}\app (\agda{snd}\app c \app x)\app \ell$.
  \item If $x=\mathtt{Cross}+$,
    then $\agda{CastBlameSafe}\app (\agda{inl}\app c \app x)\app \ell$
    and $\agda{CastBlameSafe}\app (\agda{inr}\app c \app x)\app \ell$.
  \end{itemize}
\end{lemma}
\begin{proof}
  By inversion on $\cross{c}$ and the \agda{CastBlameSafe} predicate.
\end{proof}

\begin{proposition}
  \label{prop:EDC-preblamesafe}
  EDC$'$ is an instance of the \agda{PreBlameSafe} structure.
\end{proposition}

\begin{lemma}[\agda{coerce} is blame safe]
  \label{lem:coerce-safe-EDC}
  Suppose $\ell \neq \ell'$. If $ab : A \sim B$, \\
  then $\agda{CastBlameSafe} \app (\agda{coerce}\app ab \app \ell') \app \ell$.
\end{lemma}

\begin{lemma}[\agda{applyCast} preserves blame safety in EDC$'$]
  If $\agda{CastBlameSafe}\app c \app \ell$ and $a : \act{c}$
  and $\agda{CastsAllSafe}\app V \app \ell$
  and $v : \val{V}$, then
  $\agda{CastsAllSafe}\app (\agda{applyCast}\app V \app v \app c \app a) \app \ell$.
\end{lemma}
\begin{proof}
    The proof is by cases on $\act{c}$ and inversion on
    $\agda{CastBlameSafe}\app c \app \ell$, using
    Lemma~\ref{lem:coerce-safe-EDC} in the case for projection.
\end{proof}

\begin{proposition}
  EDC$'$ is an instance of the \agda{BlameSafe} structure.
\end{proposition}

We instantiate Theorem~\ref{thm:blame-subtyping} with this
\agda{BlameSafe} instance to obtain the Blame-Subtyping Theorem for
EDC$'$.

\begin{corollary}[Blame-Subtyping Theorem for EDC$'$]
  \label{thm:blame-subtyping-EDC}
  If $M : \Gamma \vdash A$ and $\agda{CastsAllSafe} \app M \app \ell$,
  then $\neg (M \longrightarrow^{*} \key{blame} \, \ell)$.
\end{corollary}


\subsection{Lazy ``D'' Coercions (LDC)}
\label{sec:LazyCoercions}

The Lazy ``D'' Coercions~\citep{Siek:2009rt} are syntactically similar
to the coercions of Section~\ref{sec:EDC} except that include a
failure coercion, written $\cfail{\ell}$. \\[1ex]
\fbox{$c,d : \Cast{A}{B}$}
\begin{gather*}
\inference
    {}
    {\cfail{\ell} : \Cast{A}{B}}
\quad
\inference
    {}
    {\agda{id} : \Cast{a}{a}}
\quad
\inference
    {A \neq \Unk}
    {A! : \Cast{A}{\Unk}} 
\quad
\inference
    {B \neq \Unk}
    {B?^\ell : \Cast{\Unk}{B}} 
\\[2ex]
\inference
    {c : \Cast{C}{A} & d : \Cast{B}{D}}
    {c \to d : \Cast{(A \to B)}{(C \to D)}}
\quad
\inference
    {c : \Cast{A}{C} & d : \Cast{B}{D}}
    {c \otimes d : \Cast{(A \otimes B)}{(C \otimes D)}} \otimes \in \{ \times, + \}
\end{gather*}

The cast constructor is defined as follows, using the \agda{coerce}
function defined later in this section.
\[
  \coerce{A}{B}{\ell} =
  \begin{cases}
    \agda{coerce} \app d \app \ell & \text{if } d : A \smile B \\
    \cfail{\ell} & \text{otherwise}
  \end{cases}
\]

\subsubsection{Reduction Semantics and Type Safety}

Injections are categorized as inert casts.\\[1ex]
\fbox{$\inert{c}$}
\[
\inference{}
          {\inert{A!}}
\]
The coercions between function, pair, and sum types are categorized as
cross casts.\\[1ex]
\fbox{$\cross{c}$}
\begin{gather*}
  \inference{}
            {\agda{Cross}{\otimes} : \cross{c \otimes d}} \otimes \in \{ \to, \times, + \}
\end{gather*}
In addition to the identity and projection coercions and the cross
casts, the failure coercions are also active.\\[1ex]
\fbox{$\act{c}$}
\begin{gather*}
  \inference
      {}
      {\mathtt{ActId} :\act{\agda{id}}}
  \quad
  \inference
    {}
    {\mathtt{ActProj} : \act{A?^\ell}}
  \quad
  \inference
    {\cross{c}}
    {\mathtt{ActCross} :\act{c}}
    \\[1ex]
  \inference
    {}
    {\mathtt{ActFail} :\act{\cfail{\ell}}}
\end{gather*}

\begin{lemma}
  \label{lem:lazyd-coercion-active-or-inert}
  For any types $A$ and $B$, $c : \Cast{A}{B}$ is either an active or
  inert cast.
\end{lemma}

\begin{lemma}\label{lem:lazyd-coercion-inert-cross}
  If $c : \Cast{A}{(B \otimes C)}$ and $\inert{c}$,
  then $\cross{c}$ and $A \equiv D \otimes E$ for some $D$ and $E$.
\end{lemma}

The definition of cast operators such as $\agda{dom}$ are the same as
in Section~\ref{sec:EDC}.


\begin{lemma}\label{lem:lazyd-coercion-base-not-inert}
   A cast $c : \Cast{A}{b}$ is not inert.
\end{lemma}

\begin{proposition}\label{prop:lazyd-active-pre-cast-struct}
The LDC calculus is an instance of the $\agda{PreCastStruct}$
structure.
\end{proposition}

We define shallow consistency, written $A \smile B$, as follows.
\begin{gather*}
\mathtt{UnkL}{\smile}[B] : \Unk \smile B
\qquad
\mathtt{UnkR}{\smile}[B] : A \smile \Unk
\qquad
\mathtt{Base}{\smile}[b]: b \smile b
\\[1ex]
\otimes{\smile}[A,B,C,D] : (A \otimes B) \smile (C \otimes D)
\end{gather*}

The \agda{coerce} function differs from that of Section~\ref{sec:EDC}
in that we only require the source and target types to be shallowly
consistent. The \agda{coerce} function is mutually defined with the
function $\coerce{A}{B}{\ell}$ that checks whether two types are
shallowly consistent, invoking \agda{coerce} if they are and returning
a failure coercion if they are not.
\begin{align*}
  \agda{coerce} &: \forall A\,B.\, A \smile B \to \agda{Label} \to \Cast{A}{B}\\[1ex]
  \agda{coerce} \app \agda{UnkL}{\smile}[B] \app \ell &=
    \begin{cases}
      \agda{id} & \text{if } B \equiv \Unk \\
      B?^\ell & B \not\equiv \Unk
    \end{cases}\\
  \agda{coerce} \app \agda{UnkR}{\smile}[A] \app \ell &=
    \begin{cases}
      \agda{id} & \text{if } A \equiv \Unk \\
      A! & A \not\equiv \Unk
    \end{cases} \\
  \agda{coerce} \app \mathtt{Base}{\smile}[\Base] \app \ell &= \agda{id}\\
  \agda{coerce} \app (\mathtt{Fun}{\smile}[A,B,C,D]) \app \ell &=
        \coerce{C}{A}{\overline{\ell}}
    \to \coerce{B}{D}{\ell} \\
  \agda{coerce} \app (\mathtt{Pair}{\smile}[A,B,C,D]) \app \ell &=
        \coerce{A}{C}{\ell}
    \times \coerce{ B}{ D}{ \ell}\\
  \agda{coerce} \app (\mathtt{Sum}{\smile}[A,B,C,D]) \app \ell &=
        \coerce{ A}{ C}{ \ell}
    + \coerce{ B}{ D}{ \ell} 
\end{align*}

The definition of \agda{applyCast} is similar to the one for EDC
(Section~\ref{sec:EDC}) except that there is an additional case for
$\cfail{\ell}$ and the projection case checks for shallow
consistency instead of consistency, using $\coerce{A}{B}{\ell}$.
\[
  \agda{applyCast} : \forall \Gamma A B.\, (M : \Gamma \vdash A) \to \agda{Value}\app M\to (c : \Cast{A}{B}) \to \agda{Active}\app c \to \Gamma \vdash B
\]
\[
\begin{array}{lllllll}
  \agda{applyCast} & M & v & \agda{id} & a &=&
     M \\
  \agda{applyCast} & \cast{M}{A!} & v & B?^{\ell} & a &=& 
    \cast{M}{\coerce{A}{B}{\ell}} \\
  \agda{applyCast} &M& v & c \otimes d & a &=&
    \agda{eta}{\otimes}\app M\app c \app \mathtt{Cross}\otimes\\
  \agda{applyCast} & M & v & \cfail{\ell} & a &=& \blame{\ell}{}
\end{array}
\]

\begin{proposition}
The LDC calculus is an instance of the $\agda{CastStruct}$ structure.
\end{proposition}

We import and instantiate the reduction semantics and proof of type
safety from Sections~\ref{sec:dynamic-semantics-CC} and
\ref{sec:type-safety-CC}.

\begin{definition}[Reduction for LDC]
  The reduction relation $M \longrightarrow N$ for the LDC calculus
  is the reduction relation of $\CC(\CastOp)$ instantiated with
  LDC's instance of \agda{CastStruct}.
\end{definition}

\begin{corollary}[Preservation for LDC]
  \label{thm:LDC-preservation}
  If $\Gamma \vdash M : A$  and $M \longrightarrow M'$,
  then $\Gamma \vdash M' : A$.
\end{corollary}

\begin{corollary}[Progress for LDC]\label{thm:LDC-progress}
  If $\emptyset \vdash M : A$, then 
  \begin{enumerate}
  \item $M \longrightarrow M'$ for some $M'$,
  \item $\val{M}$, or
  \item $M \equiv \blame{\ell}{}$.
  \end{enumerate}
\end{corollary}

Let LDC$'$ be the variant of LDC obtained by instantiating $\CC'$
instead of $\CC$.

\subsubsection{Blame-Subtyping}

The \agda{CastBlameSafe} predicate for LDC$'$ is the same as for
EDC$'$ (Section~\ref{sec:EDC}) except that there is an additional rule
for failure coercions:
\[
\inference{\ell \neq \ell'}
          {\agda{CastBlameSafe}\app \cfail{\ell'} \app \ell}
\]

\begin{lemma}[Blame Safety of Cast Operators]
  \label{lem:blame-safety-cast-operators-LDC}
  Suppose $\agda{CastBlameSafe}\app c \app \ell$ and $c$ is a cross
  cast, that is, $x : \cross{c}$.
  \begin{itemize}
  \item If $x=\mathtt{Cross}\to$,
    then $\agda{CastBlameSafe}\app (\agda{dom}\app c \app x)\app \ell$
    and $\agda{CastBlameSafe}\app (\agda{cod}\app c \app x)\app \ell$.
  \item If $x=\mathtt{Cross}\times$,
    then $\agda{CastBlameSafe}\app (\agda{fst}\app c \app x)\app \ell$
    and $\agda{CastBlameSafe}\app (\agda{snd}\app c \app x)\app \ell$.
  \item If $x=\mathtt{Cross}+$,
    then $\agda{CastBlameSafe}\app (\agda{inl}\app c \app x)\app \ell$
    and $\agda{CastBlameSafe}\app (\agda{inr}\app c \app x)\app \ell$.
  \end{itemize}
\end{lemma}

\begin{proposition}
  \label{prop:LDC-preblamesafe}
  LDC$'$ is an instance of the \agda{PreBlameSafe} structure.
\end{proposition}

\begin{lemma}[\agda{coerce} is blame safe]
  \label{lem:coerce-safe-LDC}
  Suppose $\ell \neq \ell'$.
  \begin{enumerate}
  \item $\agda{CastBlameSafe} \app \coerce{A}{B}{\ell'} \app \ell$.
  \item If $ab : A \smile B$, then $\agda{CastBlameSafe} \app (\agda{coerce}\app ab \app \ell') \app \ell$.
  \end{enumerate}
\end{lemma}

\begin{lemma}[\agda{applyCast} preserves blame safety in LDC$'$]
  If $\agda{CastBlameSafe}\app c \app \ell$ and $a : \act{c}$
  and $\agda{CastsAllSafe}\app V \app \ell$
  and $v : \val{V}$, then
  $\agda{CastsAllSafe}\app (\agda{applyCast}\app V \app v \app c \app a) \app \ell$.
\end{lemma}
\begin{proof}
    The proof is by cases on $\act{c}$ and inversion on
    $\agda{CastBlameSafe}\app c \app \ell$, using
    Lemma~\ref{lem:coerce-safe-LDC} in the case for projection.
\end{proof}

\begin{proposition}
  LDC$'$ is an instance of the \agda{BlameSafe} structure.
\end{proposition}

We instantiate Theorem~\ref{thm:blame-subtyping} with this
\agda{BlameSafe} instance to obtain the Blame-Subtyping Theorem for
LDC$'$.

\begin{corollary}[Blame-Subtyping Theorem for LDC$'$]
  \label{thm:blame-subtyping-LDC}
  If $M : \Gamma \vdash A$ and $\agda{CastsAllSafe} \app M \app \ell$,
  then $\neg (M \longrightarrow^{*} \key{blame} \, \ell)$.
\end{corollary}


\subsection{The $\lambda\textup{\textsf{C}}$ Coercion Calculus}
\label{sec:lambda-C}

This section instantiates the Parametric Cast Calculus to obtain the
$\lambda\textup{\textsf{C}}$ calculus of \citet{Siek:2015ab}.  Again
we represent casts as coercions, but this time we must include the
notion of sequencing of two coercions, written $\cseq{c}{d}$, to
enable the factoring of casts through the ground type. As part of this
factoring, injections and projections are restricted to ground types.
We omit the failure coercion because it is not necessary for
$\lambda\textup{\textsf{C}}$.

\fbox{$c,d : \Cast{A}{B}$}
\begin{gather*}
\inference
    {}
    {\agda{id} : \Cast{a}{a}}
\quad
\inference
    {}
    {G! : \Cast{G}{\Unk}} 
\quad
\inference
    {}
    {H?^\ell : \Cast{\Unk}{H}} 
\quad
\inference
    {c : \Cast{A}{B} & d : \Cast{B}{C}}
    {\cseq{c}{d}: \Cast{A}{C}}
\\[2ex]
\inference
    {c : \Cast{C}{A} & d : \Cast{B}{D}}
    {c \to d : \Cast{(A \to B)}{(C \to D)}}
\quad
\inference
    {c : \Cast{A}{C} & d : \Cast{B}{D}}
    {c \otimes d : \Cast{(A \otimes B)}{(C \otimes D)}} \otimes \in \{ \times, + \} 
\end{gather*}

The cast constructor is defined as follows.
\begin{align*}
  \coerce{A}{\Unk}{\ell}&= \cseq{\coerce{A}{G}{\ell}}{\coerce{G}{\Unk}{\ell}} \\
  \coerce{\Unk}{A}{\ell}&= \cseq{\coerce{\Unk}{G}{\ell}}{\coerce{G}{A}{\ell}} \\
  \coerce{G}{\Unk}{\ell}&= G!\\
  \coerce{\Unk}{H}{\ell}&= H?^\ell \\
  \coerce{b}{b}{\ell} &= \agda{id}\\
  \coerce{A \to B}{A' \to B'}{\ell}&=
      \coerce{A'}{A}{\ell} \to \coerce{B}{B'}{\ell} \\
  \coerce{A \times B}{A' \times B'}{\ell}&=
      \coerce{A'}{A}{\ell} \times \coerce{B}{B'}{\ell} \\
  \coerce{A + B}{A' + B'}{\ell}&= 
      \coerce{A'}{A}{\ell} + \coerce{B}{B'}{\ell} 
\end{align*}

\subsubsection{Reduction Semantics and Type Safety}

The coercions between function, pair, and sum types are categorized as
cross casts. We do not categorize sequence coercions as cross casts,
which, for example, simplifies the definition of the \agda{dom} and
\agda{cod} functions. \\[1ex]
\fbox{$\cross{c}$}
\begin{gather*}
  \inference{}
            {\key{Cross}{\otimes} : \cross{c \otimes d}} \otimes \in \{ \to, \times, + \}
\end{gather*}

Injections and function coercions are categorized as inert casts.\\[1ex]
\fbox{$\inert{c}$}
\[
\inference{}
          {\key{InInj} : \inert{G!}}
\quad
\inference{}
          {\key{InFun} : \inert{c \to d}}
\]

The active casts in $\lambda\textup{\textsf{C}}$ include identity
casts, projections, and sequences.  The $\lambda\textup{\textsf{C}}$
calculus did not include pairs and sums~\citep{Siek:2015ab}, but here
we choose to categorize casts between pairs and sums as active casts,
as we did for $\lambda\textup{\textsf{B}}$ in
Section~\ref{sec:lambda-B}. \\[1ex]
\fbox{$\act{c}$}
\begin{gather*}
  \inference
      {}
      {\mathtt{ActId} :\act{\agda{id}}}
  \quad
  \inference
    {}
    {\mathtt{ActProj} : \act{H?^\ell}}
  \quad
  \inference
      {}
      {\mathtt{ActSeq} : \act{(\cseq{c}{d})}}
  \\[1ex]
  \inference
    {}
    {\mathtt{Act}{\otimes} :\act{(c \otimes d)}} \otimes \in \{ \times, + \}
\end{gather*}

\begin{lemma}
  \label{lem:lambda-C-active-or-inert}
  For any types $A$ and $B$, $c : \Cast{A}{B}$ is either an active or
  inert cast.
\end{lemma}

\begin{lemma}\label{lem:lambda-C-inert-cross}
  If $c : \Cast{A}{(B \otimes C)}$ and $\inert{c}$,
  then $\cross{c}$ and $A \equiv D \otimes E$ for some $D$ and $E$.
\end{lemma}

The definition of the functions such as $\agda{dom}$ are the usual
ones, but note that the $x$ parameter plays an important role in this
definition. We did not categorize sequence casts as cross casts, so
the following functions can omit the cases for $(\cseq{c}{d})$.

\begin{align*}
  \dom{(c \to d) \app \key{Cross}{\to}} &= c \\
  \cod{(c \to d) \app \key{Cross}{\to}} &= d \\
  \agda{fst}\app (c \times d) \app \key{Cross}{\times} &= c \\
  \agda{snd} \app (c \times d) \app \key{Cross}{\times} &= d \\
  \agda{inl} \app (c + d) \app \key{Cross}{+} &= c \\
  \agda{inr} \app (c + d) \app \key{Cross}{+} &= d
\end{align*}

\begin{lemma}\label{lem:lambda-C-base-not-inert}
   A cast $c : \Cast{A}{b}$ is not inert.
\end{lemma}

\begin{proposition}\label{prop:lambda-C-active-pre-cast-struct}
The $\lambda\textup{\textsf{C}}$ Calculus is an instance of the
$\agda{PreCastStruct}$ structure.
\end{proposition}

We define the \agda{applyCast} function for
$\lambda\textup{\textsf{C}}$ as follows.
\[
  \agda{applyCast} : \forall \Gamma A B.\, (M : \Gamma \vdash A) \to \agda{Value}\app M\to (c : \Cast{A}{B}) \to \agda{Active}\app c \to \Gamma \vdash B
\]
\[
\begin{array}{lllllll}
  \agda{applyCast} & M & v & \agda{id} & a &=& M \\
  \agda{applyCast} & \cast{M}{G!} & v & H?^{\ell} & a &=& 
    \begin{cases}
       M  & \text{if } G \equiv H\\
       \blame{\ell}{} & \text{otherwise}
    \end{cases} \\
  \agda{applyCast} &M& v & \cseq{c}{d} & a &=& \cast{\cast{M}{c}}{d} \\
  \agda{applyCast} &M& v & c \times d & a &=&
    \agda{eta}{\times}\app M\app c \app \mathtt{Cross}\times\\
  \agda{applyCast} &M& v & c + d & a &=&
    \agda{eta}{+}\app M\app c \app \mathtt{Cross}{+}
\end{array}
\]

\begin{proposition}
The $\lambda\textup{\textsf{C}}$ calculus is an instance of the
$\agda{CastStruct}$ structure.
\end{proposition}

We import and instantiate the reduction semantics and proof of type
safety from Sections~\ref{sec:dynamic-semantics-CC} and
\ref{sec:type-safety-CC}.

\begin{definition}[Reduction for $\lambda\textup{\textsf{C}}$]
  The reduction relation $M \longrightarrow N$ for
  $\lambda\textup{\textsf{C}}$ is the reduction relation of
  $\CC(\CastOp)$ instantiated with $\lambda\textup{\textsf{C}}$'s
  instance of \agda{CastStruct}.
\end{definition}

\begin{corollary}[Preservation for $\lambda\textup{\textsf{C}}$]
  \label{thm:lambda-C-preservation}
  If $\Gamma \vdash M : A$  and $M \longrightarrow M'$,
  then $\Gamma \vdash M' : A$.
\end{corollary}

\begin{corollary}[Progress for $\lambda\textup{\textsf{C}}$]
  \label{thm:lambda-C-progress}
  If $\emptyset \vdash M : A$, then 
  \begin{enumerate}
  \item $M \longrightarrow M'$ for some $M'$,
  \item $\val{M}$, or
  \item $M \equiv \blame{\ell}{}$.
  \end{enumerate}
\end{corollary}

Let $\lambda\textup{\textsf{C}}'$ be the variant of
$\lambda\textup{\textsf{C}}$ obtained by instantiating $\CC'$ instead
of $\CC$.

\subsubsection{Blame-Subtyping}

The \agda{CastBlameSafe} predicate for $\lambda\textup{\textsf{C}}'$
is the same as the one for EDC$'$ (Figure~\ref{fig:CastBlameSafe-EDC})
except for the additional rule for sequence coercions:
\[
\inference{\agda{CastBlameSafe}\app c \app \ell &
           \agda{CastBlameSafe}\app d \app \ell}
          {\agda{CastBlameSafe}\app (c ; d)\app \ell}
\]

\begin{lemma}[Blame Safety of Cast Operators]
  \label{lem:blame-safety-cast-operators-lambda-C}
  Suppose $\agda{CastBlameSafe}\app c \app \ell$ and $c$ is a cross
  cast, that is, $x : \cross{c}$.
  \begin{itemize}
  \item If $x=\mathtt{Cross}\to$,
    then $\agda{CastBlameSafe}\app (\agda{dom}\app c \app x)\app \ell$
    and $\agda{CastBlameSafe}\app (\agda{cod}\app c \app x)\app \ell$.
  \item If $x=\mathtt{Cross}\times$,
    then $\agda{CastBlameSafe}\app (\agda{fst}\app c \app x)\app \ell$
    and $\agda{CastBlameSafe}\app (\agda{snd}\app c \app x)\app \ell$.
  \item If $x=\mathtt{Cross}+$,
    then $\agda{CastBlameSafe}\app (\agda{inl}\app c \app x)\app \ell$
    and $\agda{CastBlameSafe}\app (\agda{inr}\app c \app x)\app \ell$.
  \end{itemize}
\end{lemma}

\begin{proposition}
  \label{prop:lambda-C-preblamesafe}
  $\lambda\textup{\textsf{C}}'$ is an instance of the
  \agda{PreBlameSafe} structure.
\end{proposition}

\begin{lemma}[\agda{applyCast} preserves blame safety in $\lambda\textup{\textsf{C}}'$]
  If $\agda{CastBlameSafe}\app c \app \ell$ and $a : \act{c}$
  and $\agda{CastsAllSafe}\app V \app \ell$
  and $v : \val{V}$, then
  $\agda{CastsAllSafe}\app (\agda{applyCast}\app V \app v \app c \app a) \app \ell$.
\end{lemma}
\begin{proof}
    The proof is by cases on $\act{c}$ and inversion on
    $\agda{CastBlameSafe}\app c \app \ell$.
\end{proof}

\begin{proposition}
  $\lambda\textup{\textsf{C}}'$ is an instance of the \agda{BlameSafe}
  structure.
\end{proposition}

We instantiate Theorem~\ref{thm:blame-subtyping} with this
\agda{BlameSafe} instance to obtain the Blame-Subtyping Theorem for
$\lambda\textup{\textsf{C}}'$.

\begin{corollary}[Blame-Subtyping Theorem for $\lambda\textup{\textsf{C}}'$]
  \label{thm:blame-subtyping-lambda-C}
  If $M : \Gamma \vdash A$ and $\agda{CastsAllSafe} \app M \app \ell$,
  then $\neg (M \longrightarrow^{*} \key{blame} \, \ell)$.
\end{corollary}


\section{Space-Efficient Parameterized Cast Calculus}
\label{sec:EfficientParamCasts}

The cast calculi in Section~\ref{sec:CC-instances} all suffer from a
space-efficiency problem~\citep{Herman:2010aa}. When a cast is applied
to a higher-order value such as a function, these calculi either wrap
it inside another function or wrap the cast itself around the value.
Either way, the value grows larger. If a value goes through many
casts, it can grow larger and larger in an unbounded fashion. This
phenomenon can cause significant space and time overheads in real
programs, for example, changing the worse-cast time complexity of
quicksort from $O(n^2)$ to
$O(n^3)$~\citep{Takikawa:2016aa,Kuhlenschmidt:2019aa}.

\citet{Herman:2010aa} proposed solving this problem by replacing casts
with the coercions of \citet{Henglein:1994nz}. Any sequence of
coercions can normalize to just three coercions, thereby providing a
space-efficient representation. \citet{Siek:2015ab} define an
algorithm for efficiently normalizing coercions and use it to define
the $\lambda\textup{\textsf{S}}$ calculus.

\citet{Siek:2010ya} propose another approach that compresses a
sequence of casts into two casts where the middle type is the least
upper bound with respect to precision. The AGT methodology uses a
similar representation~\citep{Garcia:2016aa} and recently proved space
efficiency~\citet{Banados-Schwerter:2021aa}.

In this section we develop a space-efficient version of the
parameterized cast calculus, which we name $\SC(\CastOp)$.  As a
sanity check, prove that $\SC(\CastOp)$ is type safe, but more
importantly, we prove that $\SC(\CastOp)$ is indeed space-efficient
provided that the cast representation is an instance of the structures
defined later in this section.
In Section~\ref{sec:space-calculi} we instantiate $\SC(\CastOp)$ two
different ways to reproduce the $\lambda\textup{\textsf{S}}$ calculus
and to define a new calculus that more directly maps to a compact
bit-level encoding. We then instantiate the meta-theory for
$\SC(\CastOp)$ to produce proofs of type safety and space-efficiency
for both of these calculi.

\subsection{Space-Efficient Values}
\label{sec:EfficientParamCastAux}

This subsection is parameterized over the \agda{PreCastStruct}
structure.

To prepare for the definition of space-efficient cast calculi, we
define a notion of value that may be wrapped in at most one cast.  We
accomplish this by stratifying the non-cast values, that is the simple
values $S$, from the values $V$ that may be wrapped in a cast. \\[1ex]
\fbox{$\sval{} : (\Gamma \vdash A) \to \agda{Set}$}\\
\fbox{$\val{} : (\Gamma \vdash A) \to \agda{Set}$}
\begin{gather*}
    \mathtt{S}\lambda : 
    \inference{}
              {\sval{ (\lambda M)}}
    \quad
    \mathtt{Sconst} :
    \inference{}{\sval{ (\$k)}}
    \quad
    \mathtt{Spair} :
    \inference{\val{M} & \val{N}}
              {\sval{(\key{cons}\, M\, N)}}
    \\[1ex]
    \mathtt{Sinl} :
    \inference{\val{M}}
              {\sval{(\key{inl}[B]\, M)}}
    \quad
    \mathtt{Sinr} :
    \inference{\val{M}}
              {\sval{(\key{inr}[A]\, M)}}
    \\[1ex]
    \mathtt{Vsimp} :
    \inference{ \sval{M} }
              { \val{M} }
    \quad
    \mathtt{Vcast} :
    \inference{\sval{M}}
              { \val{(\cast{M}{c})} } \agda{Inert}\,c
\end{gather*}

\begin{lemma}\label{lem:simple-star}
   If\  $\sval{M}$ and $M : \Gamma \vdash A$, then $A \not\equiv \Unk$.
\end{lemma}

\begin{lemma}[Canonical Form for type $\Unk$]
  \label{lem:eff-canonical-star}
  If\  $M : \Gamma \vdash \Unk$ and $\val{M}$,
  then $M \equiv \cast{M'}{c}$ where $M' : \Gamma \vdash A$,
  $c : \Cast{A}{\Unk}$, $\inert{c}$, and $A \not\equiv \Unk$.
\end{lemma}

\begin{lemma}\label{lem:simple-base}
  If\  $\sval{M}$ and $M : \Gamma \vdash b$,
  then $M \equiv k$ for some $k : \rep{b}$
\end{lemma}

\subsection{The \agda{ComposableCasts} Structure}
\label{sec:ComposableCasts}

The \agda{ComposableCasts} structure extends \agda{PreCastStruct}
with two more fields, one for applying a cast to a value (like
\agda{CastStruct}) and one for composing two casts into a single,
equivalent cast, for the purposes of achieving space efficiency. It
would seem reasonable to have this structure extend \agda{CastStruct}
instead of \agda{PreCastStruct}, but the problem is that the notion of
value is different. Here we use the definition of \agda{Value} from
Section~\ref{sec:EfficientParamCastAux}.

The two fields of the \agda{ComposableCasts} are:
\begin{description}
\item[$\agda{applyCast} : \forall \Gamma A B.\, (M : \Gamma \vdash A) \to \agda{SimpleValue}\app M\to (c : \Cast{A}{B}) \to \agda{Active}\app c \to \Gamma \vdash B$]
\item[$-\fatsemi- : \forall A B C.\, \Cast{A}{B} \to \Cast{B}{C} \to \Cast{A}{C}$]
\end{description}

\subsection{Reduction Semantics of $\SC(\CastOp)$}

This section is parameterized by a \agda{ComposableCasts} and
defines the Space-Efficient Parameterized Cast Calculus, written
$\SC(\CastOp)$. The syntax is the same as that of the Parameterized
Cast Calculus (Figure~\ref{fig:param-cc-terms}).

The frames of $\SC(\CastOp)$ and the $\itm{plug}$ function are
defined in Figure~\ref{fig:eff-frames}. The definitions are quite
similar to those of the Parameterized Cast Calculus
(Figure~\ref{fig:cc-values-frames}), with the notable omission of a
frame for casts, which are handled by special congruence rules.

\begin{figure}
  \fbox{$\Gamma \vdash \Frame{A }{ B}$}
  \begin{gather*}
    (\Box \app -):
    \inference{\Gamma \vdash A}
              {\Gamma \vdash \Frame{(A \to B) }{ B}}
    \qquad
    (- \app \Box):
    \inference{M : \Gamma \vdash (A \to B)}
              {\Gamma \vdash \Frame{A }{ B}} \val{M}
    \\[1ex]
    \key{if}\,\Box\,-\,- :
    \inference{\Gamma \vdash A & \Gamma \vdash A}
              {\Gamma \vdash \Frame{\Bool }{ A}}
    \\[1ex]
    \key{cons}\,{-}\,\Box :
    \inference{M : \Gamma \vdash A}
              {\Gamma \vdash \Frame{B }{ A \times B}}~\val{M}
   \qquad
    \key{cons}\,\Box\,- :
    \inference{\Gamma \vdash B}
              {\Gamma \vdash \Frame{A }{ A \times B}}
    \\[1ex]
    \pi_i\,\Box :
    \inference{}
              {\Gamma \vdash \Frame{A_1 \times A_2 }{ A_i}}
    \\[1ex]
    \key{inl}[B]\,\Box :
    \inference{}
              {\Gamma \vdash \Frame{A }{ A \times B}}
    \qquad
    \key{inr}[A]\,\Box :
    \inference{}
              {\Gamma \vdash \Frame{B }{ A \times B}}
    \\[1ex]
    \key{case}\,\Box\,-\,-:
    \inference{\Gamma \vdash \Frame{A }{ C} & \Gamma \vdash \Frame{B }{ C}}
              {\Gamma \vdash \Frame{A + B}{ C}}
  \end{gather*}
  \fbox{$\itm{plug} : \forall \Gamma A B.\, (\Gamma \vdash A) \to (\Gamma \vdash \Frame{A }{ B}) \to (\Gamma \vdash B)$}
  \begin{align*}
    \itm{plug}\app L \app (\Box \app M) &= (L \app M) \\
    \itm{plug}\app M \app (L \app \Box) &= (L \app M) \\
    \itm{plug}\app L \app (\key{if}\,\Box\, M\,N) &= \key{if}\,L\,M\,N\\
    \itm{plug}\app N (\key{cons}\,M\,\Box) &= \key{cons}\,M\,N\\
    \itm{plug}\app M (\key{cons}\,\Box\,N) &= \key{cons}\,M\,N\\
    \itm{plug}\app M (\pi_i\,\Box) &= \pi_i\,M \\
    \itm{plug}\app M (\key{inl}[B]\,\Box) &= \key{inl}[B]\,M \\
    \itm{plug}\app M (\key{inr}[A]\,\Box) &= \key{inr}[A]\,M \\
    \itm{plug}\app L (\key{case}\,\Box\,M\,N) &= \key{case}\,L\,M\,N 
  \end{align*}
  \caption{Frames of $\SC(\CastOp)$.}
  \label{fig:eff-frames}
\end{figure}

A space-efficient reduction semantics must compress adjacent casts to
prevent the growth of long sequences of them. To date, the way to
accomplish this in a reduction semantics has been to define evaluation
contexts in a subtle way, with two mutual
definitions~\citep{Herman:2006uq,Herman:2010aa,Siek:2010ya,Siek:2015ab}. Here
we take a different approach that we believe is simpler to understand
and that fits into using frames to control evaluation order. The idea
is to parameterize the reduction relation according to whether a
reduction rule can fire in any context or only in non-cast contexts,
that is, the immediately enclosing term cannot be a cast. We define
reduction context \agda{RedCtx} as follows, making it isomorphic to
the Booleans but with more specific names.\\[1ex]
\fbox{$\itm{ctx} : \agda{RedCtx}$}
\[
  \inference{}{\mathtt{Any} : \agda{RedCtx}}
  \qquad
  \inference{}{\mathtt{NonCast} : \agda{RedCtx}}
\]
So the reduction relation will take the form
\[
   \itm{ctx} \vdash M \longrightarrow N
\]

To prevent reducing under a sequence of two or more casts, the
congruence rule for casts, $\xi\mhyphen\mathtt{cast}$, requires a
non-cast context.  Further, the inner reduction must be OK with any
context (and not require a non-cast context).  The congruence rule for
all other language features, $\xi$, can fire in any context and the
inner reduction can require either any context or non-cast contexts.
The rule for composing two casts can only fire in a non-cast context,
which enforces an outside-in strategy for compressing sequences of
casts. For the same reason, the rule for applying a cast to a value
can only fire in a non-cast context.  All other reduction rules can
fire in any context.
The reduction semantics for $\SC(\CastOp)$ is defined in
Figure~\ref{fig:eff-param-cast-reduction}.

\begin{figure}[tbp]
  \fbox{$\itm{ctx} \vdash M \longrightarrow N$}
  \begin{gather}
    \inference{\itm{ctx} \vdash M \longrightarrow M'}
              {\mathtt{Any} \vdash \itm{plug}\app M \app F
                \longrightarrow \itm{plug}\app M' \app F}
      \tag{$\xi$}\label{eq:sc-xi} \\[1ex]
    \inference{\mathtt{Any} \vdash M \longrightarrow M'}
              {\mathtt{NonCast} \vdash \cast{M}{c} \longrightarrow \cast{M'}{c}}
      \tag{$\xi\mhyphen\mathtt{cast}$}\label{eq:sc-xi-cast}\\
    \inference{}
              {\mathtt{Any} \vdash \itm{plug}\app (\blame{\ell}{A}) \app F
                \longrightarrow \blame{\ell}{B}}
              \tag{$\xi\mhyphen\key{blame}$} \label{eq:sc-xi-blame}\\[1ex]
    \inference{}
              {\key{NonCast} \vdash \cast{(\blame{\ell}{})}{c}
                \longrightarrow \blame{\ell}{} }
              \tag{$\xi\mhyphen\key{cast}\mhyphen\key{blame}$}\label{eq:sc-xi-cast-blame} \\[1ex]
    \inference{}
      {\key{NonCast} \vdash \cast{S}{c} \longrightarrow \mathsf{applyCast}\app S \app c \app a}
      a : \act{c} \tag{$\key{cast}$} \label{eq:sc-cast} \\[1ex]
    \inference{}
              {\key{NonCast} \vdash \cast{\cast{M}{c}}{d} \longrightarrow
              \cast{M}{c \fatsemi d}}
              \tag{\key{compose}}\label{eq:sc-compose-casts}\\[1ex]
    \inference{}
      {\key{Any} \vdash \cast{V}{c} \app W \longrightarrow
      \cast{(V \app \cast{W}{\dom{c \app x}})}{\cod{c \app x}}}
      x : \agda{Cross}\app c
      \tag{$\key{fun \mhyphen cast}$} \label{eq:sc-fun-cast} \\[1ex]
    \inference{}
              {\mathtt{Any} \vdash \key{fst}\app (\cast{V}{c}) \longrightarrow
                \cast{(\key{fst}\app V)}{ \agda{fst}\app c \app x}}
              x : \agda{Cross}\app c
              \tag{$\key{fst \mhyphen cast}$}\label{eq:sc-fst-cast}\\[1ex]
    \inference{}
              {\mathtt{Any} \vdash \key{snd}\app (\cast{V}{c}) \longrightarrow
                \cast{(\key{snd}\app V)}{ \agda{snd}\app c \app x}}
              x : \agda{Cross}\app c
              \tag{$\key{snd \mhyphen cast}$}\label{eq:sc-snd-cast}\\
   \inference{}
             {\mathtt{Any} \vdash \key{case}\app (\cast{V}{c}) \app W_1 \app W_2 \longrightarrow
               \key{case}\app V \app W'_1 \app W'_2}
             \tag{$\key{case \mhyphen cast}$}\label{eq:sc-case-cast}\\
             \text{where }
             \begin{array}{l}
               x : \agda{Cross}\app c \\
               W'_1 = \lambda (\agda{rename}\app \key{S} \app W_1) \app (\cast{\key{Z}}{ \agda{inl}\app c\app x})\\
               W'_2 = \lambda (\agda{rename}\app \key{S} \app W_2) \app (\cast{\key{Z}}{ \agda{inr}\app c\app x})
             \end{array} \notag \\[1ex]
    \inference{}
              {\mathtt{Any} \vdash (\lambda M) \app V \longrightarrow M[V]}
              \tag{$\beta$}\label{eq:sc-beta}  \\[1ex]
    \inference{}
              {\mathtt{Any} \vdash \key{if}\app \$\key{true} \app M \app N \longrightarrow M}
              \tag{$\beta\mhyphen\key{true}$}\label{eq:sc-beta-true}  \\[1ex]
    \inference{}
              {\mathtt{Any} \vdash \key{if}\app \$\key{false} \app M \app N \longrightarrow N}
              \tag{$\beta\mhyphen\key{false}$}\label{eq:sc-beta-false}  \\[1ex]
    \inference{}
              {\mathtt{Any} \vdash \key{fst}\app (\key{cons}\app V \app W) \longrightarrow V}
              \tag{$\beta\mhyphen\key{fst}$}\label{eq:sc-beta-fst}  \\[1ex]
    \inference{}
              {\mathtt{Any} \vdash \key{snd}\app (\key{cons}\app V \app W) \longrightarrow W}
              \tag{$\beta\mhyphen\key{snd}$}\label{eq:sc-beta-snd}  \\[1ex]
    \inference{}
              {\mathtt{Any} \vdash \key{case}\app (\key{inl}\app V)\app L \app M \longrightarrow
                L \app V}
              \tag{$\beta\mhyphen\key{caseL}$}\label{eq:sc-beta-caseL}\\[1ex]
    \inference{}
              {\mathtt{Any} \vdash \key{case}\app (\key{inr}\app V)\app L \app M \longrightarrow
                M \app V} 
              \tag{$\beta\mhyphen\key{caseR}$}\label{eq:sc-beta-caseR}\\[1ex]
    \inference{}
              {\mathtt{Any} \vdash k \app k' \longrightarrow
              \llbracket k \rrbracket \app \llbracket k' \rrbracket}
              \tag{$\delta$}\label{eq:sc-delta}
  \end{gather}
  
  \caption{Reduction for the Space-Efficient Parameterized Cast Calculus
     $\SC(\CastOp)$.}
\label{fig:eff-param-cast-reduction}
\end{figure}

\subsection{Type Safety of $\SC(\CastOp)$}

Our terms are intrinsically typed, so the fact that Agda checked the
definition in Figure~\ref{fig:eff-param-cast-reduction} gives us
Preservation.

\begin{theorem}[Preservation]
  If $\Gamma \vdash M : A$  and $M \longrightarrow M'$,
  then $\Gamma \vdash M' : A$.
\end{theorem}

Next we prove Progress.
First we define the following predicate for identifying when a term is
a cast and prove a lemma about switching from \agda{NonCast} to
\agda{Any} when the redex is not a cast.

\fbox{$\agda{IsCast}\app M$}
\[
   \inference{}{\agda{IsCast}\app (\cast{M}{c})}
\]

\begin{lemma}\label{lem:switch-back}

  If $\agda{NonCast} \vdash M \longrightarrow M'$,
  then $\agda{IsCast}\app M$.
\end{lemma}

\begin{theorem}[Progress]\label{thm:sc-progress}
  If $\emptyset \vdash M : A$, then 
  \begin{enumerate}
  \item $\itm{ctx} \vdash M \longrightarrow M'$ for some $M'$ and
    $\itm{ctx}$,
  \item $\val{M}$, or
  \item $M \equiv \blame{\ell}{}$.
  \end{enumerate}
\end{theorem}
\begin{proof}
  The proof is quite similar to that of Theorem~\ref{thm:cc-progress}
  except in the case for casts, so we explain just that case here.
  \begin{description}
  \item[Case $\cast{M}{c}$]
    The induction hypothesis for $M$ yields three sub cases.
    \begin{description}
    \item[Subcase $\itm{ctx} \vdash M \longrightarrow M'$.]
      Suppose $\itm{ctx} = \key{Any}$.
      By rule \eqref{eq:sc-xi-cast} we conclude that
      \[
      \key{Any} \vdash \cast{M}{c} \longrightarrow \cast{M'}{c}
      \]
      On the other hand, suppose $\itm{ctx} = \key{NonCast}$.
      By Lemma~\ref{lem:switch-back}, $\agda{IsCast}\app M$,
      so we have $M \equiv \cast{M_1}{d}$.
      By rule \eqref{eq:sc-compose-casts} we conclude that
      \[
      \key{NonCast} \vdash \cast{\cast{M_1}{d}}{c} \longrightarrow
        \cast{M_1}{d \fatsemi c}
      \]
    \item[Subcase $M \equiv \blame{\ell}{}$.]
      By rule \eqref{eq:sc-xi-cast-blame} we conclude that
      \[
      \key{NonCast}
      \vdash \cast{(\blame{\ell}{})}{c} \longrightarrow \blame{\ell}{}
      \]
    \item[Subcase $\val{M}$.]
      Here we use the \agda{ActiveOrInert} field of the
      \agda{PreCastStruct} on the cast $c$.
      Suppose $c$ is active, so we have $a : \act{c}$.
      By rule \eqref{eq:sc-cast}, using $\val{M}$, we conclude that
      \[
      \key{NonCast} \vdash \cast{M}{c}
       \longrightarrow \agda{applyCast} \app M \app c \app a
      \]
      Suppose $c$ is inert. From $\val{M}$ we know
      that $M$ is either a simple value or a cast.
      If $M$ is a simple value, then we conclude
      that $\val{\cast{M}{c}}$.
      Otherwise, $M \equiv \cast{M_1}{d}$ and
      we conclude by rule \eqref{eq:sc-compose-casts}.
      \[
      \key{NonCast} \vdash \cast{\cast{M_1}{d}}{c}
        \longrightarrow \cast{M_1}{d \fatsemi c}
      \]
    \end{description}
    
  \end{description}
  
\end{proof}

\subsection{Space Efficiency of $\SC(\CastOp)$}
\label{sec:space-CC}


We follow the space efficiency proof of \citet{Herman:2010aa}, but
refactor it into two parts: 1) generic lemmas about the reduction of
$\SC(\CastOp)$ that appear in this section and 2) lemmas about
specific casts and coercions, which appear in
Section~\ref{sec:space-calculi}. We clarify a misleading statement by
\citet{Herman:2010aa} and fill in details needed to mechanize the
proof in Agda.

The theorem we aim to prove is that, during execution, the program's
size is bounded above by the program's idealized size multiplied by a
constant. The idealized size does not include any of the casts. The
following are excerpts from the definitions of
\agda{real\mhyphen{}size} and \idealsize{}.
\begin{align*}
  \cdots & \\
  \agda{real\mhyphen{}size}(\cast{M}{c}) &= \agda{size}(c) + \agda{real\mhyphen{}size}(M)\\
  \cdots & \\
  \idealsize(\cast{M}{c}) &= \idealsize(M )
\end{align*}
We shall prove that the size of every cast is bounded above by a
constant, so we can simplify some of the technical development by
using the following alternative definition of \agda{size} that uses
$1$ as the size of each cast.
\begin{align*}
  \agda{size}(\cast{M}{c}) &= 1 + \agda{size}(M) \\
\end{align*}

Regarding reduction of $\SC(\CastOp)$, the key property is that the
reduction rules prevent the accumulation of long sequences of adjacent
casts. In their proof, \citet{Herman:2010aa} state that there is never
a coercion adjacent to another coercion.
\begin{quote}
  ``During evaluation, the [\textsc{E-CCast}] rule prevents nesting of
  adjacent coercions in any term in the evaluation context, redex, or
  store. Thus the number of coercions in the program state is
  proportional to the size of the program state.''
\end{quote}
Of course, for the rule [\textsc{E-CCast}] to apply in the first
place, there must be two adjacent coercions. So perhaps we could amend
the statement of \citet{Herman:2010aa} to instead say that there are
never more than two.  However, even that would be technically
incorrect. Consider the following example that begins with three
separated coercions but a $\beta$-reduction brings them together.
\[
\cast{((\lambda\; \cast{` 0}{\Int ?}) \app (\cast{\key{1}}{\Int!}))}{\Int!}
\longrightarrow
  \cast{\cast{\cast{\key{\$1}}{\Int!}}{\Int ?^\ell}}{\Int!}
\]
This turns out to be the worst-case scenario. In the following we
prove that the [\textsc{E-CCast}] rule, i.e. the
\eqref{eq:sc-compose-casts} rule in this article, together with rules
about the order of evaluation, prevent nesting of more than three
adjacent coercions.

We define the Size Predicate on terms in Figure~\ref{fig:sizeok} which
only includes terms with no more than three adjacent coercions. The
judgment is written $\sizeok{n}{b}{M}$ where $M$ is a term, $n$ counts
the number of cast application expressions at the top of the term, and
$b$ indicates with \itm{true} or \itm{false} whether this term is in a
delayed context, that is, inside a $\lambda$-abstraction or a branch
of a conditional expression.  The above example with three coercions
satisfies the predicate when outside of a delayed context.
\[
  \sizeok{3}{\itm{false}}{ \cast{\cast{\cast{\key{\$1}}{\Int!}}{\Int?^\ell}}{\Int!} }
\]
The rule (\textsc{SCast1}) for cast application expressions adds one to
the count of adjacent casts and makes sure that the count does not
exceed three.

For terms in a delayed context, the rule (\textsc{SCast2}) restricts
the number of adjacent casts to two instead of three. To see why,
consider the next example in which there are three adjacent casts
inside the $\lambda$-abstraction and a $\beta$-reduction yields a term
with four adjacent casts.
\[
  \cast{
    ((\lambda\;
    \cast{
      \cast{
        \cast{\$ 8}{\Int!}
      }{ \Int?^{\ell_1} }
    }{ \Int! })
    \app
    \key{1})
  }{ \Int?^{\ell_2} }  
  \longrightarrow
  \cast{
    \cast{
      \cast{
        \cast{\$ 8}{\Int!}
      }{ \Int?^{\ell_1} }
    }{ \Int! }
  }{ \Int?^{\ell_2} }  
\]

The rule (\textsc{SVar}) starts the count at one even though a
variable is obviously not a cast application. The reason is that a
value substituted for a variable may have one cast at the top.  If we
did not count variables as one, then a variable could be surrounded by
two casts inside of a $\lambda$-abstraction which could reduce to a
term with four adjacent casts as in the following example.
\[
\cast{((\lambda\; \cast{ \cast{` 0}{\Int ?^{\ell_1}} }{\Int!})
  \app
  (\cast{\key{1}}{\Int!}))}{\Int?^{\ell_2}}
\longrightarrow
\cast{
  \cast{
    \cast{ \cast{\key{1}}{\Int!} }{\Int ?^{\ell_1}} }{\Int!}
   }{\Int?^{\ell_2}}
\]

\begin{figure}[tbp]
  \fbox{$\sizeok{n}{b}{M}$}
  \begin{gather*}
    \inference[(\textsc{SCast1})]{\sizeok{n}{\itm{false}}{M} & n \leq 2}
              {\sizeok{n + 1}{\itm{false}}{\cast{M}{c}}}
    \qquad
    \inference[(\textsc{SCast2})]{\sizeok{n}{\itm{true}}{M} & n \leq 1}
              {\sizeok{n + 1}{\itm{true}}{\cast{M}{c}}}
    \\[1em]
    \inference[(\textsc{SVar})]{}
              {\sizeok{1}{b}{` x}}
    \qquad
    \inference{\sizeok{n}{\itm{true}}{N}}
              {\sizeok{0}{b}{\lambda N}}
    \qquad
    \inference{\sizeok{n}{b}{L} & \sizeok{m}{b}{M}}
              {\sizeok{0}{b}{L \app M}} \\[1em]
    \inference{}
              {\sizeok{0}{b}{\$ k}}
    \qquad
    \inference{\sizeok{n}{b}{L} & \sizeok{m}{\itm{true}}{M}
              & \sizeok{k}{\itm{true}}{N}}
              {\sizeok{0}{b}{\key{if}_\ell}\,L\, M\, N}\\[1ex]
    \inference{\sizeok{n}{b}{M} & \sizeok{m}{b}{N}}
              {\sizeok{0}{b}{\key{cons}\,M\,N}} \qquad
    \inference{\sizeok{n}{b}{M}}
              {\sizeok{0}{b}{\key{fst}\,M}} \qquad
    \inference{\sizeok{n}{b}{M}}
              {\sizeok{0}{b}{\key{snd}\,M}} \\[1em]
    \inference{\sizeok{n}{b}{M}}
              {\sizeok{0}{b}{\key{inl}[B]\,M}} \qquad
    \inference{\sizeok{n}{b}{M}}
              {\sizeok{0}{b}{\key{inr}[B]\,M}} \\[1em]
    \inference{\sizeok{n}{b}{L} & \sizeok{m}{\itm{true}}{M} & \sizeok{k}{\itm{true}}{N}}
              {\sizeok{0}{b}{\key{case}\,L\,M\,N}} \qquad
    \inference{}
              {\sizeok{0}{b}{\key{blame}\,\ell}}
  \end{gather*}
  \caption{The Size Predicate that limit the number of adjacent casts.}
  \label{fig:sizeok}
\end{figure}

We turn to the proof of the space consumption theorem, starting with
the necessary lemmas.

The Size Predicate guarantees that the number of adjacent casts is
less-than or equal to $3$.

\begin{lemma}[Maximum of 3 Adjacent Casts]
  \label{lem:OK3}
  If\ \ $\sizeok{n}{b}{M}$ then $n \le 3$.
\end{lemma}

The Size Predicate guarantees that a term's size is bounded above by
its ideal size multiplied by a constant, plus 3.

\begin{lemma}[Size Predicate and Ideal Size]\ \\
  \label{lem:size-OK}
  If\ \ $\sizeok{n}{b}{M}$ then $\agda{size}(M) \leq 10 \cdot \idealsize(M) + 3$.
\end{lemma}

\noindent The compilation of source programs (GTCL) to the cast
calculus produces terms that satisfy the Size Predicate.

\begin{lemma}[Cast Insertion Size]\ \\
  \label{lem:compile-efficient}
  If $M : \Gamma \vdash_G A$ then
  $\sizeok{n}{b}{\compile{}{M}}$ for some $n \leq 1$.
\end{lemma}






 

\noindent Reduction preserves the Size Predicate. The proof of this
lemma involves a number of technical lemmas about substitution,
evaluation contexts, and values, which can be found in the Agda
formalization.

\begin{lemma}[Size Preservation]
  \label{lem:preserve-ok}
  If\ $M : \Gamma \vdash A$ and $M' : \Gamma \vdash A$ and
  $\sizeok{n}{\itm{false}}{M}$ and $M \longrightarrow M'$, then
  $\sizeok{m}{\itm{false}}{M'}$ for some $m$.
\end{lemma}



The next piece needed for the space efficiency theorem is to place a
bound on the size of the casts. We require their size to be bounded by
their height, as in \citet{Herman:2010aa}, so it remains to show that
the heights of all the coercions does not grow during execution. To
prove this we must place further requirements on the specific cast
calculi, which we formulate as a structure named
\agda{CastHeight} in Figure~\ref{fig:eff-cast-height}
that extends the \agda{ComposableCasts} structure in
Section~\ref{sec:EfficientParamCasts}.

\begin{figure}[tp]

\begin{description}
\item[$\agda{height} : \forall A B.~ (c : \Cast{A}{B}) \to \mathbb{N}$]
\item[$\agda{size} : \forall A B.~ (c : \Cast{A}{B}) \to \mathbb{N}$]  
\item[$\agda{compose\mhyphen{}height} : \begin{array}{l}
     \forall A B C.~ (c : \Cast{A}{B}) \to (d : \Cast{B}{C}) \\
    \to \agda{height}(c \fatsemi d) \leq \max(\agda{height}(c),\agda{height}(d))
  \end{array}$]
\item[$\agda{applyCastSize} :
  \begin{array}{l}
    \forall \Gamma A B n.~ (M : \Gamma \vdash A) \to (c : \Cast{A}{B})\\
    \to \sizeok{n}{\itm{false}}{M} \to (v : \val{M})\\
    \to\exists m.~\sizeok{m}{\itm{false}}{(\agda{applyCast}~M~v~c)}\times m\leq 2 + n
  \end{array}$]
\item[$\agda{applyCastHeight} :
  \begin{array}{l}
    \forall \Gamma A B V.~ (v : \val{V}) \\
    \to \agda{height}(\agda{applyCast}~V~v~c)
      \leq \max(\agda{height}(V), \agda{height}(c))
  \end{array}$]
\item[$\agda{dom\mhyphen{}height} :
  \begin{array}{l}
    (c : \Cast{A\to B}{C \to D}) \to (x : \cross{c})\\
    \to \agda{height}(\dom{c\app x}) \leq \agda{height}(c)
  \end{array}$]  
\item[$\agda{cod\mhyphen{}height} :
  \begin{array}{l}
    (c : \Cast{A\to B}{C \to D}) \to (x : \cross{c})\\
    \to \agda{height}(\cod{c\app x}) \leq \agda{height}(c)
  \end{array}$]  
\item[$\agda{fst\mhyphen{}height} :
  \begin{array}{l}
    (c : \Cast{A\times B}{C \times D}) \to (x : \cross{c})\\
    \to \agda{height}(\agda{fst}\app c\app x) \leq \agda{height}(c)
  \end{array}$]  
\item[$\agda{snd\mhyphen{}height} :
  \begin{array}{l}
    (c : \Cast{A\times B}{C \times D}) \to (x : \cross{c})\\
    \to \agda{height}(\agda{snd}\app c\app x) \leq \agda{height}(c)
  \end{array}$]  
\item[$\agda{inl\mhyphen{}height} :
  \begin{array}{l}
    (c : \Cast{A+ B}{C+D}) \to (x : \cross{c})\\
    \to \agda{height}(\agda{inl}\app c\app x) \leq \agda{height}(c)
  \end{array}$]  
\item[$\agda{inr\mhyphen{}height} :
  \begin{array}{l}
    (c : \Cast{A+ B}{C+D}) \to (x : \cross{c})\\
    \to \agda{height}(\agda{inr}\app c\app x) \leq \agda{height}(c)
  \end{array}$]
\item[$\agda{size\mhyphen{}height} :
    \exists k_1 k_2.~ \forall A B.~(c : \Cast{A}{B}) \to \agda{size}(c) + k_1 \le k_2 * 2^{\agda{height}(c)}
    $]
\end{description}
\caption{\agda{CastHeight} extends
  \agda{ComposableCasts} (Section~\ref{sec:ComposableCasts}).}
\label{fig:eff-cast-height}
\end{figure}

We define the \emph{cast height} of a term, written
$\agda{c\mhyphen{}height}(M)$, to be the maximum of the heights of the
all the casts in term $M$. The cast-height of a term is monotonically
decreasing under reduction.

\begin{lemma}[Preserve Height]
  \label{lem:preserve-height}
  If ~$M : \Gamma \vdash A$, $M' : \Gamma \vdash A$,
  and $\itm{ctx} \vdash M \longrightarrow M'$,
  then $\agda{c\mhyphen{}height}(M') \leq \agda{c\mhyphen{}height}(M)$.
\end{lemma}

The size of any cast $c$ is bounded above by $k_2 \cdot
2^{\agda{height}(c)}$ (according to the $\agda{size\mhyphen{}height}$
member of \agda{CastHeight}), so the
\agda{real\mhyphen{}size} of a term is bounded above by $k_2 \cdot
2^{\agda{c\mhyphen{}height}(M)}$ multiplied by its \agda{size}.

\begin{lemma}[Real Size Bounded by Size]\ \\
  \label{lem:real-size-tsize}
  If $M : \Gamma \vdash A$, 
  then
  $\agda{real\mhyphen{}size}(M) \leq k_2 \cdot 2^{\agda{c\mhyphen{}height}(M)} \cdot \agda{size}(M)$.
\end{lemma}

With the above lemmas in place, we proceed with the main theorem.  The
size of the program (including the casts) is bounded by above by its
ideal size (not including casts) multiplied by a constant.

\begin{theorem}[Space Consumption]
  \label{thm:space-consumption}
  If ~$\Gamma \vdash M : A$, then there exists $c$ such that for any
  $M' : \Gamma \vdash A$ where $\itm{ctx} \vdash \compile{\Gamma}{M}
  \longrightarrow^{*} M'$, we have $\agda{real\mhyphen{}size}(M') \leq
  c \cdot \idealsize(M')$.
\end{theorem}
\begin{proof}
  (The formal proof of this theorem is in the Agda development. Here
  we give a proof that is less formal but easier to read.)
  By Lemma~\ref{lem:compile-efficient} there is an $n$ such that
  \[
  \sizeok{n}{\itm{false}}{\compile{\Gamma}{M}}
  \]
  Using Lemma~\ref{lem:preserve-ok} and an induction on the reduction
  sequence $\itm{ctx} \vdash \compile{\Gamma}{M} \longrightarrow^{*}
  M'$, there is an $m$ such that
  \[
  \sizeok{m}{\itm{false}}{M}
  \]
  We establish the conclusion of the theorem by the following
  inequational reasoning.
  \begin{align*}
    \agda{real\mhyphen{}size}(M') & \leq k_2 \cdot 2^{\agda{c\mhyphen{}height}(M')} \cdot \agda{size}(M') & \text{by Lemma~\ref{lem:real-size-tsize}}\\
    & \leq (k_2 \cdot 2^{\agda{c\mhyphen{}height}(M')}) \cdot (3 + 10 \cdot \agda{ideal\mhyphen{}size}(M')) & \text{by Lemma~\ref{lem:size-OK}}\\
    & \leq 3 k_2 \cdot 2^{\agda{c\mhyphen{}height}(M')} +  10 k_2 \cdot 2^{\agda{c\mhyphen{}height}(M')} \cdot \agda{ideal\mhyphen{}size}(M')\\
    & \leq 3 k_2 \cdot 2^{\agda{c\mhyphen{}height}(\compile{}{M})} + 10 k_2 \cdot 2^{\agda{c\mhyphen{}height}(\compile{}{M})} \cdot \agda{ideal\mhyphen{}size}(M')
      & \text{by Lemma~\ref{lem:preserve-height}}\\
    & \leq 13 k_2 \cdot 2^{\agda{c\mhyphen{}height}(\compile{}{M})} \cdot \idealsize(M')
  \end{align*}
  So we choose $13 k_2 \cdot
  2^{\agda{c\mhyphen{}height}(\compile{}{M})}$ as the witness for $c$.
\end{proof}

We observe that the space consumption theorem of \citet{Herman:2010aa}
(as well as the above theorem) has a limitation in that it does not
prevent a cast calculus that uses the \agda{eta} cast reduction rules,
as discussed in Section~\ref{sec:eta-cast-reduction}, from consuming
an unbounded amount of space. To capture the space used by the
\agda{eta} cast rules, one would need to mark the terms introduced by
the \agda{eta} cast rules so that they can be excluded from the
\idealsize{} of term.  We do not pursue that direction at this time,
but note that the calculi that we introduce in the next section do not
use \agda{eta} cast rules.

\section{Space-Efficient Cast Calculi}
\label{sec:space-calculi}

We instantiate the Efficient Parameterized Calculus $\SC(\CastOp)$
with three different instances of the \agda{ComposableCasts} and
\agda{CastHeight} to obtain definitions and proofs of
type safety and space efficiency for two cast calculi:
$\lambda\textup{\textsf{S}}$ of \citet{Siek:2015ab} and the
hypercoercions of \citet{Lu:2019aa}.

\subsection{$\lambda\textup{\textsf{S}}$}
\label{sec:lambda-S}

The cast representation in $\lambda\textup{\textsf{S}}$ are coercions
in a particular canonical form, with a three-part grammar consisting
of top-level coercions, intermediate coercions, and ground coercions,
defined in Figure~\ref{fig:coercions-lambdaS}.  A top-level coercion
is an identity cast, a projection followed by an intermediate
coercion, or just an intermediate coercion.  An intermediate coercion
is a ground coercion followed by an injection, just a ground coercion,
or a failure coercion.  A ground coercion is an identity on base type
or a cross cast between function, pair, or sum types.

\begin{figure}[tp]
\fbox{$c,d : \Cast{A}{B}$}
\fbox{$i : \ICast{A}{B}$}
\fbox{$g,h : \GCast{A}{B}$}
\begin{gather*}
  \inference
      {}
      {\agda{id} : \Cast{\Unk}{\Unk}}
  \qquad
  \inference
      {i : \ICast{H}{B}}
      {(H?^\ell; i)  : \Cast{\Unk}{B}}
  \qquad
  \inference
      {i : \ICast{A}{B}}
      {i  : \Cast{A}{B}}
\\[1ex]  
  \inference
      {g : \GCast{A}{G}}
      {g; G! : \ICast{A}{\Unk}}
  \qquad
  \inference
      {g : \GCast{A}{B}}
      {g : \ICast{A}{B}}
  \qquad
  \inference
      {}
      {\cfail{\ell} : \ICast{A}{B}}
\\[1ex]
  \inference
      {}
      {\agda{id} : \GCast{b}{b}}
  \qquad
  \inference
      {c : \Cast{C}{A} & d : \Cast{B}{D} }
      {c \to d : \GCast{A \to B}{C \to D}}
\\[1ex]
  \inference
      {c : \Cast{A}{C} & d : \Cast{B}{D} }
      {c \times d : \GCast{A \times B}{C \times D}}
  \qquad
  \inference
      {c : \Cast{A}{C} & d : \Cast{B}{D} }
      {c + d : \GCast{A + B}{C + D}}
\end{gather*}
\fbox{$\coerce{A}{B}{\ell} = c, \coerce{A}{G}{\ell} =  g, \coerce{H}{A}{\ell} = g$}
\begin{align*}
  \coerce{\Unk}{\Unk}{\ell} &= \agda{id}\\
  \coerce{A}{\Unk}{\ell}&= \coerce{A}{G}{\ell}; G! 
    & \text{where } G = \agda{gnd} \app A, A \neq \Unk \\
  \coerce{\Unk}{A}{\ell}&= H?^\ell; \coerce{H}{A}{\ell}
    & \text{where } H = \agda{gnd} \app A, A \neq \Unk \\
  \coerce{b}{b}{\ell} &= \agda{id}\\
  \coerce{A \to B}{A' \to B'}{\ell}&=
      \coerce{A'}{A}{\ell} \to \coerce{B}{B'}{\ell} \\
  \coerce{A \times B}{A' \times B'}{\ell}&=
      \coerce{A'}{A}{\ell} \times \coerce{B}{B'}{\ell} \\
  \coerce{A + B}{A' + B'}{\ell}&= 
      \coerce{A'}{A}{\ell} + \coerce{B}{B'}{\ell} 
\end{align*}

\caption{Coercions of $\lambda\textup{\textsf{S}}$ and its cast constructor.}
\label{fig:coercions-lambdaS}
\end{figure}

The cast constructor is also defined in Figure~\ref{fig:coercions-lambdaS}.

\subsubsection{Reduction Semantics and Type Safety}

Casts between function, pair, and sum types are categorized as cross
casts. \\[1ex]
\fbox{$\cross{c}$}
\begin{gather*}
  \inference
      {}
      { \key{Cross}{\otimes} : \cross{(c \otimes d)} }
      ~ \otimes \in \{ \to, \times, + \}
\end{gather*}

The inert casts include casts between function types, injections, and
the failure coercion. \\[1ex]
\fbox{$\inert{c}$}
\begin{gather*}
  \inference
      {}
      { \inert{c \to d} }
  \qquad
  \inference
      {}
      { \inert{g; G!} }
  \qquad
  \inference
      {}
      { \inert{\cfail{\ell}} }
\end{gather*}

There are five kinds of active coercions: the identity on $\Unk$,
projections, failures, cross casts on pairs and sums, and identity on
base types. \\[1ex]
\fbox{$\act{c}$}
\begin{gather*}
  \inference
      {}
      {\key{Aid} : \act{ \agda{id} } }
  \qquad
  \inference
      {}
      {\key{Aproj} : \act{ (G?^\ell; i) } }
  \qquad
  \inference
      {}
      {\key{Afail} : \act{ \cfail{\ell} } }
  \\[1ex]
  \inference
      {\otimes \in \{ \times, +\} }
      {\key{A}{\otimes} : \act{ (c \otimes d) } }
  \qquad
  \inference
      {}
      {\key{Abase} : \act{ \agda{id} } }
\end{gather*}

\begin{lemma}
  \label{lem:lambda-S-active-or-inert}
  For any types $A$ and $B$, $c : \Cast{A}{B}$ is either an active or
  inert cast.
\end{lemma}

\begin{lemma}\label{lem:lambda-S-inert-cross}
  If $c : \Cast{A}{(B \otimes C)}$ and $\inert{c}$,
  then $\cross{c}$ and $A \equiv D \otimes E$ for some $D$ and $E$.
\end{lemma}

The definition of $\agda{dom}$, etc. for $\lambda\textup{\textsf{S}}$
is given below.
\begin{align*}
  \dom{(c \to d) \app \key{Cross}{\to}} &= c \\
  \cod{(c \to d) \app \key{Cross}{\to}} &= d \\
  \agda{fst}\app (c \times d) \app \key{Cross}{\times} &= c \\
  \agda{snd} \app (c \times d) \app \key{Cross}{\times} &= d \\
  \agda{inl} \app (c + d) \app \key{Cross}{+} &= c \\
  \agda{inr} \app (c + d) \app \key{Cross}{+} &= d
\end{align*}

\begin{lemma}\label{lem:lambda-S-base-not-inert}
   A cast $c : \Cast{A}{b}$ is not inert.
\end{lemma}

\begin{proposition}\label{prop:lambda-S-active-pre-cast-struct}
$\lambda\textup{\textsf{S}}$ is an instance of the
  $\agda{PreCastStruct}$ structure.
\end{proposition}

To support space efficiency, we define a composition operator for the
coercions of $\lambda\textup{\textsf{S}}$. The operator uses two
auxiliary versions of the operator for intermediate and ground
coercions. The operator that composes an intermediate coercion with a
coercion always yields an intermediate coercion. The operator that
composes two ground coercions always returns a ground coercion.  Agda
does not automatically prove termination for this set of mutually
recursive functions, so we manually prove termination, using the sum
of the sizes of the two coercions as the measure. \\[1ex]
\fbox{$c \fatsemi d$} \fbox{$i \fatsemi d$} \fbox{$g \fatsemi h$}
\begin{align*}
  \agda{id} \fatsemi d &= d \\
  (G?^\ell; i) \fatsemi d &= G?^\ell ; (i \fatsemi d) \\
  \\
  (g; G!) \fatsemi \agda{id} &= g; G! \\
  g \fatsemi (h; H!) &= (g \fatsemi h) ; H! \\
  (g; G!) \fatsemi (G?^\ell ; i) &= g \fatsemi i \\
  (g; G!) \fatsemi (H?^\ell ; i) &= \cfail{\ell} & \text{if } G \neq H \\
  \cfail{\ell} \fatsemi d &= \cfail{\ell} \\
  g \fatsemi \cfail{\ell} &= \cfail{\ell} \\
  \\
  \agda{id} \fatsemi \agda{id} &= \agda{id} \\
  (c_1 \to d_1) \fatsemi (c_2 \to d_2) &= (c_2 \fatsemi c_1) \to (d_1 \fatsemi d_2) \\
  (c_1 \times d_1) \fatsemi (c_2 \times d_2) &= (c_1 \fatsemi c_2) \times (d_1 \fatsemi d_2) \\
  (c_1 + d_1) \fatsemi (c_2 + d_2) &= (c_1 \fatsemi c_2) + (d_1 \fatsemi d_2) 
\end{align*}

We define \agda{applyCast} for $\lambda\textup{\textsf{S}}$ by cases on
the coercion.
\[
  \agda{applyCast} : \forall \Gamma A B.\, (M : \Gamma \vdash A) \to \agda{SimpleValue}\app M\to (c : \Cast{A}{B}) \to \agda{Active}\app c \to \Gamma \vdash B
\]
\[
\begin{array}{lllllll}
  \agda{applyCast} & M & v & \agda{id} & a &=& M \\
  \agda{applyCast} & M & v & \cfail{\ell} & a &=& \blame{\ell}{} \\
  \agda{applyCast} & \cast{M}{c} & v & (G?^\ell; i) & a &=& 
     \cast{M}{c \fatsemi (G?^\ell; i)} \\
  \agda{applyCast} & (\key{cons}\app V_1 \app V_2) & v & c \times d & a &=&
     \key{cons}\app (\cast{V_1}{c}) \app (\cast{V_2}{d}) \\
  \agda{applyCast} & (\key{inl}\app V) & v & c + d & a &=&
     \key{inl}\app (\cast{V}{c})\\
  \agda{applyCast} & (\key{inr}\app V) & v & c + d & a &=&
     \key{inr}\app (\cast{V}{d})
\end{array}
\]

\begin{proposition}
$\lambda\textup{\textsf{S}}$ is an instance of the
  $\agda{ComposableCasts}$ structure.
\end{proposition}

We import and instantiate the reduction semantics and proof of type
safety from Section~\ref{sec:EfficientParamCasts} to obtain the
following definitions and results for $\lambda\textup{\textsf{S}}$.

\begin{definition}[Reduction]
  \label{def:lambdaS-reduction}
  The reduction relation $\itm{ctx} \vdash M \longrightarrow_S N$ of
  $\lambda\textup{\textsf{S}}$ is the reduction relation of
  $\SC(\CastOp)$ instantiated with $\lambda\textup{\textsf{S}}$'s
  instance of the $\agda{ComposableCasts}$ structure.
\end{definition}

\begin{corollary}[Preservation for $\lambda\textup{\textsf{S}}$]
  If $\Gamma \vdash M : A$  and $M \longrightarrow_S M'$,
  then $\Gamma \vdash M' : A$.
\end{corollary}

\begin{corollary}[Progress for $\lambda\textup{\textsf{S}}$]
  If $\emptyset \vdash M : A$, then 
  \begin{enumerate}
  \item $M \longrightarrow_S M'$ for some $M'$,
  \item $\val{M}$, or
  \item $M \equiv \blame{\ell}{}$.
  \end{enumerate}
\end{corollary}

\subsubsection{Space Efficiency}

Next we establish that $\lambda\textup{\textsf{S}}$ is an instance of
the $\agda{CastHeight}$ structure so that we can apply
Theorem~\ref{thm:space-consumption} (Space Consumption) to obtain
space efficiency for $\lambda\textup{\textsf{S}}$.

We define the height of a coercion as follows.
\begin{align*}
  \agda{height}(\agda{id}) &= 0 \\
  \agda{height}(\cfail{\ell}) &= 0 \\
  \agda{height}(G?^\ell; i) &= \agda{height}(i) \\
  \agda{height}(g; G!) &= \agda{height}(g) \\
  \agda{height}(c \to d) &= 1 + \max(\agda{height}(c), \agda{height}(d)) \\
  \agda{height}(c \times d) &= 1 + \max(\agda{height}(c), \agda{height}(d)) \\
  \agda{height}(c + d) &= 1 + \max(\agda{height}(c), \agda{height}(d)) 
\end{align*}
The size of a coercion is given by the following definition.
\begin{align*}
  \agda{size}(\agda{id}) &= 0 \\
  \agda{size}(\cfail{\ell}) &= 0 \\
  \agda{size}(G?^\ell; i) &= 2 + \agda{size}(i) \\
  \agda{size}(g; G!) &= 2 + \agda{size}(g) \\
  \agda{size}(c \to d) &= 1 + \agda{size}(c) + \agda{size}(d) \\
  \agda{size}(c \times d) &= 1 + \agda{size}(c) + \agda{size}(d) \\
  \agda{size}(c + d) &= 1 + \agda{size}(c) + \agda{size}(d) 
\end{align*}

The cast height of the result of \agda{applyCast} applied to a simple
value $S$ and coercion $c$ is less than the max of the cast height of
$S$ and the height of $c$.
\begin{lemma}
  \label{lem:applyCast-height}
  $\agda{c\mhyphen{}height}(\agda{applyCast}\app S\app c) \leq \max(\agda{height}(S), \agda{height}(c))$
\end{lemma}

The \agda{dom}, \agda{cod}, \agda{fst}, \agda{snd}, \agda{inl}, and
\agda{inr} operators on coercions all return coercions of equal or
lesser height than their input.

\begin{lemma}\ 
  \label{lem:dom-height}
  \begin{enumerate}
  \item $\agda{height}(\agda{dom}\app c \app x) \leq \agda{height}(c)$
  \item $\agda{height}(\agda{cod}\app c \app x)) \leq \agda{height}(c)$
  \item $\agda{height}(\agda{fst}\app c \app x) \leq \agda{height}(c)$
  \item $\agda{height}(\agda{snd}\app c \app x) \leq \agda{height}(c)$
  \item $\agda{height}(\agda{inl}\app c \app x) \leq \agda{height}(c)$
  \item $\agda{height}(\agda{inr}\app c \app x) \leq \agda{height}(c)$
  \end{enumerate}
\end{lemma}

The size of a coercion $c$ is bounded by $9 \cdot
2^{\agda{height}(c)}$. We prove this by simultaneously proving the
following three facts about the three kinds of coercions.

\begin{lemma}\ 
  \label{lem:size-height}
  \begin{enumerate}
  \item $\agda{size}(c) + 5 \le 9 \cdot 2^{\agda{height}(c)}$
  \item $\agda{size}(i) + 7 \le 9 \cdot 2^{\agda{height}(i)}$
  \item $\agda{size}(g) + 9\le 9 \cdot 2^{\agda{height}(g)}$
  \end{enumerate}
\end{lemma}

The above lemmas and definitions establish the following.

\begin{proposition}
$\lambda\textup{\textsf{S}}$ is an instance of the
  $\agda{CastHeight}$ structure.
\end{proposition}

We apply Theorem~\ref{thm:space-consumption-lambdaS} (Space
Consumption) to obtain the following result for
$\lambda\textup{\textsf{S}}$.

\begin{corollary}[Space Consumption for $\lambda\textup{\textsf{S}}$]
  \label{thm:space-consumption-lambdaS}
  If ~$\Gamma \vdash M : A$, then there exists $c$ such that for any
  $M' : \Gamma \vdash A$ where $\itm{ctx} \vdash \compile{\Gamma}{M}
  \longrightarrow_S^{*} M'$, we have $\agda{real\mhyphen{}size}(M') \leq
  c \cdot \idealsize(M')$.
\end{corollary}

\subsection{Hypercoercions}
\label{sec:hypercoercions}

This section develops an alternative formulation of
$\lambda\textup{\textsf{S}}$ using a new representation, called
hypercoercions and written $\lambda\textup{\textsf{H}}$, that more
directly maps to a compact bit-level encoding. We presented
hypercoercions at the Workshop on Gradual
Typing~\citep{Lu:2019aa}. Hypercoercions are inspired by the
supercoercions of \citet{Garcia:2013fk}.

The idea behind hypercoercions is to choose a canonical representation
in which a coercion always has three parts, a beginning $p$, middle
$m$, and end $i$. The beginning part $p$ may be a projection or an
identity. The middle part $m$ is a cross cast or an identity cast at
base type.  The end part $i$ may be an injection, failure, or
identity. The definition of hypercoercions, given in
Figure~\ref{fig:hypercoercions}, factors the definition into these
three parts.

\begin{figure}[tb]
\fbox{$c,d : \Cast{A}{B}$}
\fbox{$p : \Proj{A}{ B}$}
\fbox{$m : \Middle{ A}{ B}$}
\fbox{$i : \Inj{ A}{ B}$}\\
\begin{gather*}
  \inference{}
            {\agda{id} : \Cast{\Unk}{\Unk}}
  \qquad
  \inference
      {p : \Proj{ A}{ B} & m : \Middle{ B}{ C} & i : \Inj{ C}{ D} }
      {p ; m ; i : \Cast{A}{D}}
  \\[1ex]
  \inference
      {}
      {\agda{id} : \Proj{ A}{ A} }
  \qquad
  \inference
      {}
      {H?^\ell : \Proj{\Unk}{H}}
  \\[1ex]
  \inference
      {}
      {\agda{id} : \Middle{a}{a}}
  \qquad
  \inference
      {c : \Cast{C}{A} & d : \Cast{B}{D}}
      {c \to d : \Middle{A \to B}{C \to D}}
  \\[1ex]
  \inference
      {c : \Cast{A}{C} & d : \Cast{B}{D}}
      {c \times d : \Middle{A \times B}{C \times D}}
  \qquad
  \inference
      {c : \Cast{A}{C} & d : \Cast{B}{D}}
      {c + d : \Middle{A + B}{C + D}}
  \\[1ex]
  \inference
      {}
      {\agda{id} : \Inj{A}{A}}
  \qquad
  \inference
      {}
      {G! : \Inj{G}{\Unk}}
  \qquad
  \inference
      {}
      {\cfail{\ell} : \Inj{A}{B}}
\end{gather*}
\fbox{$\coerce{A}{B}{\ell} = c,
  \coerce{A \otimes B}{A' \otimes B'}{\ell}_m = m$}
\begin{align*}
  \coerce{\Unk}{\Unk}{\ell} &=  \agda{id} \\
  \coerce{A}{\Unk}{\ell} &= \agda{id}; \coerce{A}{G}{\ell}_m; G!
    & \text{where } G = \agda{gnd}\app A \\
  \coerce{\Unk}{A}{\ell} &=  H?^\ell; \coerce{H}{A}{\ell}_m; \agda{id}
    & \text{where } H = \agda{gnd}\app A \\
  \coerce{b}{b}{\ell} &= \agda{id}; \agda{id}; \agda{id} \\
  \coerce{A \otimes B}{A' \otimes B'}{\ell}&=
      \agda{id}; \coerce{A \otimes B}{A' \otimes B'}{\ell}_m ; \agda{id} \\
  \coerce{A \to B}{A' \to B'}{\ell}_m & =
    \coerce{A'}{A}{\ell} \to \coerce{B}{B'}{\ell} \\
  \coerce{A \times B}{A' \times B'}{\ell}_m &=
      \coerce{A}{A'}{\ell} \times \coerce{B}{B'}{\ell} \\
  \coerce{A + B}{A' + B'}{\ell}_m &= 
      \coerce{A}{A'}{\ell} + \coerce{B}{B'}{\ell}
\end{align*}

\caption{Hypercoercions}
\label{fig:hypercoercions}
\end{figure}

The cast constructor is also defined in Figure~\ref{fig:hypercoercions}.

\subsubsection{Reduction Semantics and Type Safety}

A hypercoercion whose middle is a cast between function, pair, or sum
types, is a cross cast, provided the hypercoercion begins and ends
with the identity.\\[1ex]
\fbox{$\cross{c}$}
\begin{gather*}
  \inference
      {}
      {\key{Cross}{\otimes} : \cross{(\agda{id} ; c \otimes d ; \agda{id}})}
      \otimes \in \{ \to, \times, + \}
\end{gather*}

A hypercoercion that begins with identity and ends with
an injection to $\Unk$ is inert.
A hypercoercion that begins and ends with identity, but whose middle
is a cast between function types, is also insert. \\[1ex]
\fbox{$\inert{c}$}
\begin{gather*}
  \inference
      {}
      {\inert{(\agda{id} ; m ; G!)}}
  \qquad
  \inference
      {}
      {\inert{(\agda{id} ; c \to d; \agda{id})}}
\end{gather*}

There are four kinds of active hypercoercions.  The identity
hypercoercion is active, as is a hypercoercion that begins with a
projection from $\Unk$ or ends with a failure coercion.  Furthermore,
if the hypercoercion begins and ends with identity, and the middle is
either the identity or a cast between pair or sum types, then it is
active. \\[1ex]
\fbox{$\act{c}$}
\begin{gather*}
  \inference
      {}
      {\key{Aid} : \act{ \agda{id} } }
  \quad
  \inference
      {}
      {\key{Aproj} :  \act{(H?^\ell ; m ; i)}}
  \quad
  \inference
      {}
      {\key{Afail} : \act{(p ; m ; \cfail{\ell})}}
  \\[1ex]
  \inference
      {\otimes \in \{ \times, + \}}
      {\key{A}\otimes : \act{(\agda{id} ; c \otimes d ; \agda{id})}} 
  \quad
  \inference
      {}
      {\key{Abase} : \act{(\agda{id} ; \agda{id} ; \agda{id})}} 
\end{gather*}

\begin{lemma}
  \label{lem:hypercoercion-active-or-inert}
  For any types $A$ and $B$, $c : \Cast{A}{B}$ is either an active or
  inert cast.
\end{lemma}

\begin{lemma}\label{lem:hypercoercion-inert-cross}
  If $c : \Cast{A}{(B \otimes C)}$ and $\inert{c}$,
  then $\cross{c}$ and $A \equiv D \otimes E$ for some $D$ and $E$.
\end{lemma}

The definition of $\agda{dom}$, etc. for hypercoercions is given below.
\begin{align*}
  \dom{(\agda{id} ; c \to d ; \agda{id}) \app \key{Cross}{\to}} &= c \\
  \cod{(\agda{id} ; c \to d ; \agda{id}) \app \key{Cross}{\to}} &= d \\
  \agda{fst}\app (\agda{id}; c \times d; \agda{id}) \app \key{Cross}{\times} &= c \\
  \agda{snd} \app (\agda{id}; c \times d; \agda{id}) \app \key{Cross}{\times} &= d \\
  \agda{inl} \app (\agda{id} ; c + d; \agda{id}) \app \key{Cross}{+} &= c \\
  \agda{inr} \app (\agda{id} ; c + d; \agda{id}) \app \key{Cross}{+} &= d
\end{align*}

\begin{lemma}\label{lem:hypercoercion-base-not-inert}
   A cast $c : \Cast{A}{b}$ is not inert.
\end{lemma}

\begin{proposition}\label{prop:hypercoercion-active-pre-cast-struct}
Hypercoercions are an instance of the $\agda{PreCastStruct}$
structure.
\end{proposition}

To support space efficiency, we define the composition operator on
hypercoercions in Figure~\ref{fig:hyper-composition}. It is a mutually
recursive with the definition of composition on the middle
parts. Thankfully, Agda's termination checker approves of this
definition even though the contravariance in function coercions means
that it is not technically structurally recursive.

\begin{figure}[tbp]
\fbox{$c \fatsemi d$} \fbox{$m_1 \fatsemi m_2$}
\begin{align*}
  c \fatsemi \agda{id} &= c \\
  \agda{id} \fatsemi (p;m;i)&= p;m;i \\
  (p_1;m_1;\agda{id}) \fatsemi (\agda{id};m_2;i_2) &= p; (m_1 \fatsemi m_2);i_2\\
  (p_1;m_1;G!)\fatsemi(G?^\ell;m_2;i_2)&= p_1; (m_1\fatsemi m_2); i_2\\
  (p_1;m_1;G!)\fatsemi(H?^\ell;m_2;i_2)&=
      p_1; m_1; \cfail{\ell} & \text{if } G \neq H \\
  (p_1; m_1; \cfail{\ell}) \fatsemi (p_2; m_2; i_2) &= p_1; m_1; \cfail{\ell}\\
\\
\agda{id} \fatsemi \agda{id} &= \agda{id} \\
(c_1 \to d_1) \fatsemi (c_2  \to d_2) &= (c_2 \fatsemi c_1) \to (d_1 \fatsemi d_2) \\
(c_1 \times d_1) \fatsemi (c_2  \times d_2) &= (c_1 \fatsemi c_2) \times (d_1 \fatsemi d_2) \\
(c_1 + d_1) \fatsemi (c_2  + d_2) &= (c_1 \fatsemi c_2) + (d_1 \fatsemi d_2) 
\end{align*}
\caption{Composition of Hypercoercions}
\label{fig:hyper-composition}
\end{figure}

We define the \agda{applyCast} function for hypercoercions by cases on
$\act{c}$.
\[
  \agda{applyCast} : \forall \Gamma A B.\, (M : \Gamma \vdash A) \to \agda{SimpleValue}\app M\to (c : \Cast{A}{B}) \to \agda{Active}\app c \to \Gamma \vdash B
\]
\[
\begin{array}{lllllll}
  \agda{applyCast} & M & v & \agda{id} & \key{Aid} &=& M \\
  \agda{applyCast} & \cast{M}{c} & v & (H?^\ell; m; i) & \key{Aproj} &=&
      \cast{M}{c \fatsemi (H?^\ell; m; i)}\\
  \agda{applyCast} & M & v & (\agda{id};m;\cfail{\ell}) & \key{Afail} &=& \blame{\ell}{} \\
  \agda{applyCast} & (\key{cons}\app V_1 \app V_2) & v & (\agda{id};c \times d;\agda{id})  & \key{A}{\times} &=&
     \key{cons}\app (\cast{V_1}{c}) \app (\cast{V_2}{d}) \\
  \agda{applyCast} & (\key{inl}\app V) & v & (\agda{id};c + d;\agda{id}) & \key{A}{+} &=&
     \key{inl}\app (\cast{V}{c})\\
  \agda{applyCast} & (\key{inr}\app V) & v & (\agda{id};c + d;\agda{id}) & \key{A}{+} &=&
     \key{inr}\app (\cast{V}{d}) \\
  \agda{applyCast} & M & v & (\agda{id}; \agda{id}; \agda{id}) & \key{Abase} &=& M
\end{array}
\]

\begin{proposition}
Hypercoercions are an instance of the $\agda{ComposableCasts}$ structure.
\end{proposition}

We import and instantiate the reduction semantics and proof of type
safety from Section~\ref{sec:EfficientParamCasts}.

\begin{definition}[Reduction]
  \label{def:hyper-reduction}
  The reduction relation $\itm{ctx} \vdash M \longrightarrow_H N$ of
  $\lambda\textup{\textsf{H}}$ is the reduction relation of
  $\SC(\CastOp)$ instantiated with $\lambda\textup{\textsf{H}}$'s
  instance of the $\agda{ComposableCasts}$ structure.
\end{definition}

\begin{corollary}[Preservation for $\lambda\textup{\textsf{H}}$]
  If $\Gamma \vdash M : A$  and $M \longrightarrow_H M'$,\\
  then $\Gamma \vdash M' : A$.
\end{corollary}

\begin{corollary}[Progress for $\lambda\textup{\textsf{H}}$]
  If $\emptyset \vdash M : A$, then 
  \begin{enumerate}
  \item $M \longrightarrow_H M'$ for some $M'$,
  \item $\val{M}$, or
  \item $M \equiv \blame{\ell}{}$.
  \end{enumerate}
\end{corollary}

\subsubsection{Space Efficiency}

Next we establish that $\lambda\textup{\textsf{H}}$ is an instance of
the $\agda{CastHeight}$ structure so that we can apply
Theorem~\ref{thm:space-consumption} (Space Consumption) to obtain
space efficiency for $\lambda\textup{\textsf{H}}$.

We define the height of a hypercoercion to be the height of its middle
part.
\begin{align*}
  \agda{height}(\agda{id}) &= 0 \\
  \agda{height}(p;m;i) &=  \agda{height}(m)\\
  \agda{height}(c \to d) &= 1 + \max(\agda{height}(c), \agda{height}(d))  \\
  \agda{height}(c \times d) &= 1 + \max(\agda{height}(c), \agda{height}(d))  \\
  \agda{height}(c + d) &= 1 + \max(\agda{height}(c), \agda{height}(d))  
\end{align*}

The size of a hypercoercion is given by the following definition.
\begin{align*}
  \agda{size}(\agda{id}) &= 0 \\
  \agda{size}(\cfail{\ell}) &= 0 \\
  \agda{size}(G!)&= 1\\
  \agda{size}(G?^\ell)&= 1\\
  \agda{size}(p; m; i) &= 2 + \agda{size}(p) + \agda{size}(m) + \agda{size}(i)\\
  \agda{size}(c \to d) &= 1 + \agda{size}(c) + \agda{size}(d) \\
  \agda{size}(c \times d) &= 1 + \agda{size}(c) + \agda{size}(d) \\
  \agda{size}(c + d) &= 1 + \agda{size}(c) + \agda{size}(d) 
\end{align*}

The cast height of the result of \agda{applyCast} applied to a simple
value $S$ and coercion $c$ is less than the max of the cast height of
$S$ and the height of $c$.
\begin{lemma}
  \label{lem:applyCast-height-hyper}
  $\agda{c\mhyphen{}height}(\agda{applyCast}\app S\app c) \leq \max(\agda{height}(S), \agda{height}(c))$
\end{lemma}

The \agda{dom}, \agda{cod}, \agda{fst}, \agda{snd}, \agda{inl}, and
\agda{inr} operators on coercions all return coercions of equal or
lesser height than their input.

\begin{lemma}\ 
  \label{lem:dom-height-hyper}
  \begin{enumerate}
  \item $\agda{height}(\agda{dom}\app c \app x) \leq \agda{height}(c)$
  \item $\agda{height}(\agda{cod}\app c \app x)) \leq \agda{height}(c)$
  \item $\agda{height}(\agda{fst}\app c \app x) \leq \agda{height}(c)$
  \item $\agda{height}(\agda{snd}\app c \app x) \leq \agda{height}(c)$
  \item $\agda{height}(\agda{inl}\app c \app x) \leq \agda{height}(c)$
  \item $\agda{height}(\agda{inr}\app c \app x) \leq \agda{height}(c)$
  \end{enumerate}
\end{lemma}

The size of a hypercoercion $c$ is bounded by $9 \cdot
2^{\agda{height}(c)}$. We prove this by simultaneously proving the
following three facts about the four kinds of hypercoercions.

\begin{lemma}\ 
  \label{lem:size-height-hyper}
  \begin{enumerate}
  \item $\agda{size}(c) + 5 \le 9 \cdot 2^{\agda{height}(c)}$
  \item $\agda{size}(p) \le 1$
  \item $\agda{size}(i) \le 1$
  \item $\agda{size}(m) + 9 \le 9 \cdot 2^{\agda{height}(m)}$
  \end{enumerate}
\end{lemma}

The above lemmas and definitions establish the following.

\begin{proposition}
$\lambda\textup{\textsf{H}}$ is an instance of the
  $\agda{CastHeight}$ structure.
\end{proposition}

We apply Theorem~\ref{thm:space-consumption-lambdaS} (Space
Consumption) to obtain the following result for
$\lambda\textup{\textsf{H}}$.

\begin{corollary}[Space Consumption for $\lambda\textup{\textsf{H}}$]
  \label{thm:space-consumption-hyper}
  If ~$\Gamma \vdash M : A$, then there exists $c$ such that for any
  $M' : \Gamma \vdash A$ where $\itm{ctx} \vdash \compile{\Gamma}{M}
  \longrightarrow_H^{*} M'$, we have $\agda{real\mhyphen{}size}(M') \leq
  c \cdot \idealsize(M')$.
\end{corollary}

\section{Conclusion}
\label{sec:conclusion}

In this paper we present two parameterized cast calculi,
$\CC(\CastOp)$ and its space-efficient partner $\SC(\CastOp)$.  We
prove type safety, blame safety, and the gradual guarantee for the
former. We prove type safety and space-efficiency for the later.  We
instantiate $\CC(\CastOp)$ a half-dozen ways to reproduce some results
from the literature but to also fill in many gaps.  We instantiate
$\SC(\CastOp)$ two different ways to reproduce
$\lambda\textup{\textsf{S}}$ and to create a new space-efficient
calculus based hypercoercions. All of this is formalized in Agda.



\section{Conflicts of Interest}

The authors are employed at Indiana University.

\pagebreak

\bibliographystyle{abbrvnat} \bibliography{all}

\begin{thebibliography}{67}
\providecommand{\natexlab}[1]{#1}
\providecommand{\url}[1]{\texttt{#1}}
\expandafter\ifx\csname urlstyle\endcsname\relax
  \providecommand{\doi}[1]{doi: #1}\else
  \providecommand{\doi}{doi: \begingroup \urlstyle{rm}\Url}\fi

\bibitem[Ahmed et~al.(2011)Ahmed, Findler, Siek, and Wadler]{Ahmed:2011fk}
A.~Ahmed, R.~B. Findler, J.~G. Siek, and P.~Wadler.
\newblock Blame for {All}.
\newblock In \emph{Symposium on Principles of Programming Languages}, January
  2011.

\bibitem[Ahmed et~al.(2017)Ahmed, Jamner, Siek, and Wadler]{Ahmed:2017aa}
A.~Ahmed, D.~Jamner, J.~G. Siek, and P.~Wadler.
\newblock Theorems for free for free: Parametricity, with and without types.
\newblock In \emph{International Conference on Functional Programming}, ICFP,
  September 2017.

\bibitem[Allende et~al.(2013)Allende, Fabry, and Tanter]{Allende:2013ab}
E.~Allende, J.~Fabry, and E.~Tanter.
\newblock Cast insertion strategies for gradually-typed objects.
\newblock In \emph{Proceedings of the 9th Symposium on Dynamic Languages}, DLS
  '13, pages 27--36, New York, NY, USA, 2013. ACM.

\bibitem[Aydemir et~al.(2005)Aydemir, Bohannon, Fairbairn, Foster, Pierce,
  Sewell, Vytiniotis, Weirich, and Zdancewic]{Aydemir:2005fk}
B.~E. Aydemir, A.~Bohannon, M.~Fairbairn, J.~N. Foster, B.~C. Pierce,
  P.~Sewell, D.~Vytiniotis, G.~W.~S. Weirich, and S.~Zdancewic.
\newblock Mechanized metatheory for the masses: The {POPLmark} challenge.
\newblock May 2005.

\bibitem[Ba\~{n}ados Schwerter et~al.(2021)Ba\~{n}ados Schwerter, Clark,
  Jafery, and Garcia]{Banados-Schwerter:2021aa}
F.~Ba\~{n}ados Schwerter, A.~M. Clark, K.~A. Jafery, and R.~Garcia.
\newblock Abstracting gradual typing moving forward: Precise and
  space-efficient.
\newblock \emph{Proc. ACM Program. Lang.}, 5\penalty0 (POPL), Jan. 2021.
\newblock \doi{10.1145/3434342}.
\newblock URL \url{https://doi.org/10.1145/3434342}.

\bibitem[Bierman et~al.(2014)Bierman, Abadi, and Torgersen]{Bierman:2014aa}
G.~Bierman, M.~Abadi, and M.~Torgersen.
\newblock Understanding {TypeScript}.
\newblock In R.~Jones, editor, \emph{ECOOP 2014 -- Object-Oriented
  Programming}, volume 8586 of \emph{Lecture Notes in Computer Science}, pages
  257--281. Springer Berlin Heidelberg, 2014.

\bibitem[Bove et~al.(2009)Bove, Dybjer, and Norell]{Bove:2009aa}
A.~Bove, P.~Dybjer, and U.~Norell.
\newblock A brief overview of agda --- a functional language with dependent
  types.
\newblock In \emph{Proceedings of the 22Nd International Conference on Theorem
  Proving in Higher Order Logics}, TPHOLs '09, pages 73--78, Berlin,
  Heidelberg, 2009. Springer-Verlag.

\bibitem[Bracha(2004)]{Bracha:2004wa}
G.~Bracha.
\newblock Pluggable type systems.
\newblock In \emph{OOPSLA'04 Workshop on Revival of Dynamic Languages}, 2004.

\bibitem[Bracha and Griswold(1993)]{Bracha:1993sn}
G.~Bracha and D.~Griswold.
\newblock Strongtalk: typechecking {Smalltalk} in a production environment.
\newblock In \emph{OOPSLA '93: Proceedings of the eighth annual conference on
  Object-oriented programming systems, languages, and applications}, pages
  215--230, New York, NY, USA, 1993. ACM Press.
\newblock ISBN 0-89791-587-9.

\bibitem[Castagna and Lanvin(2017)]{Castagna:2017aa}
G.~Castagna and V.~Lanvin.
\newblock Gradual typing with union and intersection types.
\newblock In \emph{International Conference on Functional Programming}, 2017.

\bibitem[Castagna et~al.(2019)Castagna, Lanvin, Petrucciani, and
  Siek]{Castagna:2019aa}
G.~Castagna, V.~Lanvin, T.~Petrucciani, and J.~G. Siek.
\newblock Gradual typing: A new perspective.
\newblock \emph{Proc. ACM Program. Lang.}, 3\penalty0 (POPL):\penalty0
  16:1--16:32, Jan. 2019.
\newblock ISSN 2475-1421.
\newblock \doi{10.1145/3290329}.
\newblock URL \url{http://doi.acm.org/10.1145/3290329}.

\bibitem[Chaudhuri()]{Chaudhuri:2014aa}
A.~Chaudhuri.
\newblock Flow: a static type checker for javascript.
\newblock URL \url{http://flowtype.org/}.

\bibitem[Chung et~al.(2018)Chung, Li, Nardelli, and Vitek]{Chung:2018aa}
B.~Chung, P.~Li, F.~Z. Nardelli, and J.~Vitek.
\newblock {KafKa: Gradual Typing for Objects}.
\newblock In T.~Millstein, editor, \emph{32nd European Conference on
  Object-Oriented Programming (ECOOP 2018)}, volume 109 of \emph{Leibniz
  International Proceedings in Informatics (LIPIcs)}, pages 12:1--12:24,
  Dagstuhl, Germany, 2018. Schloss Dagstuhl--Leibniz-Zentrum fuer Informatik.
\newblock ISBN 978-3-95977-079-8.
\newblock \doi{10.4230/LIPIcs.ECOOP.2018.12}.
\newblock URL \url{http://drops.dagstuhl.de/opus/volltexte/2018/9217}.

\bibitem[Eremondi et~al.(2019)Eremondi, Tanter, and Garcia]{Eremondi:2019aa}
J.~Eremondi, E.~Tanter, and R.~Garcia.
\newblock Approximate normalization for gradual dependent types.
\newblock \emph{Proc. ACM Program. Lang.}, 3\penalty0 (ICFP), July 2019.
\newblock \doi{10.1145/3341692}.
\newblock URL \url{https://doi.org/10.1145/3341692}.

\bibitem[Felleisen and Friedman(1986)]{Felleisen:kx}
M.~Felleisen and D.~P. Friedman.
\newblock Control operators, the {SECD}-machine and the lambda-calculus.
\newblock pages 193--217, 1986.

\bibitem[Flanagan(2006)]{Flanagan:2006mn}
C.~Flanagan.
\newblock Hybrid type checking.
\newblock In \emph{{POPL} 2006: The 33rd {ACM SIGPLAN-SIGACT} Symposium on
  Principles of Programming Languages}, pages 245--256, Charleston, South
  Carolina, January 2006.

\bibitem[Garcia(2013)]{Garcia:2013fk}
R.~Garcia.
\newblock Calculating threesomes, with blame.
\newblock In \emph{{ICFP} '13: Proceedings of the International Conference on
  Functional Programming}, 2013.

\bibitem[Garcia and Cimini(2015)]{Garcia:2015aa}
R.~Garcia and M.~Cimini.
\newblock Principal type schemes for gradual programs.
\newblock In \emph{Proceedings of the 42nd Annual ACM SIGPLAN-SIGACT Symposium
  on Principles of Programming Languages}, POPL '15, pages 303--315. ACM, 2015.

\bibitem[Garcia et~al.(2016)Garcia, Clark, and Tanter]{Garcia:2016aa}
R.~Garcia, A.~M. Clark, and E.~Tanter.
\newblock Abstracting gradual typing.
\newblock In \emph{Proceedings of the 43rd Annual ACM SIGPLAN-SIGACT Symposium
  on Principles of Programming Languages}, POPL 2016, pages 429--442, New York,
  NY, USA, 2016. ACM.
\newblock ISBN 978-1-4503-3549-2.
\newblock \doi{10.1145/2837614.2837670}.
\newblock URL \url{http://doi.acm.org/10.1145/2837614.2837670}.

\bibitem[Garc\'{\i}a-P{\'e}rez et~al.(2014)Garc\'{\i}a-P{\'e}rez, Nogueira, and
  Sergey]{Garcia-Perez:2014aa}
A.~Garc\'{\i}a-P{\'e}rez, P.~Nogueira, and I.~Sergey.
\newblock Deriving interpretations of the gradually-typed lambda calculus.
\newblock In \emph{Proceedings of the ACM SIGPLAN 2014 Workshop on Partial
  Evaluation and Program Manipulation}, PEPM '14, pages 157--168, New York, NY,
  USA, 2014. ACM.
\newblock ISBN 978-1-4503-2619-3.
\newblock \doi{10.1145/2543728.2543742}.
\newblock URL \url{http://doi.acm.org/10.1145/2543728.2543742}.

\bibitem[Greenman(2020)]{Greenman:2020aa}
B.~Greenman.
\newblock \emph{Deep and Shallow Types}.
\newblock PhD thesis, Northeastern University, November 2020.

\bibitem[Greenman and Felleisen(2018)]{Greenman:2018aa}
B.~Greenman and M.~Felleisen.
\newblock A spectrum of type soundness and performance.
\newblock \emph{Proc. ACM Program. Lang.}, 2\penalty0 (ICFP):\penalty0
  71:1--71:32, July 2018.
\newblock ISSN 2475-1421.
\newblock \doi{10.1145/3236766}.
\newblock URL \url{http://doi.acm.org/10.1145/3236766}.

\bibitem[Greenman and Migeed(2018)]{Greenman:2018ab}
B.~Greenman and Z.~Migeed.
\newblock On the cost of type-tag soundness.
\newblock In \emph{Proceedings of the ACM SIGPLAN Workshop on Partial
  Evaluation and Program Manipulation}, PEPM '18, pages 30--39, New York, NY,
  USA, 2018. ACM.
\newblock ISBN 978-1-4503-5587-2.
\newblock \doi{10.1145/3162066}.
\newblock URL \url{http://doi.acm.org/10.1145/3162066}.

\bibitem[Gronski et~al.(2006)Gronski, Knowles, Tomb, Freund, and
  Flanagan]{Gronski:2006uq}
J.~Gronski, K.~Knowles, A.~Tomb, S.~N. Freund, and C.~Flanagan.
\newblock Sage: Hybrid checking for flexible specifications.
\newblock In \emph{Scheme and Functional Programming Workshop}, pages 93--104,
  2006.

\bibitem[Henglein(1994)]{Henglein:1994nz}
F.~Henglein.
\newblock Dynamic typing: syntax and proof theory.
\newblock \emph{Science of Computer Programming}, 22\penalty0 (3):\penalty0
  197--230, June 1994.

\bibitem[Herman et~al.(2007)Herman, Tomb, and Flanagan]{Herman:2006uq}
D.~Herman, A.~Tomb, and C.~Flanagan.
\newblock Space-efficient gradual typing.
\newblock In \emph{Trends in Functional Prog. (TFP)}, page XXVIII, April 2007.

\bibitem[Herman et~al.(2010)Herman, Tomb, and Flanagan]{Herman:2010aa}
D.~Herman, A.~Tomb, and C.~Flanagan.
\newblock Space-efficient gradual typing.
\newblock \emph{Higher-Order and Symbolic Computation}, 23\penalty0
  (2):\penalty0 167--189, 2010.

\bibitem[Howard(1980)]{Howard:1980aa}
W.~A. Howard.
\newblock \emph{The formulae-as-types notion of construction}.
\newblock Academic Press, 1980.

\bibitem[Kuhlenschmidt et~al.(2019)Kuhlenschmidt, Almahallawi, and
  Siek]{Kuhlenschmidt:2019aa}
A.~Kuhlenschmidt, D.~Almahallawi, and J.~G. Siek.
\newblock Toward efficient gradual typing for structural types via coercions.
\newblock In \emph{Conference on Programming Language Design and
  Implementation}, PLDI. ACM, June 2019.

\bibitem[Lennon-Bertrand et~al.(2020)Lennon-Bertrand, Maillard, Tabareau, and
  Tanter]{Lennon-Bertrand:2020to}
M.~Lennon-Bertrand, K.~Maillard, N.~Tabareau, and {\'E}.~Tanter.
\newblock Gradualizing the calculus of inductive constructions, 2020.

\bibitem[Lu(2020)]{Lu:2020aa}
K.-C. Lu.
\newblock Equivalence of cast representations in gradual typing.
\newblock Master's thesis, Indiana University, April 2020.

\bibitem[Lu et~al.(2020)Lu, Siek, and Kuhlenschmidt]{Lu:2019aa}
K.-C. Lu, J.~G. Siek, and A.~Kuhlenschmidt.
\newblock Hypercoercions and a framework for equivalence of cast calculi.
\newblock In \emph{Workshop on Gradual Typing}, 2020.

\bibitem[Maidl et~al.(2014)Maidl, Mascarenhas, and Ierusalimschy]{Maidl:2014aa}
A.~M. Maidl, F.~Mascarenhas, and R.~Ierusalimschy.
\newblock Typed lua: An optional type system for lua.
\newblock In \emph{Proceedings of the Workshop on Dynamic Languages and
  Applications}, Dyla'14, pages 3:1--3:10, New York, NY, USA, 2014. ACM.

\bibitem[Matthews and Findler(2007)]{Matthews:2007zr}
J.~Matthews and R.~B. Findler.
\newblock Operational semantics for multi-language programs.
\newblock In \emph{The 34th ACM SIGPLAN-SIGACT Symposium on Principles of
  Programming Languages}, January 2007.

\bibitem[Muehlboeck and Tate(2017)]{Muehlboeck:2017aa}
F.~Muehlboeck and R.~Tate.
\newblock Sound gradual typing is nominally alive and well.
\newblock \emph{Proc. ACM Program. Lang.}, 1\penalty0 (OOPSLA):\penalty0
  56:1--56:30, Oct. 2017.
\newblock ISSN 2475-1421.
\newblock \doi{10.1145/3133880}.
\newblock URL \url{http://doi.acm.org/10.1145/3133880}.

\bibitem[New et~al.(2019)New, Jamner, and Ahmed]{New:2019ab}
M.~S. New, D.~Jamner, and A.~Ahmed.
\newblock Graduality and parametricity: Together again for the first time.
\newblock \emph{Proc. ACM Program. Lang.}, 4\penalty0 (POPL), Dec. 2019.
\newblock \doi{10.1145/3371114}.
\newblock URL \url{https://doi.org/10.1145/3371114}.

\bibitem[Nipkow et~al.(2007)Nipkow, Paulson, and Wenzel]{Nipkow:2002jl}
T.~Nipkow, L.~C. Paulson, and M.~Wenzel.
\newblock \emph{Isabelle/HOL --- A Proof Assistant for Higher-Order Logic},
  volume 2283 of \emph{LNCS}.
\newblock Springer, November 2007.

\bibitem[Siek and Garcia(2012)]{Siek:2012uq}
J.~G. Siek and R.~Garcia.
\newblock Interpretations of the gradually-typed lambda calculus.
\newblock In \emph{Scheme and Functional Programming Workshop}, 2012.

\bibitem[Siek and Taha(2006{\natexlab{a}})]{Siek:2006bh}
J.~G. Siek and W.~Taha.
\newblock Gradual typing for functional languages.
\newblock In \emph{Scheme and Functional Programming Workshop}, pages 81--92,
  September 2006{\natexlab{a}}.

\bibitem[Siek and Taha(2006{\natexlab{b}})]{Siek:2006fk}
J.~G. Siek and W.~Taha.
\newblock Gradual typing for objects: Isabelle formaliztaion.
\newblock Technical Report CU-CS-1021-06, University of Colorado, Boulder, CO,
  December 2006{\natexlab{b}}.

\bibitem[Siek and Taha(2006{\natexlab{c}})]{Siek:2006hh}
J.~G. Siek and W.~Taha.
\newblock Gradual typing: Isabelle/isar formalization.
\newblock Technical Report TR06-874, Rice University, Houston, Texas,
  2006{\natexlab{c}}.

\bibitem[Siek and Taha(2007)]{Siek:2007qy}
J.~G. Siek and W.~Taha.
\newblock Gradual typing for objects.
\newblock In \emph{European {C}onference on {O}bject-{O}riented {P}rogramming},
  volume 4609 of \emph{LCNS}, pages 2--27, August 2007.

\bibitem[Siek and Vachharajani(2008)]{Siek:2008sf}
J.~G. Siek and M.~Vachharajani.
\newblock Gradual typing and unification-based inference.
\newblock In \emph{DLS}, 2008.

\bibitem[Siek and Vitousek(2013)]{Siek:2013aa}
J.~G. Siek and M.~M. Vitousek.
\newblock Monotonic references for gradual typing.
\newblock \emph{CoRR}, abs/1312.0694, 2013.

\bibitem[Siek and Wadler(2010)]{Siek:2010ya}
J.~G. Siek and P.~Wadler.
\newblock Threesomes, with and without blame.
\newblock In \emph{Symposium on {P}rinciples of {P}rogramming {L}anguages},
  POPL, pages 365--376, January 2010.

\bibitem[Siek et~al.(2009)Siek, Garcia, and Taha]{Siek:2009rt}
J.~G. Siek, R.~Garcia, and W.~Taha.
\newblock Exploring the design space of higher-order casts.
\newblock In \emph{European Symposium on Programming}, ESOP, pages 17--31,
  March 2009.

\bibitem[Siek et~al.(2015{\natexlab{a}})Siek, Thiemann, and
  Wadler]{Siek:2015ab}
J.~G. Siek, P.~Thiemann, and P.~Wadler.
\newblock Blame and coercion: Together again for the first time.
\newblock In \emph{Conference on Programming Language Design and
  Implementation}, PLDI, June 2015{\natexlab{a}}.

\bibitem[Siek et~al.(2015{\natexlab{b}})Siek, Vitousek, Cimini, and
  Boyland]{Siek:2015ac}
J.~G. Siek, M.~M. Vitousek, M.~Cimini, and J.~T. Boyland.
\newblock Refined criteria for gradual typing.
\newblock In \emph{SNAPL: Summit on Advances in Programming Languages},
  {LIPIcs}: {Leibniz} International Proceedings in Informatics, May
  2015{\natexlab{b}}.

\bibitem[Siek et~al.(2015{\natexlab{c}})Siek, Vitousek, Cimini,
  Tobin-Hochstadt, and Garcia]{Siek:2015aa}
J.~G. Siek, M.~M. Vitousek, M.~Cimini, S.~Tobin-Hochstadt, and R.~Garcia.
\newblock Monotonic references for efficient gradual typing.
\newblock In \emph{European Symposium on Programming}, ESOP, April
  2015{\natexlab{c}}.

\bibitem[Takikawa(2016)]{Takikawa:2016ab}
A.~Takikawa.
\newblock \emph{The Design, Implementation, And Evaluation Of A Gradual Type
  System For Dynamic Class Composition}.
\newblock PhD thesis, Northeastern University, April 2016.

\bibitem[Takikawa et~al.(2012)Takikawa, Strickland, Dimoulas, Tobin-Hochstadt,
  and Felleisen]{Takikawa:2012ly}
A.~Takikawa, T.~S. Strickland, C.~Dimoulas, S.~Tobin-Hochstadt, and
  M.~Felleisen.
\newblock Gradual typing for first-class classes.
\newblock In \emph{{C}onference on {O}bject {O}riented {P}rogramming {S}ystems
  {L}anguages and {A}pplications}, OOPSLA '12, pages 793--810, 2012.

\bibitem[Takikawa et~al.(2016)Takikawa, Feltey, Greenman, New, Vitek, and
  Felleisen]{Takikawa:2016aa}
A.~Takikawa, D.~Feltey, B.~Greenman, M.~New, J.~Vitek, and M.~Felleisen.
\newblock Is sound gradual typing dead?
\newblock In \emph{Principles of Programming Languages}, POPL. ACM, January
  2016.

\bibitem[{The {Coq} Dev. Team}(2004)]{The-Coq-Development-Team:2004kf}
{The {Coq} Dev. Team}.
\newblock \emph{{The Coq Proof Assistant Reference Manual -- Version V8.0}},
  Apr. 2004.
\newblock \url{http://coq.inria.fr}.

\bibitem[Tobin-Hochstadt and Felleisen(2006)]{Tobin-Hochstadt:2006fk}
S.~Tobin-Hochstadt and M.~Felleisen.
\newblock Interlanguage migration: From scripts to programs.
\newblock In \emph{Dynamic Languages Symposium}, 2006.

\bibitem[Tobin-Hochstadt and Felleisen(2008)]{Tobin-Hochstadt:2008lr}
S.~Tobin-Hochstadt and M.~Felleisen.
\newblock The design and implementation of {Typed} {Scheme}.
\newblock In \emph{Symposium on {P}rinciples of {P}rogramming {L}anguages},
  January 2008.

\bibitem[Toro and Tanter(2020)]{Toro:2020aa}
M.~Toro and {\'E}.~Tanter.
\newblock Abstracting gradual references.
\newblock \emph{Science of Computer Programming}, 197:\penalty0 102496, 2020.
\newblock ISSN 0167-6423.
\newblock \doi{https://doi.org/10.1016/j.scico.2020.102496}.
\newblock URL
  \url{http://www.sciencedirect.com/science/article/pii/S0167642320301052}.

\bibitem[Toro et~al.(2019)Toro, Labrada, and Tanter]{Toro:2019aa}
M.~Toro, E.~Labrada, and E.~Tanter.
\newblock Gradual parametricity, revisited.
\newblock \emph{Proc. ACM Program. Lang.}, 3\penalty0 (POPL):\penalty0
  17:1--17:30, Jan. 2019.
\newblock ISSN 2475-1421.
\newblock \doi{10.1145/3290330}.
\newblock URL \url{http://doi.acm.org/10.1145/3290330}.

\bibitem[Verlaguet and Menghrajani()]{Verlaguet:aa}
J.~Verlaguet and A.~Menghrajani.
\newblock Hack: a new programming langauge for {HHVM}.
\newblock URL
  \url{https://code.facebook.com/posts/264544830379293/hack-a-new-programming-language-for-hhvm/}.

\bibitem[Vitek(2016)]{Vitek:2016aa}
J.~Vitek.
\newblock Gradual types for real-world objects.
\newblock In \emph{Script To Program Evolution Workshop}, STOP, 2016.

\bibitem[Vitousek and Siek(2016{\natexlab{a}})]{Vitousek:2016ab}
M.~Vitousek and J.~Siek.
\newblock From optional to gradual typing via transient checks.
\newblock In \emph{Script To Program Evolution Workshop}, STOP,
  2016{\natexlab{a}}.

\bibitem[Vitousek et~al.(2017)Vitousek, Swords, and Siek]{Vitousek:2017aa}
M.~Vitousek, C.~Swords, and J.~G. Siek.
\newblock Big types in little runtime.
\newblock In \emph{Symposium on {P}rinciples of {P}rogramming {L}anguages},
  POPL, 2017.

\bibitem[Vitousek and Siek(2016{\natexlab{b}})]{Vitousek:2016aa}
M.~M. Vitousek and J.~G. Siek.
\newblock Gradual typing in an open world.
\newblock Technical Report TR729, Indiana University, October
  2016{\natexlab{b}}.

\bibitem[Vitousek et~al.(2019)Vitousek, Siek, and Chaudhuri]{Vitousek:2019aa}
M.~M. Vitousek, J.~G. Siek, and A.~Chaudhuri.
\newblock Optimizing and evaluating transient gradual typing.
\newblock In \emph{Proceedings of the 15th ACM SIGPLAN International Symposium
  on Dynamic Languages}, DLS 2019, pages 28--41, New York, NY, USA, 2019.
  Association for Computing Machinery.
\newblock ISBN 9781450369961.
\newblock \doi{10.1145/3359619.3359742}.
\newblock URL \url{https://doi.org/10.1145/3359619.3359742}.

\bibitem[Wadler and Findler(2007)]{Wadler:2007lr}
P.~Wadler and R.~B. Findler.
\newblock Well-typed programs can't be blamed.
\newblock In \emph{Workshop on Scheme and Functional Programming}, pages
  15--26, 2007.

\bibitem[Wadler and Findler(2009)]{Wadler:2009qv}
P.~Wadler and R.~B. Findler.
\newblock Well-typed programs can't be blamed.
\newblock In \emph{European {S}ymposium on {P}rogramming}, ESOP, pages 1--16,
  March 2009.

\bibitem[Wadler and Kokke(2019)]{Wadler:2019aa}
P.~Wadler and W.~Kokke.
\newblock \emph{Programming Language Foundations in {A}gda}.
\newblock 2019.
\newblock Available at \url{http://plfa.inf.ed.ac.uk/}.

\bibitem[Xie et~al.(2018)Xie, Bi, and Oliveira]{Xie:2018aa}
N.~Xie, X.~Bi, and B.~C. d.~S. Oliveira.
\newblock Consistent subtyping for all.
\newblock In A.~Ahmed, editor, \emph{Programming Languages and Systems}, pages
  3--30, Cham, 2018. Springer International Publishing.
\newblock ISBN 978-3-319-89884-1.

\end{thebibliography}

\end{document}